\documentclass[final, 11pt]{article}

\usepackage{amsmath,amssymb,amsfonts}
\usepackage{mathtools}
\usepackage{epsfig}
\usepackage{latexsym,nicefrac,bbm}
\usepackage{xspace}
\usepackage{color,fancybox,graphicx,url,subfigure}
\usepackage{hyperref,enumitem}
\usepackage{fullpage}
\usepackage{mathabx}
\usepackage{tabularx}
\usepackage{mdframed}
\usepackage{longtable}
\usepackage{tabu}
\usepackage{stfloats}
\usepackage{framed}

\def\showauthornotes{0}

\def\showdraftbox{0}
\allowdisplaybreaks
\newcommand\np{\mbox{\bf NP}\xspace}

\newcommand\sat{\mbox{SAT}\xspace}

\newtheorem{theorem}{Theorem}[section]

\newtheorem{definition}[theorem]{Definition}
\newtheorem{lemma}[theorem]{Lemma}

\newtheorem{proposition}[theorem]{Proposition}
\newtheorem{corollary}[theorem]{Corollary}
\newtheorem{claim}[theorem]{Claim}

\def\FullBox{\hbox{\vrule width 6pt height 6pt depth 0pt}}

\def\qed{\ifmmode\qquad\FullBox\else{\unskip\nobreak\hfil
\penalty50\hskip1em\null\nobreak\hfil\FullBox
\parfillskip=0pt\finalhyphendemerits=0\endgraf}\fi}

\def\qedsketch{\ifmmode\Box\else{\unskip\nobreak\hfil
\penalty50\hskip1em\null\nobreak\hfil$\Box$
\parfillskip=0pt\finalhyphendemerits=0\endgraf}\fi}

\newenvironment{proof}{\begin{trivlist} \item {\bf Proof:~~}}
   {\qed\end{trivlist}}

\newcommand\rea{\mathbb R}

\newcommand\B{\{0,1\}}      % boolean alphabet  use in math mode

\newcommand\true{\mbox{\sc True}}
\newcommand\false{\mbox{\sc False}}

\newcommand{\marginlabel}[1]%
{\mbox{}\marginpar{\it{\raggedleft\hspace{0pt}#1}}}
\newcommand\poly{\mathrm{poly}}  

\newcommand{\ceil}[1]{\left\lceil\, {#1}\,\right\rceil}

\newcommand{\pair}[1]{\left\langle{#1}\right\rangle} %for inner product

\definecolor{Mygray}{gray}{0.8}

 \ifcsname ifcommentflag\endcsname\else
  \expandafter\let\csname ifcommentflag\expandafter\endcsname
                  \csname iffalse\endcsname
\fi

\ifnum\showauthornotes=1

\else

\fi

\ifnum\showauthornotes=1
\newcommand{\Authornote}[2]{{\sf\small\color{red}{[#1: #2]}}}
\newcommand{\Authoredit}[2]{{\sf\small\color{red}{[#1]}\color{blue}{#2}}}
\newcommand{\Authorcomment}[2]{{\sf \small\color{Mygray}{[#1: #2]}}}
\newcommand{\Authorfnote}[2]{\footnote{\color{red}{#1: #2}}}
\newcommand{\Authorfixme}[1]{\Authornote{#1}{\textbf{??}}}
\newcommand{\Authormarginmark}[1]{\marginpar{\textcolor{red}{\fbox{%\Large
#1:!}}}}
\else
\newcommand{\Authornote}[2]{}
\newcommand{\Authoredit}[2]{}
\newcommand{\Authorcomment}[2]{}
\newcommand{\Authorfnote}[2]{}
\newcommand{\Authorfixme}[1]{}
\newcommand{\Authormarginmark}[1]{}
\fi

\newcommand\calC{\mathcal{C}}
\newcommand\calD{\mathcal{D}}
\newcommand\calW{\mathcal{W}}

\newcommand\av{\mathop{\mbox{\bf E}}}
\newcommand\var
{
  \mathop{
    \mathchoice
    {\mbox{\bf Var}}
    {{\bf Var}}
    {\mbox{\bf \scriptsize Var}}
    {}
  }
}

\newcommand{\norm}[1]{\ensuremath{\left\lVert #1 \right\rVert}}

\newlength{\pgmtab}  %  \pgmtab is the width of each tab in the
 {
	\begin{enumerate}}{\end{enumerate}}

\newcounter{lecnum}

\newlength{\tpush}
\ifnum\showdraftbox=1
\newcommand{\draftbox}{\begin{center}
  \fbox{%
    \begin{minipage}{2in}%
      \begin{center}%
          \large\textsc{Working Draft}\\%
        Please do not distribute%
      \end{center}%
    \end{minipage}%
  }%
\end{center}
\vspace{0.2cm}}
\else
\newcommand{\draftbox}{}
\fi

\newcommand{\Ssnote}{\Authornote{SS}}

\newcommand{\Sknote}{\Authornote{SK}}

\newcommand{\defeq}{\stackrel{\textup{def}}{=}} 
\newcommand{\eps}{\varepsilon}
\renewcommand{\epsilon}{\varepsilon}
\newcommand{\nfrac}{\nicefrac}

\newcommand{\E}{\mathbf{E}}

\newcommand{\activedeg}{\mathsf{activedegree}}
\newcommand{\Var}{\mathbf{Var}}
\newcommand{\cnt}{\mathsf{count}}

\newcommand{\Act}{\mathsf{Active}}
\newcommand{\varest}{\mathsf{Uvar}}
\newcommand{\meanest}{\mathsf{Lmean}}
\newcommand{\val}{\mathsf{val}}

\newcommand{\etal}{{\em et al.\ }}

\newcommand\maxsat{\mbox{\sc Max-2-SAT}\xspace}
\newcommand\maxcut{\mbox{\sc Max-CUT}\xspace}
\newcommand\maxand{\mbox{\sc Max-2-AND}\xspace}
\newcommand\maxcsp{\mbox{\sc Max-CSP}\xspace}
\newcommand\maxcspwq{\mbox{\sc Max-$w$-CSP$_q$}\xspace}
\newcommand\maxconj{\mbox{\sc Max-$w$-ConjSAT$_q$} \xspace}
\newcommand\maxfcsp{\mbox{\sc Max-$\calF$-CSP} \xspace}
\newcommand\maxwsat{\mbox{\sc Max-$w$-SAT}\xspace}
\newcommand\opt{\textrm{\sc Opt}}

\newcommand{\varcalc}{\mathsf{TrueVar}}
\newcommand{\meancalc}{\mathsf{TrueMean}}

\newcommand{\finalW}{\mathsf{finalwt}}
\newcommand{\forward}{{\sf{forward}}}
\newcommand{\backward}{{\sf{backward}}}
\newcommand{\heaviest}{{\sf{heaviest}}}
\newcommand{\calT}{{\mathcal T}}

\newcommand{\Inst}{\mathcal{I}}
\newcommand{\calA}{{\mathcal A}}
\newcommand{\calB}{{\mathcal B}}

\newcommand{\NAE}{\sf{NAE}}
\newcommand{\typeA}{{\sf{typeA}}}
\newcommand{\typeB}{{\sf{typeB}}}
\newcommand{\typeAB}{{\sf{typeAB}}}
\newcommand{\typeC}{{\sf{typeC}}}
\newcommand{\out}{{\sf{out}}}

\newcommand{\Cgood}{{\sf{Cgood}}}
\newcommand{\score}{{\sf{score}}}

\newcommand{\calF}{{\mathcal{F}}}
\newcommand{\XOR}{{\mathsf{XOR}}}
\newcommand{\Dist}{{\mathsf{Dist}}}
\newcommand{\vect}{{\vec{t}}}
\newcommand{\vecz}{{\vec{z}}}
\newcommand{\smooth}{{\mathsf{smooth}}}
\newcommand{\LL}{{\mathsf{ConjSAT\mbox{-}LP}}}
\newcommand{\LLand}{{\mathsf{MAX2AND\mbox{-}LP}}}

\newcommand{\LLwsat}{{\mathsf{MAXwSAT\mbox{-}LP}}}
\newcommand{\flag}{\mathsf{flag}}

\newcommand{\algouwmaxcut}{{\sc Sim-UnweightedMC}}

\newcommand{\algowand}{{\sc Sim-Max2AND}}
\newcommand{\algowcsp}{{\sc Sim-MaxConjSAT}}
\newcommand{\algowwsat}{{\sc Sim-MaxwSAT}}

\newtheorem{observation}[theorem]{Observation}
\title{Simultaneous Approximation of Constraint Satisfaction Problems}
\author{Amey Bhangale
\thanks{Department of Computer Science. Rutgers University.  Research supported in part by NSF grant CCF-1253886. {\tt amey.bhangale@rutgers.edu}}
\and 
Swastik Kopparty \thanks{Department of Mathematics \& Department of
  Computer Science. Rutgers University. Research supported in part by a Sloan Fellowship and NSF grant CCF-1253886. {\tt swastik.kopparty@rutgers.edu}}
\and
Sushant Sachdeva \thanks{
Department of Computer Science, Yale University. Research supported by the
NSF grants CCF-0832797, CCF-1117309, and Daniel Spielman's \& Sanjeev Arora's Simons Investigator Grants. Part of this work was
    done when this author was at the Simons Institute for
    the Theory of Computing, UC Berkeley, and at the Department of
    Computer Science, Princeton University. Email: {\tt
      sachdeva@cs.yale.edu}}
}

\begin{document}
\maketitle

%\draftbox
\thispagestyle{empty}
\draftbox

\begin{abstract}
  Given $k$ collections of 2SAT 
  clauses on the same set of variables $V$, can we find one
  assignment that satisfies a large fraction of clauses from {\em
    each} collection? We consider such {\em simultaneous} constraint
  satisfaction problems, and design the first nontrivial approximation
  algorithms in this context.

  Our main result is that for every CSP $\calF$, for $k <
  \tilde{O}(\log^{\nfrac{1}{4}} n)$, there is a polynomial time
  constant factor {\em Pareto} approximation algorithm for $k$
  simultaneous \maxfcsp instances. Our methods are quite general, and
  we also use them to give an improved approximation factor for
  simultaneous {\sc Max-$w$-SAT} (for $k
  <\tilde{O}(\log^{\nfrac{1}{3}} n)$).  In contrast, for $k =
  \omega(\log n)$, no nonzero approximation factor for $k$
  simultaneous \maxfcsp instances can be achieved in polynomial time
  (assuming the Exponential Time Hypothesis).

These problems are a natural meeting point for the theory of constraint
satisfaction problems and multiobjective optimization.
We also suggest a number of interesting directions for future research.

\vspace{3mm}
% %\begin{framed}
% \begin{center}
% \bf This is a new version of a paper first written in 
% November 2013, with key new  insights and significantly improved results.
% \end{center}
%\end{framed}

\end{abstract}
\clearpage
\setcounter{page}{1}

\section{Introduction}
% \Sknote{ We should put in a very prominent message saying that:
% significantly generalized and strengthened from the version of
% November 2013. It is very standard for reviewers to just copy their
% old review, assuming that the paper is the same.}
The theory of approximation algorithms for constraint satisfaction
problems (CSPs) is a very central and well developed part of modern
theoretical computer science.  Its study has involved fundamental
theorems, ideas, and problems such as the PCP theorem, linear and
semidefinite programming, randomized rounding, the Unique Games
Conjecture, and deep connections between them~\cite{AroraS98,
  AroraLMSS98, GW-sdp95, Khot02, Raghavendra08, RaghavendraS09}.

In this paper, we initiate the study of {\em simultaneous
  approximation algorithms} for constraint satisfaction problems. A
typical such problem is the simultaneous \maxcut problem: Given a
collection of $k$ graphs $G_i = (V, E_i)$ on the same vertex set $V$, the
problem is to find a single cut (i.e., a partition of $V$) so that in
{\em every} $G_i$, a large fraction of the edges go across the cut.

More generally, let $q$ be a constant positive integer, and let
$\calF$ be a set of bounded-arity predicates on $[q]$-valued
variables. Let $V$ be a set of $n$ $[q]$-valued variables.  An
$\calF$-CSP is a weighted collection $\calW$ of constraints on $V$,
where each constraint is an application of a predicate from $\calF$ to
some variables from $V$.  For an assignment $f: V \to [q]$ and a
$\calF$-CSP instance $\calW$, we let $\val(f, \calW)$ denote the total
weight of the constraints from $\calW$ satisfied by $f$.
The  {\sc Max}-$\calF$-CSP problem is to find $f$ which maximizes
$\val(f, \calW)$. If $\calF$ is the set of all predicates on $[q]$
of arity $w$, then {\sc Max}-$\calF$-CSP is also called $\maxcspwq$.
%\Sknote{Defined maxwcsp  here}
%\Ssnote{Change $[q]$ to  $[q]$ everywhere?}

We now describe the setting for the problem we consider:
\emph{$k$-fold simultaneous} {\sc Max}-$\calF$-CSP. Let $\calW_1,
\ldots, \calW_k$ be $\calF$-CSPs on $V$, each with total weight $1$.
Our high level goal is to find an assignment $f: V \to [q]$ for which
$\val(f, \calW_\ell)$ is large for all $\ell \in [k]$.
% \Ssnote{We don't need
%   total weight 1. Or say, wlog, each instance has total weight 1. Or
%   we can stress in the notation section, that the total weight need
%   not be 1.}

These problems fall naturally into the domain of multi-objective
optimization: there is a common search space, and multiple objective
functions on that space.  Since even optimizing one of these objective
functions could be NP-hard, it is natural to resort to approximation
algorithms. Below, we formulate some of the approximation criteria
that we will consider, in decreasing order of difficulty:
\begin{enumerate}
\item {\bf Pareto approximation:} 
Suppose $(c_1, \ldots, c_k) \in [0,1]^k$ is such that
there is an assignment $f^*$ with $\val(f^*, \calW_\ell) \geq c_\ell$
for each $\ell \in [k]$.

An $\alpha$-Pareto approximation algorithm in this context is an algorithm,
which when given $(c_1, \ldots, c_k)$ as input, finds an assignment $f$ such that
$\val(f, \calW_\ell) \geq \alpha \cdot c_\ell$, for each $\ell \in [k]$.

\item {\bf Minimum approximation:} This is basically the Pareto
  approximation problem when $c_1 = c_2= \ldots = c_k$. Define $\opt$
  to be the maximum, over all assignments $f^*$, of $\min_{\ell \in
    [k]} \val(f^*, \calW_\ell)$.

  An $\alpha$-minimum approximation algorithm in this context is an
  algorithm which finds an assignment $f$ such that $\min_{\ell \in
    [k]} \val(f, \calW_\ell) \geq \alpha \cdot \opt$.

\item {\bf Detecting Positivity:} This is a very special case of the above,
where the goal is simply to determine whether there is an assignment $f$ which
makes $\val(f, \calW_\ell) > 0$ for all $\ell \in [k]$.

At the surface, this problem appears to be a significant
weakening of the the simultaneous approximation goal.
\end{enumerate}

When $k = 1$, minimum approximation and Pareto approximation correspond
to the classical \maxcsp approximation problems
(which have received much attention). Our focus in this paper is on
general $k$. As we will see in the discussions below, the nature of
the problem changes quite a bit for $k > 1.$  In
particular, direct applications of classical techniques like random
assignments and convex programming relaxations fail to give even a
constant factor approximation.

The theory of {\em exact} multiobjective optimization has been very well studied, 
(see eg.~\cite{PapadimitriouY00, Diakonikolas11} and the references
therein). For several optimization problems such as shortest paths,
minimum spanning trees, matchings, etc, there are polynomial time
algorithms that solve the multiobjective versions exactly. For
{\sc Max-SAT}, simultaneous approximation was studied by Gla\ss er \etal~\cite{Glasser11}.

We have two main motivations for studying simultaneous approximations
for CSPs. Most importantly, these are very natural algorithmic
questions, and capture naturally arising constraints in a way which
more na\"{i}ve formulations (such as taking linear combinations of the
given CSPs) cannot.  Secondly, the study of simultaneous approximation
algorithms for CSPs sheds new light on various aspects of standard
approximation algorithms for CSPs. For example, our algorithms are
able to favorably exploit some features of the trivial
random-assignment-based $\nfrac{1}{2}$-approximation algorithm for
\maxcut, that are absent in the more sophisticated SDP-based
0.878-approximation algorithm of Goemans-Williamson~\cite{GW-sdp95}.
%\Ssnote{Can we ask more concretely, under what conditions on the graph
%does the
%GW rounding concentrate? Even what is the variance of the GW rounding?}
\Ssnote{Add something to the motivation?}

%\Ssnote{The last line above, ``more useful'' seems quite
%  provocative. We should probably tone it down a bit, and specify more
  %clearly what do we mean}

%\subsection{Some observations about simultaneous approximation}
%\Ssnote{This is too long a section. We should tighten it up.}
%{\color{red} We begin with some simple observations about simultaneous
%  CSPs}
% , and
% then discuss why a direct application of the classical CSP algorithms 
% fails in this setting.
 % \Ssnote{I am not sure if presenting such an
 %  enumerated list of disconnected facts is a good way.}
%
%\begin{enumerate}
%
% \item \Ssnote{This will go -- not so interesting.}
% Detecting positivity for $k$-fold simultaneous
% \maxfcsp has a trivial $n^{O(k)}$ time algorithm: For every possible
% way of choosing one constraint per instance, check if it is possible to satisfy
% those constraints.
% \Ssnote{Are you sure of the reduction to Subsetsum?}
% \Sknote{Yeah I think so.}
%\end{enumerate}

\subsection{Observations about simultaneous approximation}
We now discuss why a direct application of the classical CSP
algorithms fails in this setting, and limitations on the
approximation ratios that can be achieved.

%\begin{enumerate}
%\item 

We begin with a trivial remark. 
Finding an $\alpha$-minimum (or Pareto) approximation to the
 $k$-fold \maxfcsp is at least as hard as finding an
  $\alpha$-approximation the classical \maxfcsp problem (i.e., $k =
  1$).  Thus the known limits on polynomial-time approximability
  extend naturally to our setting.
% \Ssnote{Changed
%     this line.} \Ssnote{Do we need this as a separate point, or should
%   we merge it where we say $k=1$ corresponds to classical Max-CSP problems?}
% Thus there are some known basic limits on the
%   approximation factors achievable in polynomial time.

\medskip 
\noindent \textbf{\textsc{Max-1-SAT}.} The simplest simultaneous CSP is {\sc Max-1-SAT}.
 The problem of getting a $1$-Pareto or $1$-minimum approximation
  to $k$-fold simultaneous {\sc Max-1-SAT} is essentially the \np-hard
  SUBSET-SUM problem. There is a simple $2^{\poly(k/\epsilon)}\cdot
  \poly(n)$-time $(1-\epsilon)$-Pareto approximation algorithm based
  on dynamic programming.

It is easy to
  see that detecting positivity of a $k$-fold simultaneous {\sc
    Max-1-SAT} is exactly the same problem as detecting satisfiability
  of a SAT formula with $k$ clauses (a problem studied in the
fixed parameter tractability community. Thus, this problem
can be solved in time  $2^{O(k)} \cdot \poly(n)$ (see~\cite{Marx13}),
and under the Exponential
Time Hypothesis, one does not expect a polynomial time algorithm
 when $k = \omega(\log n)$. 
% \Ssnote{Does this add anything above the randomized
%     algorithm? We don't even need to talk of de-randomizing it. Trying
%   out $2^{O(k)}$ assignments gets the answer with constant
%   probability. We're talking of randomized algorithms anyway.}
% \Ssnote{What is the reference? Are we referring to slides?
%     Is it known that this is tight? It might be worth mentioning if
%     there is a lower bound here.}
% \Sknote{ It is tight .. but that will take away
% the punch from the lower bound. I wasn't able to find a paper with this result
% and the slides don't cite so let us just cite the slides .. it is no big deal.}
% \Sknote{Slides here: http://www.cs.bme.hu/~dmarx/papers/marx-bergen-2013-csp.pdf}

\medskip
\noindent {\bf Random Assignments.} Let us consider algorithms based on
random assignments. A typical example is \maxcut. A uniformly random
cut in a weighted graph graph cuts $\nfrac{1}{2}$ the total weight in
expectation. This gives a $\nfrac{1}{2}$-approximation to the
classical \maxcut problem.

  % First recall the simple random partition algorithm in the classical
  % case.  
  If the cut value is concentrated around $\nfrac{1}{2},$ with high
  probability, we would obtain a cut that's simultaneously good for
  all instances.  For an {\em unweighted} graph\footnote{We use the
    term ``unweighted" to refer to instances where all the constraints
    have the same weight.  When we talk about simultaneous
    approximation for unweighted instances $\calW_1, \ldots, \calW_k$
    of MAX-$\calF$-CSP, we mean that in each instance $\calW_i$, all
    constraints with nonzero weight have the equal weights (but that
    equal weight can be different for different $i$).}  $G$ with
  $\omega(1)$ edges, a simple variance calculation shows that a
  uniformly random cut in the graph cuts a $\left(\frac{1}{2}
    -o(1)\right)$ fraction of the edges {\em with high probability}.
  Thus by a union bound, for $k = O(1)$ simultaneous unweighted
  instances $G_1, \ldots, G_k$ of \maxcut, a uniformly random cut
  gives a $\left(\frac{1}{2} - o(1)\right)$-minimum (and Pareto)
  approximation with high probability. However, for {\em weighted}
  graphs, the concentration no longer holds, and the algorithm fails
  to give any constant factor approximation.
%\Ssnote{We should explain in a line or two why it breaks down, and
%  that's there's no trivial fix.}

  For general CSPs, even for unweighted instances, the total weight
  satisfied by a random assignment does not necessarily concentrate.
  In particular, there is no ``trivial" random-assignment-based
  constant factor approximation algorithm for simultaneous general
  CSPs.

\medskip
\noindent {\bf SDP Algorithms.}
How do algorithms based on semi-definite programming (SDP)
  generalize to the simultaneous setting?

  For the usual \maxcut problem ($k=1$), the celebrated
  Goemans-Williamson SDP algorithm~\cite{GW-sdp95} gives a
  $0.878$-approximation. The SDP relaxation generalizes naturally to
  to the simultaneous setting; it allows us to find a vector solution
  which is a simultaneously good cut for $G_1, \ldots, G_k$. Perhaps we apply
  hyperplane rounding to the SDP solution to obtain a simultaneously
  good cut for all $G_i$? We know that each $G_i$ gets a good cut in
  expectation, but we need each $G_i$ to get a good cut {\em with high
    probability} to guarantee a simultaneously good cut.

However, there are cases where the hyperplane rounding fails
completely.  For weighted instances, the SDP does not have {\em any}
constant integrality gap.  For unweighted instances, for every fixed
$k$, we find an instance of $k$-fold simultaneous \maxcut (with
arbitrarily many vertices and edges) where the SDP relaxation has
value $1- \Omega\left(\frac{1}{k^2}\right)$, while the optimal
simultaneous cut has value only $\nfrac{1}{2}$.
%\Ssnote{maybe ``we construct $k$
%%unweighted  instances of simultaneous max-cut'' and ``arbitrarily many
%edges'' rather than 'vertices'}
Furthermore, applying the hyperplane
rounding algorithm to this vector solution gives (with probability 1)
a simultaneous cut value of {\em 0}. These integrality gaps are
described in Section~\ref{section:sdp}.

Thus the natural extension of SDP based techniques for
simultaneous approximation fail quite spectacularly. A-priori, this
failure is quite surprising, since SDPs (and LPs) generalize to the
multiobjective setting seamlessly. 
% \Ssnote{This line seems too similar
% to the opening line of the second para in this point.}

%Note that our algorithm
%in Theorem~\ref{ adsf }, achieves essentially the same approximation ratio
%as this integrality gap.

\medskip
\noindent {\bf Matching Random Assignments?}
Given the ease and simplicity of algorithms based on random
assignments for $k=1,$ giving algorithms in the simultaneous setting
that match their approximation guarantees is a natural
benchmark. Perhaps it is always possible to do as well in the
simultaneous setting as a random assignment for one instance?

% A simple analysis of an algorithm based on randomized
%   assignments (\emph{e.g.} \maxcut) implies that the if the objective
%   value does not concetrate, then there must exist a variable with
%   large degree. Since there can't be too many variables with large
%   degree, this suggests that if we guess the assignment to all large
%   degree variables and assign the rest randomly, we should be able to
%   do as well in the simultaneous setting as a random assignment for
%   one instance. 
  Somewhat surprisingly, this is incorrect.  For
  simultaneous {\sc Max-E$w$-SAT} ({\sc CNF-SAT} where every clause
  has exactly $w$ distinct literals), a simple reduction from {\sc
    Max-E$3$-SAT} (with $k=1$) shows that it is \np-hard to give a
  $(\nfrac{7}{8} + \epsilon)$-minimum approximation for $k$-fold
  simultaneous {\sc Max-E$w$-SAT} for large enough constants $k$. 
 \begin{proposition}
\label{proposition:results:hardness:maxwsat}
For all integers $w \ge 4$ and $\epsilon > 0$, given $k \ge 2^{w-3}$
instances of {\sc Max-E$w$-SAT} that are simultaneously satisfiable,
it is \np-hard to find a $(\nfrac{7}{8}+\eps)$-minimum (or
Pareto) approximation.
\end{proposition} 
On the other hand, a random assignment to a single {\sc Max-E$w$-SAT}
instance satisfies a $1-2^{-w}$ fraction of constraints in expectation.

  This shows that simultaneous CSPs can have worse approximation
  factors than that expected from a random assignment. In particular,
  it shows that simultaneous CSPs can have worse approximation factors
  than their classical ($k = 1$) counterparts.

\subsection{Results}
%\Ssnote{In this section, we need to make sure it's clear why the
%  theorem is non-trivial. It should be clear that the basic techniques
%fail. For \maxcut, we want to make it clear that the random assignment
%algorithm fails because the cut value may not be concentrated, because
%of a skewed weight distribution. For \maxsat, we need to clarify why
%more ideas are required, and why we need to not just pick heavy
%variables, but also recurse over all possible assignments to the
%picked variables.}
Our results address the approximability of $k$-fold simultaneous
\maxfcsp for large $k$.  Our main algorithmic result shows that for
every $\calF$, and $k$ not too large, $k$-fold simultaneous \maxfcsp
has a constant factor Pareto approximation algorithm.
\begin{theorem}
\label{theorem:results:maxcsp}
Let $q, w$ be constants.
Then for every $\epsilon > 0$, 
there is a $2^{O(\nfrac{k^4}{\eps^2}\log(\nfrac{k}{\eps}))} \cdot \poly(n)$-time $\left(\frac{1}{q^{w-1}} - \epsilon \right)$-Pareto
approximation algorithm for $k$-fold simultaneous
\maxcspwq. 
\end{theorem}

The dependence on $k$ implies that the algorithm runs in polynomial
time up to $k = \tilde{O}((\log n)^{\nfrac{1}{4}})$ simultaneous
instances \footnote{The $\tilde{O}(\cdot)$ hides $\poly(\log \log n)$
  factors.}.  The proof of the above Theorem appears in
Section~\ref{section:maxcsp}, and involves a number of ideas.  In
order to make the ideas clearer, we first describe the main ideas for
approximating simultaneous $\maxand$ (which easily implies the $q = w = 2$ special case of the above
theorem); this appears in Section~\ref{section:maxand}.

For particular CSPs, our methods allow us to do significantly better,
as demonstrated by our following result for \maxwsat.
\begin{theorem}
\label{theorem:results:maxwsat}
Let $w$ be a constant. For every $\eps > 0$, there is a
$2^{O(\nfrac{k^3}{\eps^2}\log(\nfrac{k}{\eps}))} \cdot \poly(n)$-time
$\left( \nfrac{3}{4} -\epsilon\right)$-Pareto approximation algorithm
for $k$-fold \maxwsat.
\end{theorem}
% Suppose we're given $\eps \in
% (0,\nfrac{2}{5}],$ $k$ simultaneous \maxwsat instances $\calW_1,
% \ldots, \calW_\ell$ on $n$ variables, and target objective value $c_1,
% \ldots,c_k$ with the guarantee that there exists an assignment
% $f^\star$ such that for each $\ell \in [k],$ we have $\val( f^\star,
% \calW_\ell) \ge c_\ell.$ Then, the algorithm {\algowwsat} runs in time
% $\exp(\nfrac{k^3 w^2}{\eps^2}\log
% (\nfrac{kw^2}{\eps^2}))\cdot\poly(n),$ and with probability at least
% $0.9,$ outputs an assignment $f$ such that for each $\ell \in [k],$ we
% have, $\val(f, \calW_\ell) \ge \left( \nfrac{3}{4} -\epsilon\right)
% \cdot c_\ell.$

Given a single {\sc Max-E$w$-SAT} instance, a random assignment
satisfies a $1-2^{-w}$ fraction of the constraints in expectation. The
approximation ratio achieved by the above theorem seems unimpressive
in comparison (even though it is for general \maxwsat). However,
Proposition~\ref{proposition:results:hardness:maxwsat} demonstrates it
is \np-hard to do much better.
%

%work our way up to the full proof by gradually introducing these ideas
%in proofs of simpler results. In Section~\ref{section:max-cut}, using
%the simplest versions of our ideas, we give a $(\nfrac{1}{2} -
%\epsilon)$-Pareto approximation algorithm for simultaneous \maxcut.
%In Section~\ref{section:maxsat}, we combine this basic approach with
%linear-programming relaxations to get a $(\nfrac{3}{4} -
%\epsilon)$-Pareto approximation algorithm for simultaneous \maxsat.
%In Section~\ref{section:maxand}, we introduce some significant
%additions to the both the algorithm and the analysis framework to get
%a $(\nfrac{1}{2} -\epsilon)$-Pareto approximation for simultaneous
%\maxand (this implies the same factor Pareto approximation algorithm
%for the general simultaneous {\sc Max-2-CSP}, where every arity-2
%predicate on Boolean variables is allowed). Finally, in
%Section~\ref{section:maxcsp}, we give the full proof of
%Theorem~\ref{theorem:results:maxcsp}.

\paragraph{Remarks}
\begin{enumerate}
\item As demonstrated by
  Proposition~\ref{proposition:results:hardness:maxwsat}, it is
  sometimes impossible to match the approximation ratio achieved by a
  random assignment for $k=1$. By comparison, the approximation ratio
  given by Theorem~\ref{theorem:results:maxcsp} is slightly better
  than that achieved by a random assignment ($\nfrac{1}{q^w}$). This
  is comparable to the best possible approximation ratio for $k=1,$
  which is $w/q^{w-1}$ up to constants~\cite{MakarychevM12, Chan13}.
  % (also known to be essentially optimal under $\p \neq
  % \np$~\cite{Chan13}).
  Our methods also prove that picking the best assignment out of
  $2^{O(\nfrac{k^4}{\eps^2}\log(\nfrac{k}{\eps}))}$ independent and
  uniformly random assignments achieves a
  $\left(\nfrac{1}{q^w}-\eps\right)$-Pareto approximation with
  high probability.

% Because of the
%   SAT hardness, it's not trivial to achieve random, we beat it, and
%   comparable to the best UGC hardness.

\item Our method is quite general. For any CSP with a convex
  relaxation and an associated rounding algorithm that assigns each
  variable independently from a distribution with certain smoothness
  properties (see Section~\ref{section:maxand:lp}), it can be combined
  with our techniques to achieve essentially the same approximation
  ratio for $k$ simultaneous instances.

\item We reiterate that Pareto approximation algorithms achieve a
  multiplicative approximation for each instance. One could also
  consider the problem of achieving simultaneous approximations with
  an $\alpha$-multiplicative and $\epsilon$-additive error. This
  problem can be solved by a significantly simpler algorithm and
  analysis (but note that this variation does not even imply an
  algorithm for detecting positivity).
% \Ssnote{I think additive
%  approximations are considerably simpler. You don't need to do any
%  perturbation etc, just remove heavy vertices until the remaining
%  weight is tiny.}
\end{enumerate}

\subsection{Complementary results}

\subsubsection{Refined hardness results}

As we saw earlier, assuming ETH, there is no algorithm for even
detecting positivity of $k$-fold simultaneous {\sc Max-1-SAT} for
$k=\omega(\log n).$ There are trivial examples of CSPs for which
detecting positivity (and in fact $1$-Pareto approximation)
can be solved efficiently: eg. simultaneous CSPs based on
monotone predicates (where no
negations of variables are allowed) are maximally satisfied by the
all-1s assignment.
Here we prove that for any
``nontrivial" collection of Boolean predicates $\calF$, assuming ETH,
there is no polynomial time algorithm for detecting positivity for
$k$-fold simultaneous \maxfcsp instances for $k = \omega(\log n).$ In
particular, it is hard to obtain any poly-time constant factor
approximation for $k = \omega(\log n).$ This implies a complete {\em
  dichotomy theorem} for constant factor approximations of $k$-fold
simultaneous Boolean CSPs.

% We prove that for In contrast to the above approximation algorithms
% for $k = \log^{O(1)}n$, when $k = \omega(\log n)$ 

% More precisely, we show that for every ``nontrivial" collection of
% Boolean predicates $\calF$, even detecting positivity for simultaneous
% \maxfcsp instances is hard. 
%, for $k = \log^{\Omega(1)} n$.

%Let
%$\calF$ be a finite set of predicates. An instance of the $k$-fold
%simultaneous MAX-$\calF$-CSP problem is given by $k$ collections
%$\calC_1, \ldots, \calC_k$, of constraints on the variables $x_1,
%\ldots, x_n$, where each constraint is an application of a predicate
%from $\calF$ on some subset of the variables

A predicate $P: \{0,1\}^w \to \{\true, \false\}$ is said to be {\em
  $0$-valid}/{\em $1$-valid} if the
all-$0$-assignment/all-$1$-assignment satisfies $P$. We call a
collection $\calF$ of predicates {\em $0$-valid}/{\em $1$-valid} if
all predicates in $\calF$ are $0$-valid/$1$-valid.  Clearly, if
$\calF$ is $0$-valid or $1$-valid, the simultaneous \maxfcsp instances
can be solved exactly (by considering the
all-$0$-assignment/all-$1$-assignment). Our next theorem shows that
detecting positivity of $\omega(\log n)$-fold simultaneous \maxfcsp,
for all other $\calF$, is hard.
\begin{theorem}
\label{thm:introdichotomy}
Assume the Exponential Time Hypothesis~\cite{ImpagliazzoP01, ImpagliazzoPZ01}. Let $\calF$ be a fixed finite
set of Boolean predicates. If $\calF$ is not $0$-valid or $1$-valid,
then for $k = \omega(\log n )$, detecting positivity of $k$-fold
simultaneous \maxfcsp on $n$ variables requires time super-polynomial
in $n$.
\end{theorem}
%\Ssnote{We should state the ETH somewhere, maybe in the hardness section.}

Crucially, this hardness result holds even if we require
that every predicate in an instance has all its inputs being {\em distinct variables}.
%\footnote{Without this requirement,
%for many $\calF$ the hardness basically follows from the hardness of simultaneous MAX-1-SAT;
%this hardness reduction does not seem to capture the true nature of max-$\calF$-CSP.}.

% \Ssnote{There is actually a small issue to note here. There are $k$
%   instances on $\Theta(k)$ variables and with only a constant number
%   of constraints in each of the instances, already requires time
%   superpolynomial in $n$ for $k = \omega(\log n).$ Currently we are
%   interpreting this as $\omega(\log n)$ instances on $n$
%   variables. One things that's not coming out in the current
%   statements is that the instance sizes are not unreasonably large
%   (right now they could in fact be superpolynomial in $n$). }

Our proof uses techniques underlying the dichotomy
theorems of Schaefer~\cite{Schaefer78} for exact CSPs, and of Khanna
\etal~\cite{KhannaSTW01} for {\sc Max}-CSPs (although our easiness criterion
is different from the easiness criteria in both these papers).

\subsubsection{Simultaneous approximations via SDPs}

It is a tantalizing possibility that one could use SDPs to improve the
LP-based approximation algorithms that we develop.  Especially for constant
$k$, it is not unreasonable to expect that one could obtain a constant
factor Pareto or minimum approximation, for $k$-fold simultaneous
CSPs, better than what can be achieved by linear programming methods.

In this direction, we show how to use simultaneous SDP relaxations
to obtain a polynomial time $(\nfrac{1}{2} + \Omega(\nfrac{1}{k^2}))$-minimum approximation
for $k$-fold simultaneous \maxcut on {\em unweighted graphs}.
\begin{theorem}
\label{theorem:results:unweightedmc}
For large enough $n$, there is an algorithm that, given $k$-fold simultaneous unweighted
\maxcut instances on $n$ vertices, runs in time $2^{2^{2^{O(k)}}}
\cdot \poly(n),$ and computes a $\left(\frac{1}{2} +
  \Omega\left(\frac{1}{k^2}\right)\right)$-minimum approximation.
\end{theorem}

%\Sknote{Define MAX-w-CSP somewhere?}

\subsection{Our techniques}
%Due to space constraints, we only give an overview of some of the
%ideas that go into our algorithms and their analysis here. The technical
%parts of the paper are deferred to Part~\ref{part:paper} of the
%appendix. 
For the initial part of this discussion, we focus on the $q
= w = 2$ case, and only achieve a $1/4 - \epsilon$ Pareto
approximation.

\paragraph{Preliminary Observations}
First let us  analyze the behavior of the uniformly random assignment algorithm.
It is easy to compute, for each instance $\ell \in [k]$,
 the expected weight of satisfied constraints in instance $\ell$, which will be
at least $\frac{1}{4}$ of the total weight all constraints in instance $\ell$.
If we knew for some reason that in each instance the weight of satisfied constraints
was concentrated around this expected value {\em with high probability}, then we could take a union
bound over all the instances and conclude that a random assignment satisfies many constraints
in each instance with high probability. It turns out that for
any instance where the desired concentration does not occur,
there is some variable in that instance which has
high degree (i.e., the weight of all constraints involving that variable
is a constant fraction of the total weight of all constraints).
Knowing that there is such a high degree variable seems very useful for our
goal of finding a good assignment, since we can potentially influence the satisfaction
of the instance quite a bit by just by changing this one variable.

This motivates a high-level plan: either proceed by using the absence of influential variables
to argue that a random assignment will succeed, or proceed by trying to set the influential variables.

\paragraph{An attempt}
The above high-level plan motivates the following high-level algorithm.
First we identify a set $S \subseteq V$ of ``influential" variables.
This set of influential variables should be of small ($O(\log n)$) size,
so that we can try out all assignments to these variables. Next, we 
take a random assignment to the remaining variables,
$g: V \setminus S \to \{0,1\}$.
Finally, for each possible assignment $h: S \to \{0,1\}$, we consider the assignment
$h \cup g  : V \to \{0,1\}$ as a candidate solution for our simultaneous CSP.
We  output the assignment, if any, that has $\val(h \cup g, \calW_\ell) \geq \alpha \cdot c_\ell$ for each $\ell \in [k]$.
This concludes the description of the high-level algorithm.

For the analysis, we would start with the ideal assignment $f^* : V \to \{0,1\}$
achieving $\val(f^*, \calW_\ell) \geq c_\ell$ for each $\ell \in [k]$.
Consider the step of the algorithm
where $h$ is taken to equal $h^* \defeq f^*|_S$. We would like to say that
for each $\ell \in [k]$ we have:
$$\val(h^* \cup g, \calW_\ell) \geq  (\frac{1}{4} - \epsilon) \cdot \val(f^*, \calW_\ell),$$ 
with high probability, when $g: V \setminus S \to \{0,1\}$ is chosen uniformly at random.
(We could then conclude the analysis by a union bound.)

A simple calculation shows that $\E[\val(h^* \cup g, \calW_\ell)] \geq \frac{1}{4} \cdot \val(f^*, \calW_\ell)$,
so each instance is well satisfied in expectation. Our hope is thus that
$\val(h^* \cup g, \calW_\ell)$ is concentrated around its mean with high probability.

\iffalse
{{\bf Alternate text:}
If for some instance $\ell$ we were so lucky that $\val(h^* \cup g, \calW_\ell)$
is close to its expectation with high probability, then that instance $\ell$ is
also well-satisfied with high probability.
However, for some instances we may not be so lucky. For these instances,
we would instead hope to show that $\val(h^* \cup g, \calW_\ell)$ already has
a high .
{\bf end alternate text.}
}\fi

There are two basic issues with this approach\footnote{
These problems do not arise if we only aim for the weaker ``additive-multiplicative"
Pareto approximation guarantee (where one allows for both some additive loss and 
multiplicative loss in the approximation), and in fact the above mentioned high-level plan does work.
The pure multiplicative approximation guarantee seems to be significantly more delicate.}:
\begin{enumerate}
\item The first issue is how to define the set $S$ of influential variables.
 For some special CSPs (such as \maxcut and {\sc Max-SAT}), there is a natural choice which works
(to choose a set of variables with high degree, which is automatically small). But for general CSPs,
it could be the case that variables with exponentially small degree are important contributors
to the ideal assignment $f^*$.

%There are $k$-fold simultaneous MAXCSPs where in each instance a constant fraction
%of the variables variables have high
%(say $1/n$) total degree, yet in every assignment some instance
%has very small (say $2^{-\Omega(\sqrt{n}}$) satisfied weight. Thus for these instances,
%a constant fraction of all variables are ``influential" (they have the ability to fully satisfy
%an instance), and so it is infeasible to brute force search over all assignments to these
%influential variables.

\item Even if one chooses the set $S$ of influential variables appropriately, the analysis
cannot hope to argue that $\val(h^* \cup g, \calW_\ell)$ concentrates around its expectation
with high probability. Indeed, it can be the case that for a random assignment
$g$, $\val(h^* \cup g, \calW_\ell)$ is not concentrated at all.
\end{enumerate}

\paragraph{A working algorithm:}
Our actual algorithm and analysis solve these problems by proceeding in a slightly different way.
The first key idea is to find the set of influential variables by 
iteratively including variables into this set, and simultaneously assigning these variables. This
leads to a tree-like evolution of the set of influential variables. 
The second key idea is in the analysis: instead of arguing about the performance
of the algorithm when considering the partial assignment $h^* = f^*|_S$,
we will perform a delicate {\em perturbation} of $h^*$ to obtain an $h': S \to \{0,1\}$,
and show that $\val(h' \cup g, \calW_\ell)$ is as large as desired. Intuitively,
this perturbation only slightly worsens the satisfied weight of $h^*$, while reducing the 
reliance of the good assignment $f^*$ on any specialized properties of $f^*|_S$.

To implement this, the algorithm will maintain a tree of possible evolutions of a
set $S \subseteq V$ and a partial assignment $\rho : S \to \{0,1\}$.
In addition, every variable $x \in S$ will be labelled by an instance $\ell \in [k]$.
The first stage of the algorithm will grow this tree in several steps.
In the beginning, at the root of the tree, we have $S = \emptyset$.
At every stage, we will either terminate that branch of the tree,
or else increase the size of the set $S$ by $1$ (or $2$),
and consider all $2$ (or $4$) extensions of $\rho$ to the newly grown $S$.

To grow the tree, the algorithm considers a random assignment $g: V \setminus S \to \{0,1\}$,
and computes, for each instance $i \in [k]$, the expected satisfied weight
$\E_g [ \val(\rho \cup g , \calW_\ell)   ] $
and the variance of the satisfied weight $\Var_g[ \val(\rho \cup g, \calW_\ell)]$.
We can thus classify instances as concentrated or non-concentrated.
If more than $t$ variables in $S$ are labelled by instance $\ell$ (where $t = O_{k,\epsilon}(1)$
is some parameter to be chosen), we call instance $\ell$ saturated.
If every unsaturated instance is concentrated, then we are done with
this $S$ and $\rho$, and this branch of the tree gets terminated.

Otherwise, we know that there some unsaturated instance $\ell$ which is not concentrated.
We know that this instance $\ell$ must have some variable $x \in V \setminus S$
which has high {\em active degree} (this is the degree after taking into account the partial
assignment $\rho$). The algorithm now takes two cases:
\begin{itemize}
\item {\bf Case 1:} If this high-active-degree variable $x$ is involved in a high-weight constraint 
on $\{x,y\}$ for some $y \in V\setminus S$, then we include both $x, y$ into the set $S$, and consider
all $4$ possible extensions of $\rho$ to this new $S$. $x,y$ are both  labelled with instance $\ell$.
\item {\bf Case 2:} Otherwise, every constraint involving $x$ is low-weight (and in particular there must
be many of them), and in this case we include $x$ into the set $S$, and consider
both possible extensions of $\rho$ to this new $S$. $x$ is labelled with instance $\ell$.
\end{itemize}
This concludes the first stage of the algorithm, which created a tree whose leaves contain
various $(S, \rho)$ pairs.

For the second stage of the algorithm we visit each leaf $(S, \rho)$.
We choose a uniformly random $g: V \setminus S \to \{0,1\}$, and
consider for every $h: S \to \{0,1\}$, the assignment $h \cup g: V \to
\{0,1\}$. Note that we go over all assignments to the set $S,$
independent of the partial assignment to $S$ associated with the leaf.

\paragraph{The analysis:}
At the end of the evolution, at every leaf of the tree every instance
is either highly-concentrated or saturated. If instance $\ell$ is
highly-concentrated, we will have the property that the random
assignment to $V \setminus S$ has the right approximation factor for
instance $\ell$.  If the instance $\ell$ is saturated, then we know
that there are many variables in $S$ labelled by instance $\ell$; and
at the time these variables were brought into $S$, they had high
active degree.

The main part of the analysis is then a delicate {\em perturbation}
procedure, which starts with the partial assignment $h^* \defeq
f^*|_S$, and perturbs it to some $h': S \to \{0,1\}$ with a certain
robustness property. Specifically, it ensures that for every saturated
instance $\ell \in [k]$.  we have $\val( h' \cup g, \calW_\ell)$ is at
least as large as the total weight in instance $\ell$ of all
constraints not wholly contained within $S$. At the same time, the
perturbation ensures that for unsaturated instances $\ell \in [k]$,
$\val(h' \cup g, \calW_\ell)$ is almost as large as $\val(h^*\cup g,
\calW_\ell)$.  This yields the desired Pareto approximation.  The
perturbation procedure modifies the assignment $h^*$ at a few
carefully chosen variables (at most two variables per saturated
instance). After picking the variables for an instance, if the
variables were brought into $S$ by Case 1, we can satisfy the heavy
constraint involving them. Otherwise, we use a Lipschitz concentration
bound to argue that a large fraction of the constraints involving the
variable and $V \setminus S$ can be satisfied; this is the second place where we use the randomness in the choice of
$g$.

As we mentioned earlier, this perturbation is necessary!
It is not true the assignment $h^* \cup g$ will give a good Pareto
approximation with good probability \footnote{See
  Section~\ref{section:example} for an example}.

\paragraph{Improved approximation, and generalization:}
To get the claimed $(\frac{1}{2} - \epsilon)$-Pareto approximation for
the $q = w = 2$ case, we replace the uniformly random choice of $g: V
\setminus S \to \{0,1\}$ by a suitable LP relaxation + randomized
rounding strategy. Concretely, at every leaf $(S, \rho)$, we do the
following.  First we write an LP relaxation of the residual MAX-2-CSP
problem. Then, using a rounding algorithm of Trevisan (which has some
desirable smoothness properties), we choose $g: V \setminus S \to
\{0,1\}$ by independently rounding each variable. Finally, for all $h
: S \to \{0,1\}$, we consider the assignment $h \cup g$. The analysis
is nearly identical (but crucially uses the smoothness of the
rounding), and the improved approximation comes from the improved
approximation factor of the classical LP relaxation for MAX-2-CSP.

The generalization of this algorithm to general $q,w$ is technical but
straightforward. One notable change is that instead of taking 2 cases
each time we grow the tree, we end up taking $w$ cases. In case $j,$
we have a set of $j$ variables such that the total weight of
constraints involving all the $j$ variables is large, however for
every remaining variable $z$, the weight of contraints involving all
the $j$ variables together with $z$ is small. The analysis
of the perturbation is similar.

The algorithm for \maxwsat uses the fact that the LP rounding gives a
$\nfrac{3}{4}$ approximation for \maxwsat. Moreover, since a \maxwsat
constraint can be satisfied by perturbing any one variable, the
algorithm does not require a tree of evolutions. It only maintains a
set of ``influential'' variables, and hence, is simpler.

\subsection{Related Work}

The theory of {\em exact} multiobjective optimization has been very well studied, 
(see eg.~\cite{PapadimitriouY00, Diakonikolas11} and the references
therein). 

The only directly comparable work for simultaneous approximation
algorithms for CSPs we are aware of is the work of Gla\ss er
\etal\cite{Glasser11} \footnote{They also give Pareto approximation
  results for simultaneous TSP (also see references therein).}. They
give a $\nfrac{1}{2}$-Pareto approximation for {\sc Max-SAT} with a
running time of $n^{O(k^2)}.$ For bounded width clauses, our algorithm
does better in both approximation guarantee and running time.
 
For \maxcut, there are a few results of a similar flavor. 
% The only prior work on simultaneous approximation algorithms for CSPs
% was for the problem of simultaneous \maxcut.  
For two graphs, the
results of Angel \etal\cite{Angel06} imply a $0.439$-Pareto
approximation algorithm (though their actual results are
incomparable to ours).  Bollob\'{a}s and Scott~\cite{BollobasS04} asked what
is the largest simultaneous cut in two unweighted graphs with $m$
edges each. Kuhn and Osthus~\cite{KuhnO07}, using the second moment
method, proved that for $k$ simultaneous unweighted instances, there
is a simultaneous cut that cuts at least $m/2 - O(\sqrt{km})$ edges in
each instance, and give a deterministic algorithm to find it (this
leads to a $(\frac{1}{2} - o(1))$-Pareto approximation for unweighted
instances with sufficiently many edges).  Our main theorem implies the
same Pareto approximation factor for simultaneous \maxcut on general
weighted instances, while for $k$-fold simultaneous \maxcut on
unweighted instances, our Theorem~\ref{theorem:results:unweightedmc}
gives a $\left(\frac{1}{2} + \Omega(\frac{1}{k^2})\right)$-minimum
approximation algorithm.  
% \Anote{Shall we include references that we
%   had found later? \cite{Glasher11}
%   Ex. \href{link}{http://arxiv.org/pdf/1107.0634.pdf} ; Our algorithm
%   for \maxwsat gets better approximation factor and better running
%   time.}

\subsection{Discussion}
We have only made initial progress on what we believe is a large
number of interesting problems in the realm of simultaneous
approximation of CSPs. We list here a few of the interesting
directions for further research:
\begin{enumerate}
\item When designing SDP-based algorithms for the classical {\sc
    Max-CSP} problems, we are usually only interested in the expected
  value of the rounded solution.  For $k$-fold simultaneous \maxfcsp
  with $k > 1,$ we are naturally led to the question of how
  concentrated the value of the solution output by the rounding is
  around its mean.

Decorrelation of SDP rounding arises in recent algorithms~\cite{Barak-Rag-Ste, Raghavendra-Tan, Venkat-Sinop} based on
SDP hierarchies. It would be interesting to see if such ideas could be useful in this
context.

Another interesting question of this flavor is whether
  there are natural conditions under which the Goemans-Williamson hyperplane rounding
  gives a good solution for MAXCUT with high probability.

%This seems to be related to the methods used to round SDP hierarchies. 

%\item By combining SDP rounding with randomized rounding, we obtained
%  a $\left(\frac{1}{2} + \epsilon\right)$ approximation for unweighted \maxcut. Is
%  it possible to get a $\left(\frac{1}{2} + \epsilon\right)$ approximation for
%  general weighted \maxcut? 

\item When $k = O(1)$, for each $\calF$, one can ask the question:
what is the best Pareto approximation factor achievable for
$k$-fold \maxfcsp in polynomial time?  While in Theorem~\ref{theorem:results:maxcsp}
we do not focus on giving improved approximation factors for special $\calF$,
our methods will give better approximation factors for any $\calF$
which has a good LP relaxation that comes equipped with a sufficiently smooth
independent-rounding algorithm.

It would be very interesting if one could employ SDPs for
approximating simultaneous \maxfcsp.
A particularly nice question here:
  {\em Is there a polynomial time $0.878$-Pareto approximation
    algorithm for $O(1)$-fold simultaneous \maxcut?} We do not even know
  a $(1/2 + \epsilon)$-Pareto approximation algorithm (but note that
  Theorem~\ref{theorem:results:unweightedmc} does give this for
  $O(1)$-fold simultaneous {\em unweighted} \maxcut).

\item As demonstrated by hardness result for \maxwsat given in
  Proposition~\ref{proposition:results:hardness:maxwsat}, even for
  constant $k,$ the achievable approximation factor can be strictly
  smaller than its classical counterpart. It would be very interesting
  to have a systematic theory of hardness reductions for simultaneous
  CSPs for $k=O(1).$ The usual paradigm for proving hardness of
  approximation based on label cover and long codes seems to break
  down completely for simultaneous CSPs.
% \Ssnote{Change this to a systematic
%     theory for hardness of simultaneous CSP} For large $k$, it seems
%   quite unlikely that one can achieve the same Pareto approximation
%   factor for $k$-fold simultaneous CSPs as for their classical
%   versions. However, for simultaneous CSPs, the usual paradigm for
%   proving hardness of approximation based on label cover and long
%   codes seems to break down completely. It would be very interesting
%   to find a CSP $\calF$ for which the achievable Pareto approximation
%   factor for $O(1)$-fold simultaneous \maxfcsp is strictly smaller
%   than its classical counterpart.
\end{enumerate}

\subsection{Organization of this paper}
We first present the notation required for our algorithms in
Section~\ref{section:notation}. We then describe our Pareto
approximation algorithm for \maxand (which is equivalent to {\sc Max-2-CSP$_2$}), and its generalization to
\maxcspwq in Sections~\ref{section:maxand} and \ref{section:maxcsp}
respectively. We then present our improved Pareto approximation for
\maxwsat in Section~\ref{section:maxwsat}.

% We first describe and analyze our Pareto approximation algorithms for
% simultaneous \maxcut and \maxsat in Section~\ref{section:max-cut}
% and~\ref{section:maxsat} respectively. In
% Section~\ref{section:notation} we present the notation required for
% our Pareto approximation algorithms for \maxand and \maxcspwq
% presented in Section~\ref{section:maxand} and~\ref{section:maxcsp}
% respectively. 
We present the additional results in the
appendix. The dichotomy theorem for the hardness of arbitrary CSPs is
presented in Section~\ref{section:hardness}, followed by our improved
minimum approximation algorithm for unweighted \maxcut in
Section~\ref{section:maxcut-unweighted}, and the SDP integrality gaps
in Section~\ref{section:sdp}. % Finally, for completeness, we include
% some required concentration inequalities (Appendix~\ref{section:concentration}).
% , and deferred proofs
% (Appendix~\ref{section:deferred}).

% some significant additions to the both the algorithm and the analysis
% framework to get a $(\nfrac{1}{2} -\epsilon)$-Pareto approximation for
% simultaneous \maxand (this implies the same factor Pareto
% approximation algorithm for the general simultaneous {\sc Max-2-CSP},
% where every arity-2 predicate on Boolean variables is
% allowed). Finally, in Section~\ref{section:maxcsp}, we give the full
% proof of Theorem~\ref{theorem:results:maxcsp}.

\section{Notation for the main algorithms}
\label{section:notation}

We now define some common notation that will be required for the
following sections on algorithms for \maxand and and for general
MAX-$\calF$-CSP. For the latter, will stop referring to the set of predicates
$\calF$, and simply present an algorithm for the problem \maxcspwq: this
is the MAX-$\calF$-CSP problem, where $\calF$ equals the set of all predicates on $w$ variables from the domain $[q]$. For \maxand, the alphabet $q$ and arity $w$ are both 2. % {\color{red}{\em Note that some notation regarding active constraints has changed:}} \Anote{Need to change the next statement} we will be talking about active constraints given {\em partial assignments $h: S\to [q]$} rather than active constraints given sets $S$.

Let $V$ be a set of $n$ variables. Each variable will take values from
the domain $[q]$.  Let $\calC$ denote a set of constraints of interest
on $V$ (for example, for studying \maxand, $\calC$ would be the set of
AND constraints on pairs of literals of variables coming from $V$). We use the notation $v \in C$ to denote that the $v$ is one of
the variables that the constraint $C$ depends on. Analogously, we
denote $T \subseteq C$ if $C$ depends on all the variables in $T.$ A
weighted MAXCSP instance on $V$ is given by a weight function $\calW :
\calC \to \rea_+,$ where for $C \in \calC$, $\calW(C)$ is the weight
of the constraint $C$. We will assume that $\sum_{C \in \calC}
\calW(C) = 1$.

A partial assignment $\rho$ is a pair $(S_\rho, h_\rho)$, where
$S_\rho \subseteq V$ and $h_\rho : S_\rho \to [q]$. (We also call a
function $h : S \to [q]$, a partial assignment, when $S$ is understood
from the context).  We say a contraint $C \in \calC$ is active given
$\rho$ if $C$ depends on some variable in $V \setminus S_\rho$, and
there exists full assignments $g_0, g_1: V \to [q]$ with
$g_i{|_{S_\rho}} = h_\rho$, such that $C$ evaluates to $\false$ under
the assignment $g_0$ and $C$ evaluates to $\true$ under the assignment
$g_1$.  (colloquially: $C$'s value is not fixed by $\rho$).  We denote
by $\Act(\rho)$ the set of constraints from $\calC$ which are active
given $\rho$. For a partial assignment $\rho$ and $C \in
\calC\setminus\Act(\rho)$, let $C(\rho) = 1$ if $C$'s value is fixed
to $\true$ by $\rho$, and let $C(\rho) = 0$ if $C$'s value is fixed to
$\false$ by $\rho$. For disjoint subsets $S_1, S_2 \subseteq V$ and partial assignments $f_1 : S_1 \to [q]$ and $f_2 : S_2 \to [q]$, let $f = f_1\cup f_2$ denote the assignment $f : S_1 \cup S_2 \to [q]$ with $f(x) = f_1(x)$ if $x\in S_1$, and $f(x) = f_2(x)$ if $x \in S_2$.
%  Let $S\subseteq V$ be any subset, for two
% assignments $f_1 : S \to [q]$ and $f_2 : V\setminus S \to [q]$,
% let $f = f_1\cup f_2$ is an assignment $f : V \to [q]$ such that
% $f(x) = f_1(x)$ if $x\in S$ otherwise $f(x) = f_2(x)$. 
Abusing
notation, for a partial assignment $\rho$ and an assignment $g : V
\setminus S_{\rho^\star} \to [q],$ we often write $\rho \cup g$
instead of $h_\rho \cup g.$ For two constraints $C_1,C_2 \in \calC,$
we say $C_1 \sim_{\rho} C_2$ if they share a variable that is
contained in $V \setminus S_\rho$. %Note that if $C_1 \sim_{\rho} C_2,$
%then both the constraints are active given $\rho.$

% \Sknote{Note the change in definition of $C(\rho)$: it is defined for all inactive constraints. I think this is the right way to do it, and I think it is consistent with everything later. So we should use these definitions in the notation section.}

Define the {\em active degree given $\rho$} of a variable $v \in V\setminus S_\rho$ by:
$$\activedeg_\rho(v, \calW) \defeq \sum_{ C\in \Act(\rho), C \owns v} \calW(C).$$	
For a subset $T\subseteq V\setminus S_\rho$ of variables, define its
active degree given $\rho$ by:
$$\activedeg_\rho(T, \calW) \defeq \sum_{ C\in \Act(\rho), C \supseteq T} \calW(C).$$	
Define the active degree of the whole instance $\calW$ given $\rho$:
$$\activedeg_\rho(\calW) \defeq \sum_{v \in V\setminus S_{\rho}} \activedeg_{\rho}(v, \calW).$$
For a partial assignment $\rho$, we define its value on an instance $\calW$ by:
$$\val(\rho, \calW) \defeq  \sum_{C \in \calC\setminus \Act(\rho)} \calW(C) C(\rho).$$
Thus, for a total assignment $f: V \to [q]$ extending $\rho$, we have the equality:
$$\val(f, \calW) - \val(\rho, \calW) = \sum_{C \in  \Act(\rho)} \calW(C) C(f).$$
%%% Local Variables: 
%%% mode: latex
%%% TeX-master: "new-simopt"
%%% End: 

\section{Simultaneous \maxand}
\label{section:maxand}

In this section, we give our approximation algorithm for simultaneous \maxand.
Via a simple reduction given Section~\ref{sec:reducsimple}, this implies the
$q = w = 2$ case of our main theorem, Theorem~\ref{theorem:results:maxcsp}.

%This algorithm considers a tree of possible evolutions of the set
%$S$, and brings both important variables and important constraints into $S$.
%The analysis uses a significantly more involved perturbation argument. 
%The randomness of the LP rounding algorithm gets used to not only 
%get simultaneously good values on all the low variance instances with high
%probability, but also to ensure enough diversity in the neighborhoods of
%variables in $S$. This diversity plays a crucial role in the perturbation 
%argument.

\iffalse{

Let $G = \{(C_1, w_1), \ldots, (C_m, w_m) \}$ be a \maxand instance on
a set of $n$ variables $V$, where each $C_i$ is a 2-AND clause along
with its associated weight $w_i.$ We will assume that $\sum_{i=1}^m
w_i = 1$. For a clause $C,$ let $C^+$ (resp.  $C^-$) denote the set of
variables $v \in V$ that appear unnegated (resp. negated) in the
clause $C.$ We say $v \in C$ if the variable $v$ appears in the clause
$C$ (unnegated or negated).

Let $f: V \to \{0,1\}$ be an assignment.
For each constraint $C_i$ of $G$, define $C_i(f)$ by:
  $$C_i(f) \defeq  \begin{cases} 1 & C_i \mbox{ is satisfied by $f$,}
    \\   
  0 & \mbox{ otherwise.} \end{cases}$$ Then define $G(f)$ by:
$$G(f) \defeq \sum_{i=1}^m w_i C_i(f).$$

}\fi

\subsection{Random Assignments}
We begin by giving a sufficient
condition for the value of a \maxand to be highly concentrated under
independent random assignments to the variables.
%\Ssnote{Why MAX-2-CSP?}

%We say that constraint $C_j$ is active if $\E[Y_j] \not\in \{0, 1\}$, 
%i.e., its value is not fixed.

Let $\rho$ be a partial assignment. Let $p: V\setminus S_{\rho} \to
[0,1]$ be such that $p(v) \in [\frac{1}{4}, \frac{3}{4}]$ for each $v
\in V\setminus S_{\rho}$. Let $g: V\setminus S_\rho \to [q]$ be a
random assignment obtained by sampling $g(v)$ for each $v$
independently with $\E[g(v)] = p(v)$. Define the random variable
$$Y \defeq \val(\rho \cup g, \calW) - \val(\rho, \calW) = \sum_{C \in \Act(\rho)} \calW(C) C(g).$$
The random variable $Y$ measures the contribution of active
constraints to $\val(\rho \cup g, \calW).$ Note that the two
quantities $\E[Y]$ and $\Var[Y]$ can be computed efficiently given
$p$.  We denote these by $\meancalc_\rho(p, \calW)$ and
$\varcalc_\rho(p, \calW)$. The following lemma proves that either $Y$
is concentrated, or there exists an active variable that contributes a
significant fraction of the total active-degree of the instance.
\begin{lemma}
\label{lemma:max2and:truemeanvar}
Let $p, Y$ be as above.
\begin{enumerate}
\item If $\varcalc_\rho(p, \calW) < \delta_0 \epsilon_0^2 \cdot \meancalc_\rho(p, \calW)^2$ then $\Pr[Y <  (1-\epsilon_0) \av[Y] ] < \delta_0$.
\item If $\varcalc_\rho(p, \calW) \ge \delta_0\epsilon_0^2 \cdot
  \meancalc_\rho(p, \calW)^2 $, then there exists $v \in V\setminus
  S_{\rho}$ such that
$$\activedeg_\rho(v, \calW) \geq \frac{\epsilon_0^2\delta_0}{64}\cdot \activedeg_{\rho}( \calW).$$
\end{enumerate}
\end{lemma}
The above lemma is a special case of
Lemma~\ref{lemma:maxcsp:truemeanvar} which is proved in
Section~\ref{sec:maxcspwq:random}, and hence we skip the proof. The
first part is then a simple application of the Chebyshev
inequality. For the second part, we use the assumption that $\varcalc$
is large, to deduce that there exists a constraint $C$ such that the
total weight of constraints that share a variable from $V \setminus S$
with $C,$ \emph{i.e.}, $\sum_{C_2 \sim_S C} \calW(C_2),$ is large. It
then follows that at least one variable $v \in C$ must have large
activedegree given $S.$
% \begin{proof}
% Case 1 of the lemma follows immediately from Chebyshev's inequality.  Case 2 is essentially identical to Lemma~\ref{lemma:uvar_lmean_relation}, and we omit the proof. (Note that $\meancalc(p, \calW)$ can be lower bounded by $\frac{1}{16}\cdot\sum_{C \in \Act(\rho)} \calW(C).)$
% \end{proof}

\subsection{LP relaxations}
\label{section:maxand:lp}
Let $(c_\ell)_{\ell \in [k]}$ be the given target values for the Pareto approximation problem.
Given a partial assignment
$\rho$, we can write the feasibility linear program for simultaneous
$\maxand$ as shown in figure~\ref{fig:max2and-lp1}, ~\ref{fig:max2and-lp2}. In this LP, for a
constraint $C$, $C^+$ ($C^-$) denotes set of variables that appears as
a positive (negative) literal in $C$. 

For $\vect$, $\vecz$ satisfying linear constraints $\LLand_1(\rho)$, let
$\smooth(\vect)$ denote the map $p: V \setminus S_\rho \to [0,1]$ with
$p(v) = \frac{1}{4} + \frac{t_{v}}{2}$. Note that $p(v) \in
[\nfrac{1}{4}, \nfrac{3}{4}]$ for all $v$. 

Given $\vec{t},\vecz$ satisfying $\LLand_1,$ the rounding algorithm
from~\cite{Trevisan98} samples each variable $v$ independently with
probabily $\smooth(\vect)(v).$ Note that this rounding algorithm is
\emph{smooth} in the sense that each variable is sampled independently
with a probability that is bounded away from 0 and 1. This will be
crucial for our algorithm. The following
theorem from~\cite{Trevisan98} proves that this rounding algorithm
finds a good integral assignment.
%
%
%\Anote{The following lemma is similar to Lemma~\ref{lemma:maxcsp:lp}, Should we remove this and just refer that lemma. Note also the LP for \maxand is slightly different since we are talking about literals here}
\begin{lemma}[\cite{Trevisan98}]
\label{lemma:max2and:lp}
Let $\rho$ be a partial assignment.
\begin{enumerate}
\item {\bf Relaxation:} For every $g_0 : V \setminus S_\rho \to \{0,1\}$, there exist
$\vect$, $\vecz$ satisfying $\LLand_1(\rho)$ such that for every
\maxand instance $\calW$:
$$ \sum_{C \in \calC} \calW(C) z_C = \val(g_0 \cup \rho, \calW).$$
\item {\bf Rounding:}Suppose $\vect, \vecz$ satisfy $\LLand_1(\rho)$. Then for every
  \maxand instance $\calW$:
  % \[\meancalc_\rho(\smooth(\vect), \calW) \geq
  % \frac{1}{2} \cdot \sum_{C \in \Act(\rho)}  \calW(C)z_C.\]
  \[\val(\rho, \calW) + \meancalc_\rho(\smooth(\vect), \calW) \geq
  \frac{1}{2} \cdot \sum_{C \in \calC}  \calW(C)z_C.\]
\end{enumerate}
\end{lemma}
\begin{proof}
  We begin with the first part.  For $v \in S_\rho$, define $t_{v} =
  \rho(v).$.  For $v \in V \setminus S_{\rho},$ define $t_{v} =
  g_0(v).$ For $C \in \calC$, define $z_C = 1$ if $C(g_0 \cup \rho) =
  1$, and define $z_C = 0$ otherwise.  It is easy to see that these
  $\vect, \vecz$ satisfies $\LLand_1(\rho)$, and that for every
  instance $\calW$:
$$ \sum_{C \in \calC} \calW(C) z_C = \val(g_0 \cup \rho, W).$$

Now we consider the second part. Let $\calW$ be any instance of
$\maxand$.  Let $p = \smooth(\vect)$. Let $g: V \setminus S_\rho \to
\{0,1\}$ be sampled as follows: independently for each $v \in V
\setminus S_\rho$, $g(v)$ is sampled from $\{0,1\}$ such that
$\av[g(v)] = p(v)$.  We have:
\begin{align}
\label{trevround}
  \val(\rho, \calW) + \meancalc_\rho(\smooth(\vect), \calW) & =
  \sum_{C \in \calC \setminus\Act(\rho)} \calW(C)C(\rho) +
  \E\left[\sum_{C \in \Act(\rho)} \calW(C)C(\rho \cup g) \right].
\end{align}

% $$\meancalc_\rho(\smooth(\vect), \calW) \defeq \E\left[\sum_{C \in \Act(\rho)} \calW(C)C(\rho \cup g) \right] \geq
% \frac{1}{2} \cdot \sum_{C \in \Act(\rho)} \calW(C) z_C.$$ 
We will now understand the two terms of the right hand side.

For $C \in \calC \setminus \Act(\rho),$ it is easy to verify that if
$z_C > 0,$ we must have $C(\rho) = 1.$
Thus:
$$  \sum_{C \in \calC \setminus\Act(\rho)} \calW(C)C(\rho) \geq  \sum_{C \in \calC \setminus\Act(\rho)} \calW(C) z_C.$$

To understand the second term, we have the following claim.
\begin{claim}
\label{claim:max2and:trev_rounding}
For $C \in \Act(\rho)$, 
$\E[C(\rho \cup g)]  \ge \frac{1}{2}\cdot z_C$.
\end{claim}
\begin{proof}
Suppose there are exactly $h$ variables in $C$ which are not in
$S_\rho$. We have $h\leq 2.$
\begin{align*}
\E[C(\rho \cup g)] =  \Pr[C \text{ is satisfied by }\rho \cup g] &=
  \left(\prod_{v\in C^+, v\in V\setminus S_\rho}\frac{1}{4} +
    \frac{t_v}{2}\right)\cdot\left(\prod_{v\in C^-, v\in V\setminus
      S_\rho}\frac{1}{4} + \frac{1-t_v}{2}\right) \\
  & \geq \left(\frac{1}{4} + \frac{z_C}{2}\right)^h \geq
  \left(\frac{1}{4} + \frac{z_C}{2}\right)^2 \geq \frac{z_C}{2}.
\end{align*}
\end{proof}

This claim implies that:
$$ \E\left[\sum_{C \in \Act(\rho)} \calW(C)C(\rho \cup g) \right] \geq \frac{1}{2} \sum_{C \in \calC} \calW(C) z_C.$$

Substituting back into Equation~\eqref{trevround}, we get the Lemma.
\end{proof}

% \begin{theorem}
% \label{theorem:trev_maxand_rounding}
% Given a feasible solution to the linear program $\LLand_1(\rho),$ if
% we round the fractional solution to a random assignment by
% independently sampling an assignment for each $v \in V$ such that
% $\Pr[ v = \true] = \frac{1}{4} + \frac{t_v}{2}$ for all $v \in
% V\setminus S_\rho$ and $\Pr[ v = \true] = t_v$ for $v\in S_\rho$, then
% the probability that an active constraint $C$ is satisfied is at least
% $\frac{z_C}{2}$.
% \end{theorem}

\begin{figure*}
\begin{tabularx}{\textwidth}{|X|}
\hline
 \vspace{-27pt}
\begin{center}
\[\begin{array}{rrllr}
& z_C &\leq & t_v & \forall C\in \calC, v\in C^+ \\
& z_C &\leq & 1 - t_v & \forall C\in \calC, v\in C^- \\
& 1\geq t_v &\geq & 0 & \forall v \in V\setminus S_\rho\\
& t_v & = & h_\rho(v) & \forall v \in S_\rho 
\end{array}\]
 \vspace{-25pt}
\end{center}
\\
\hline
\end{tabularx}
\caption{Linear inequalities $\LLand_1(\rho)$}
  \label{fig:max2and-lp1}
\end{figure*}

\begin{figure*}
\begin{tabularx}{\textwidth}{|X|}
\hline
\vspace{-27pt}
\begin{center}
\[\begin{array}{rrllr}
\sum_{C\in \calC} \calW_\ell(C)\cdot z_C  \geq c_\ell & \forall \ell\in [k]\\
\vect, \vecz  \text{ satisfy }  \LLand_1(\rho).&
\end{array}\]
 \vspace{-25pt}
\end{center}
\\
\hline
\end{tabularx}
\caption{Linear inequalities $\LLand_2(\rho)$}
  \label{fig:max2and-lp2}
\end{figure*}

\subsection{The Algorithm}

We now give our Pareto approximation algorithm for {\sc Max-2-AND} in Figure~\ref{fig:weighted-max2and}.

\begin{figure*}[h!]
\begin{longtabu}{|X|}
\hline
%\vspace{0mm}
{\bf Input}: $k$ instances of \maxand $\calW_1,\ldots,\calW_k$ on the
variable set $V,$ $\eps > 0$ and target objective values $c_1, \ldots, c_k.$\\
{\bf Output}: An assignment to $V$\\
{\bf Parameters:}  $\delta_0 = \frac{1}{10(k+1)}$, $\epsilon_0 = \eps$, $\gamma=\frac{\eps_0^2\delta_0}{16}$, $t = \ceil{\frac{20k^2}{\gamma}\log \frac{k}{\gamma}}$
%$t$ such that $t = \frac{32k^2}{\gamma}\cdot\log\left( \frac{16000tk}{\gamma}\right)$
\begin{enumerate}[itemsep=0mm]
\item Initialize tree $T$ to be an empty quaternary tree (i.e., just 1
  root node). Nodes of the tree will be indexed by strings in
  $(\{0,1\}^2)^*$.
\item With each node $\nu$ of the tree, we associate:
\begin{enumerate}[itemsep=0mm]
\vspace{-2mm}
\item A partial assignment $\rho_\nu.$
\item A special pair of variables $\calT_\nu^1,\calT_\nu^2 \in V \setminus S_{\rho_{\nu}}$.
\item A special instance $\Inst_\nu \in [k]$.
\item A collection of integers $\cnt_{\nu,1}, \ldots, \cnt_{\nu, k}$.
\item A trit representing whether the node $\nu$ is living, dead, or exhausted.
%\item A partial assignment $\rho_\nu,$ a special pair of variables $\calT_\nu^1,\calT_\nu^2 \in V \setminus S_{\rho_{\nu}},$ a special instance $\Inst_\nu \in [k],$ a collection of integers $\cnt_{\nu,1}, \ldots, \cnt_{\nu, k}$ and a trit representing whether the node $\nu$ is living, dead or exhausted.
\end{enumerate}
\end{enumerate}

\\ 
\hline
\end{longtabu}
\end{figure*}

%\Ssnote{Important note: Now $p$ is only defined on $V \setminus S.$}
\begin{figure*}
\begin{longtabu}{|X|}
\hline
%\vspace{0mm}

\begin{enumerate}[itemsep=0mm]
\item[3.] Initialize the root node $\nu_0$ to (1) $\rho_{\nu_0} \leftarrow
  (\emptyset, \emptyset)$, (2) $\forall\ell \in [k],
  \cnt_{\nu_0,\ell} \leftarrow 0$, (3) living.
\item[4.] While there is a living leaf $\nu$ of $T,$ do the following: 
\begin{enumerate}
\item Check the feasibility of linear inequalities $\LLand_2(\rho_\nu)$. 
\begin{enumerate}
\item If there is a feasible solution $\vect, \vecz$, then define $p_\nu :
  V\setminus S_{\rho_\nu} \to [0,1]$ as $p_\nu =\smooth(\vect)$.
\item If not, then declare $\nu$ to be dead and return to Step $4$. 
\end{enumerate}
\item For each $\ell \in [k]$,
compute $\varcalc_{\rho_\nu}(p_{\nu}, \calW_\ell)$ and $\meancalc_{\rho_\nu}(p_{\nu}, \calW_\ell)$.
\item If $\varcalc_{\rho_\nu}(p_{\nu}, \calW_\ell) \geq \delta_0
  \epsilon_0^2\cdot \meancalc_{\rho_\nu}(p_{\nu}, \calW_\ell)^2 $,
  then set $\flag_\ell \leftarrow \true$, else set $\flag_\ell
  \leftarrow \false$.
\item Choose the smallest $\ell \in [k]$, such that
$\cnt_\ell < t$ AND $\flag_\ell = \true$ (if any):
\begin{enumerate}
\item Find $x \in V \setminus S_{\rho_\nu}$ that maximizes
  $\activedeg_{\rho_\nu}(x, \calW_\ell).$ Note that it will satisfy
  $\activedeg_{\rho_\nu}(x, \calW_\ell) \geq \gamma \cdot
  \activedeg_{\rho_{\nu}}(\calW_\ell)$.

\item Among all the active constraints $C \in \calC$ such that $x \in C$ and $C
  \cap S_{\rho_\nu} = \emptyset,$ find the one that maximizes
  $\calW_\ell(C)$. Call this constraint $C^\star$. Let $y$ be the
  other variable contained in $C^\star$ (if there is no other
  variable, set $y = x$).

Set $\calT_\nu^1 \leftarrow x$ and $\calT_\nu^2 \leftarrow y$.
Set $\Inst_{\nu} \leftarrow \ell$.

\item Create four children of $\nu$, with labels $\nu b_1 b_2$ for each
  $b_1, b_2 \in \{0,1\}$ and set
\begin{itemize}
\item $\rho_{\nu b_1 b_2} \leftarrow (S_{\rho_\nu} \cup
  \{\calT_\nu^1,\calT_\nu^2\}, h^{b_1 b_2})$,
where $h^{b_1 b_2}$ extends $h_{\rho_{\nu}}$ by $h^{b_1 b_2
}(\calT_\nu^1) = b_1$ and $h^{b_1 b_2}(\calT_\nu^2) = b_2$.

\item $\forall \ell' \in [k]$ with $\ell' \neq \ell$, set $\cnt_{\nu
    b_1 b_2, \ell'} \leftarrow \cnt_{\nu, \ell'}$. Set $\cnt_{\nu b_1
    b_2, \ell} \leftarrow \cnt_{\nu, \ell}+1$.

\item Set $\nu b_1 b_2$ to be living.

\end{itemize}
\end{enumerate}
\item If no such $\ell$ exists, declare $\nu$ to be exhausted. 

\end{enumerate}
\end{enumerate}
\begin{enumerate}

\item[5.] Now every leaf of $T$ is either exhausted or dead. For each exhausted leaf $\nu$ of $T$:
\begin{enumerate}
\item 
\label{alg:max-and:step:random-assignment}
Let $g_{\nu} : V\setminus S_{\rho_\nu} \to \{0,1\}$ be a random
assignment where, for each $v \in V \setminus S_{\rho_\nu},$
$g_{\nu}(v)$ is sampled independently with $\E[g_{\nu}(v)] =
p_{\nu}(v)$. Note that $\E[g_{\nu}(v)] \in [\frac{1}{4}, \frac{3}{4}]$.
\item 
\label{alg:max-and:step:all-assignments}
For every assignment $h: S_{\rho_\nu} \to \{0,1\},$
% \begin{itemize}
%\item Set $f \leftarrow h \cup g_{\nu}$.
compute $\out_{h,g_\nu} \leftarrow \min_{\ell \in [k]} \frac{\val(h \cup g_{\nu},
    \calW_\ell)}{c_\ell}.$ If $c_\ell = 0$ for some $\ell \in [k],$
    we interpret $\frac{\val(h \cup g_{\nu},\calW_l)}{c_\ell}$ as $+\infty.$
%\end{itemize}
\end{enumerate}

\item[6.] Output the largest $\out_{h,g_\nu}$ seen, and the assignment $h \cup g_{\nu}$ that produced it.
\end{enumerate}
\\ 
\hline
\end{longtabu}
\caption{Algorithm {\algowand} for approximating weighted simultaneous
  \maxand}
  \label{fig:weighted-max2and}
\end{figure*}
\clearpage

%\Ssnote{Now, $p$ is only defined on $V \setminus S$ and it's denoted
%  as $p_\nu.$}
%\Ssnote{Need to go through the algorithm again.}

\subsection{Analysis}
Notice that the depth of the tree $T$ is at most $kt$, and that for every
$\nu$, we have that $|S_{\rho_{\nu}}| \leq 2kt$. This implies that the running time
is at most $2^{O(kt)} \cdot \poly(n)$.

Let $f^\star: V \to \{0,1\}$ be an assignment such that
$\val(f^\star, \calW_\ell) \geq c_\ell$ for each $\ell \in [k]$. Let $\nu^\star$ be the the unique leaf of the tree $T$ for which
$\rho_{\nu^\star}$ is consistent with $f^\star$. (This $\nu^\star$ can be
found as follows: start with $\nu$ equal to the root. Set $\nu$ to equal
the unique child of $\nu$ for which $\rho_\nu$ is consistent with $f^\star$, and repeat
until $\nu$ becomes a leaf. This leaf is $\nu^\star$).
Observe that since $f^\star$ is an assignment such that
$\val(f^\star, \calW_\ell) \ge c_\ell$ for every $\ell \in [k],$ by
picking $g_0 = f^\star|_{V \setminus S^\star}$ in part 1 of
Lemma~\ref{lemma:max2and:lp}, we know that $\LLand_2(\rho^\star)$ is
feasible, and hence $\nu^\star$ must be an exhausted leaf (and not
dead).\\
\\
Define $\rho^\star = \rho_{\nu^\star}$, $S^\star = S_{\rho^\star},$
$h^\star = h_{\rho^\star},$ and $p^\star = p_{\nu^\star}$.
At the completion of Step~4, if $\ell
\in [k]$ satisfies $\cnt_{\nu^\star, \ell} = t$, we call instance $\ell$ a {\em
  high variance} instance.  Otherwise we call instance $\ell$ a
\emph{low variance} instance.\\

\subsubsection{Low Variance Instances}
First we show that for the leaf $\nu^*$ in Step $5,$ combining the
partial assignment $h^\star$ with a random assignment $g_{\nu^\star}$
in step $5(b)$ is good for any low variance instances with high
probability. 
\begin{lemma}
\label{lemma:low_variance:maxand}
Let $\ell \in [k]$ be any low variance instance. For the leaf node
$\nu^\star$, let $g_{\nu^\star}$ be the random assignment sampled in
%Step~\ref{alg:max-and:step:random-assignment} 
Step 5.(a). of {\algowand}. Then with probability at least
$1-\delta_0,$ the assignment $f=h^\star \cup g_{\nu^\star}$ satisfies:
$$ \Pr_{g_{\nu^\star}} \left[ \val(f,\calW_\ell) \geq
  (\nicefrac{1}{2}-\nicefrac{\epsilon}{2})\cdot c_\ell \right] \geq 1-\delta_0.$$
\end{lemma}
\begin{proof}
  For every low variance instance $\ell$, we have that
  $\varcalc_{\rho_{\nu^\star}}(p^\star , \calW_\ell) < \delta_0\epsilon_0^2
  \cdot \meancalc_{\rho_{\nu^\star}}(p^\star , \calW_\ell)^2.$ Define $Y
  \defeq \val(\rho^\star \cup g_{\nu^\star}, \calW_\ell) -
  \val(\rho^\star,\calW_\ell)$.  By
  Lemma~\ref{lemma:max2and:truemeanvar}, we have $\Pr[ Y < (1-
  \epsilon_0) \E[Y] ] < \delta_0.$ 
Thus, with probability at least $1-\delta_0,$ we have,
\begin{align*}
  \val(f,\calW_\ell) &\geq \val(\rho^\star,\calW_\ell) +
  (1-\epsilon_0) \E[Y]\\
  &= \val(\rho^\star,\calW_\ell) + (1-\epsilon_0) \cdot \meancalc_{\rho^\star}(\smooth(\vect), \calW_\ell) \\
  &= (1-\epsilon_0) \cdot \left(\val(\rho^\star,\calW_\ell) +
    \meancalc_{\rho^\star}(\smooth(\vect), \calW_\ell) \right) \\
  &\geq \frac{1}{2}\cdot(1-\epsilon_0)\cdot \sum_{C \in \calC}
  \calW_\ell(C) \cdot z_C  \ge \frac{1}{2}\cdot(1-\epsilon_0) \cdot c_\ell\geq
  \left(\frac{1}{2} - \frac{\epsilon}{2}\right) \cdot c_\ell,
\end{align*}
where we have used the second part of
  Lemma~\ref{lemma:max2and:lp}.
% \begin{align*}
% \E[Y] &=  \meancalc_\rho(\smooth(\vect), \calW_\ell) \geq \frac{1}{2}\cdot \sum_{C \in \Act(\rho^\star)} \calW_\ell(C)\cdot z_C.
% \end{align*}
%
% Thus, when $\nu$ is taken to be $\nu^\star$ and  when
% $f_{|S^\star}$ is taken to be $h^\star$ in Step $5$, we have that for any low
% variance instance $\ell\in [k]$, $\val(f, \calW_\ell)$ will be at
% least $(\nicefrac{1}{2}-\nicefrac{\epsilon}{2}) \cdot c_{\ell}$ with probability at least $1-\delta_0$. 
\end{proof}
Next, we will consider a small perturbation of $h^\star$ which will
ensure that the algorithm performs well on high variance instances too.
We will ensure that this perturbation does not affect the
success on the low variance instances.
% by more than $\nicefrac{\eps}{2}\cdot\val(f,\calW_\ell)$ that we get at step $5(b)$ when we set $h=h^\star$.

\subsubsection{ High variance instances}
Fix a high variance instance $\ell$. Let $\nu$ be an ancestor of $\nu^\star$ with $\Inst_{\nu} = \ell$.
Define:
\begin{align*}
\activedeg_\nu \defeq \activedeg_{\rho_{\nu}}(\calT_\nu^1, \calW_\ell).
\end{align*}
Let $\calC_{\nu}$ be the set of all constraints $C$ containing
$\calT_\nu^1$ which are active given $\rho_\nu$.  We call a constraint
$C$ in $\calC_{\nu}$ a $\backward$ constraint if $C$ only involves
variables from $S_{\rho_{\nu}} \cup \{\calT_\nu^1\}$. Otherwise we
call $C$ in $\calC_{\nu}$ a $\forward$ constraint.  Let
$\calC_{\nu}^{\backward}$ and $\calC_{\nu}^{\forward}$ denote the sets
of these constraints.  Finally, we denote $\calC_{\nu}^{\out}$ the set
of binary constraints on $\calT_\nu^1$ and a variable from $V
\setminus S^\star$.

%contain define $\backward$ constraints to be those constraints containing $\Ver_{\nu}(1)$ which are active given $\rho_{\parent(\nu)}$ containing $\calT_\nu^1$ and some other variable in $S_\rho$ and $\forward$ constraints are the remaining active constraints containing $\calT_\nu^1$.

Define $\backward$ degree and $\forward$ degree  as follows:
\begin{align*}
\backward_{\nu} &\defeq \sum_{C \in \calC_{\nu}^\backward}\calW_{\ell}(C),\\
\forward_{\nu} &\defeq \sum_{C \in \calC_{\nu}^\forward}\calW_{\ell}(C).
\end{align*}
Note that:
$$\activedeg_\nu =  \backward_{\nu}  + \forward_{\nu}.$$
Now we consider variable $\calT_\nu^2$.  Let $\heaviest_\nu$ be the
total $\calW_\ell$ weight of all the constraints containing both
$\calT_\nu^1$ and $\calT_\nu^2$. Based on all this, we classify $\nu$
into one of three categories:
\begin{enumerate}
\item If $\backward_{\nu} \geq \frac{1}{2}
  \cdot\activedeg_\nu$, then we call $\nu$ a $\typeA$ node.
\item Otherwise, if $\heaviest_\nu \geq \frac{1}{100tk}
  \cdot\activedeg_\nu$, then we call $\nu$ a $\typeB$ node. In this
  case we have some $\calW_\ell$ constraint $C$ containing
  $\calT_\nu^1$ and $\calT_\nu^2$ with $\calW_\ell(C) \geq
  \frac{1}{1600 tk} \cdot \activedeg_\nu$.
\item Otherwise, we call $\nu$ a $\typeC$ node. In this case, for
  every $v \in V \setminus S_{\rho_\nu},$ the total weight of the
  constraints involving $v$ and $\calT_\nu^1,$ \emph{i.e.},
  $\activedeg_{\rho_\nu}(\calT_\nu^1 \cup v,\calW_\ell)$ is bounded by
  $\frac{1}{100tk} \cdot \activedeg_{\nu}.$ In particular, every
  constraint $C \in \calC_{\nu}^\forward$ must have $\calW_\ell(C) <
  \frac{1}{100tk} \cdot \activedeg_{\nu}$ . Since $|S^\star| \leq
  2tk$, the total weight of constraints containing $\calT_\nu^1$ and
  some variable in $S^\star\setminus S_{\rho_\nu}$ is at most
  $|S^\star\setminus S_{\rho_\nu}|\cdot \frac{1}{100tk} \cdot
  \activedeg_{\nu}$ which is at most $\frac{2}{100} \cdot
  \activedeg_{\nu}$. Hence we have:
\begin{align*}
  \sum_{C \in \calC_{\nu}^{\out}} \calW_\ell(C)  &= \forward_{\nu} - \left\{\substack{\text{total weight of constraints containing}\\ \text{$\calT_\nu^1$ and some variable in $S^\star\setminus S_{\rho_\nu}$ } }\right\}\\
&\geq \left(\frac{1}{2} - \frac{2}{100}\right) \activedeg_\nu  > \frac{1}{4} \cdot\activedeg_\nu.
\end{align*}
\end{enumerate}
For nodes $\nu$ which are $\typeC$, the variable $\calT_\nu^1$ has a
large fraction of its active degree coming from constraints between
$\calT_\nu^1$ and $V \setminus S^\star$.

For a partial assignment $g : V\setminus S^\star \to \{0,1\}$, we say that $g$
is $\Cgood$ for $\nu$ if there exists a setting of variable $\calT_\nu^1$ that satisfies at
least $\frac{1}{64}\cdot\activedeg_\nu$ weight amongst constraints containing variable $\calT_\nu^1$ and some other variable in
$V\setminus S^\star$.
%For us, the importance of the event $A_\nu(g)$ is this: if even $A_\nu(g)$ occurs, 
%then there is a setting of the variable $\calT_\nu^1$ such that a constant fraction of the
%constraints between $\calT_\nu^1$ and $V \setminus S^\star$ are satisfied.
The next lemma shows that for every $\typeC$ node $\nu$, with high
probability, the random assignment $g_{\nu^\star}: V
\setminus S^\star \to \{0,1\}$ is $\Cgood$ for $\nu$.
\begin{lemma}
\label{lemma:maxand:hoeffding}
Consider a $\typeC$ node $\nu.$ Suppose $g : V \setminus S^\star \to
\{0,1\}$ is a partial assignment obtained by independently sampling
$g(v)$ with $\E[g(v)] \in
[\nfrac{1}{4}, \nfrac{3}{4}]$ for each $v \in V\setminus S^\star.$ Then:
$$ \Pr_{g} [ g \mbox { is $\Cgood$ for $\nu$} ] \geq 1- 2\cdot e^{-tk/100}.$$
\end{lemma}
\begin{proof}
  Let $\ell = \Inst_\nu$.
% \Anote{Either delete the following paragraph or update the analysis with the notation $C(b\cup x)$}
% For each constraint $C \in \calC_\nu^{\out},$ for each assignment $b :
% \calT^1_\nu \rightarrow \{0,1\},$ and for each $x: {V \setminus S^\star}
% \to \{0, 1\},$ define $C(b\cup x) \in \{0,1\}$ as follows. $C(b\cup x)
% $ equals $1$ iff $C$ is satisfied by setting $v$ to $x(v)$ for each
% $v\in V\setminus S^\star$ and $w$ to $b(w)$ for each $w\in \calT_\nu.$

  For each constraint $C \in \calC_{\nu}^{\out}$ and each $g:
  \{0,1\}^{V \setminus S^\star} \to \{0,1\}$, define $Z_C^{(1)}(g),
  Z_C^{(0)}(g) \in \{0,1\}$ as follows. $Z_C^{(1)}(g)$ equals $1$ iff
  $C$ is satisfied by extending the assignment $g$ with $\calT_\nu^1
  \leftarrow 1.$ Similarly, $Z_C^{(0)}(g)$ equals $1$ iff $C$ is
  satisfied by extending the assignment $g$ with $\calT_\nu^1
  \leftarrow 0.$
%\Sknote{We need better language for this!}

% Define $\score^{b}: [q]^{V \setminus S^\star} \to \mathbb R$ by
% $$\score^{b}(g)  = \sum_{C \in \calC_{\calT_\nu}^\out} \calW_\ell(C) \cdot C(b\cup g) .$$
For $b =0,1,$ we define $\score^{(b)}: \{0,1\}^{V \setminus S^\star} \to
\mathbb R$ as follows:
$$\score^{(b)}(g)  \defeq \sum_{C \in \calC_{\nu}^\out} \calW_\ell(C) \cdot Z_C^{(b)}(g).$$
In words, $\score^{(b)}(g)$ is the total weight of constraints between $\calT_\nu^1$ and $V\setminus S^*$
satisfied by setting $\calT_{\nu}^1$ to $b$ and setting $V \setminus S^*$ according to $g$.
%$$\score^{(0)}(g)  \defeq \sum_{C \in \calC_{\nu}^\out} \calW_\ell(C) \cdot Z_C^{(0)}(g).$$

Note that since $g(v)$ is sampled independently for $v \in V \setminus
S^\star$ with $\E[g(v)] \in [\nfrac{1}{4}, \nfrac{3}{4}],$ we have
$\E_g[Z_C^{(1)}(g) + Z_C^{(0)}(g)] \geq \frac{1}{4}$. Thus:
\begin{align*}
\E_g[ \score^{(1)}(g)  + \score^{(0)}(g) ] &= \sum_{C \in \calC_{\nu}^\out} W_\ell(C) \E[Z_C^{(1)}(g)] + \sum_{C \in \calC_{\nu}^\out} W_\ell(C) \E[Z_C^{(0)}(g)]\\
&\geq \frac{1}{4} \sum_{C \in \calC_{\nu}^{\out}} \calW_\ell(C).
\end{align*}

So one of $\E[\score^{(1)}(g)]$ and $\E[\score^{(0)}(g)]$ is at least
$\frac{1}{8} \sum_{C \in \calC_{\nu}^{\out}} \calW_\ell(C) \geq
\frac{1}{32} \activedeg_\nu$.  Suppose it is $\E[\score^{(1)}(g)]$
(the other case is identical). We are going to use McDiarmid's
inequality to show the concentration of $\score^{(1)}(g)$ around its
mean\footnote{In this case we could have simply used a Hoeffding-like inequality,
but later when we handle larger-width constraints we will truly use the added generality of McDiarmid's inequality.}.

Since $\nu$ is $\typeC$, we know that for every vertex $v \in V
\setminus S^\star,$ changing $g$ on just $v$ can change the value of
$\score^{(1)}(g)$ by at most $c_v \defeq
\activedeg_{\rho_\nu}(\calT_\nu^1\cup v, \calW_\ell) \leq \frac{1}{100tk} \cdot
\activedeg_{\nu}.$
% where $c_v$ is defined to be the sum of
% $\calW_\ell(C)$ over all constraints $C$ depending on both $v$ and
% $\calT_\nu^1$.  
Thus by McDiarmid's inequality (Lemma~\ref{lemma:mcdiarmid}),
\begin{align*}
\Pr_{g} [ g \mbox { is not $\Cgood$ for $\nu$} ] &\leq \Pr_{g} \left[ \score^{(1)}(g) < \frac{1}{64}\cdot\activedeg_\nu\right]\\
& \leq \Pr_{g}\left[ | \score^{(1)}(g) - \E_{g}[\score^{(1)}(g)] | > \frac{1}{64}\cdot\activedeg_\nu\right]\\
& \leq 2\cdot\exp\left( \frac{-2\cdot\activedeg_\nu^2}{(64)^2 \sum_{v\in  V\setminus S^\star} c_v^2 }\right)\\
& \leq 2\cdot\exp\left( \frac{-2\cdot\activedeg_\nu^2}{(64)^2\cdot(\max_v{c_v})\cdot\sum_{v\in  V\setminus S^\star} c_v}\right)\\
& \leq 2\cdot\exp\left( \frac{-2\cdot\activedeg_\nu^2}{(64)^2\cdot (\max_v{c_v})\cdot \activedeg_\nu}\right)\\
& \leq 2\cdot\exp\left( \frac{-2\cdot\activedeg_\nu}{(64)^2 \cdot(\frac{\activedeg_\nu}{100tk})}\right) \leq 2\cdot\exp\left( \frac{-200tk}{(64)^2}\right) \leq  2\cdot\exp\left( \frac{-tk}{100}\right).
\end{align*}
\end{proof}

For a high variance instance $\ell$,
let $\nu^\ell_1, \ldots, \nu^\ell_t$ be the sequence of $t$ nodes with $\Inst_{\nu} = \ell$
which lie on the path from the root to $\nu^\star$.
Set $\finalW_\ell = \activedeg_{\rho^\star}(\calW_{\ell})$ (in words:
this is the active degree left over in instance $\ell$ after the restriction $\rho^\star$).
\begin{lemma}
\label{lem:potential}
For every high variance instance $\ell \in [k]$ and for each $i \leq [t/2]$,
$$\activedeg_{\nu^\ell_i} \geq \gamma\cdot(1-\gamma)^{-t/2} \cdot \finalW_\ell \geq 1600tk\cdot \finalW_\ell.$$
\end{lemma}
\begin{proof}
Fix a high variance instance $\ell\in [k]$. Note that $b_i = \activedeg_{\rho_{\nu_i^\ell}}(\calW_\ell)$ decreases as
$i$ increases. The main observation is that
\begin{enumerate}
\item $b_{i+1} \leq (1-\gamma) \cdot b_i$.
\item $\activedeg_{\nu_i^\ell}  \geq \gamma b_i.$
\end{enumerate}
Thus for all $\nu_i^\ell$ with $i\in \{1, \ldots, \nfrac{t}{2}\}$, we
have $\activedeg_{\nu^\ell_i} \geq \gamma\cdot(1-\gamma)^{-t/2} \cdot
\finalW_\ell$ and also the choice of parameters implies for those
$\nu_i^\ell$ $\activedeg_{\nu_i^\ell}$ is at least
$1600tk\cdot\finalW_\ell$.
\end{proof}

\subsubsection{Putting everything together}
We now show that when $\nu$ is taken to equal $\nu^\star$ in
Step $5$, then with high probability over the choice of $g$ in 
Step $5(a)$ there is a setting of $h$ in Step $5(b)$
such that $\forall \ell \in [k], \val(h \cup g_{\nu^\star},
\calW_\ell) \geq (\frac{1}{2} - \epsilon) \cdot c_\ell.$ 
\begin{theorem}
\label{thm:maxandapprox}
Suppose the algorithm {\algowand} is given as inputs $\eps >
0,$ $k$ simultaneous weighted \maxand instances $\calW_1,\ldots,
\calW_k$ on $n$ variables, and target objective value $c_1,
\ldots,c_k$ with the guarantee that there exists an assignment
$f^\star$ such that for each $\ell \in [k],$ we have $\val(f^\star,
\calW_\ell) \ge c_\ell.$ Then, the algorithm runs in
 $2^{O(\nfrac{k^4}{\eps^2}\log(\nfrac{k}{\eps}))} \cdot \poly(n)$ time, and with
probability at least 0.9, outputs an assignment $f$ such that for each
$\ell \in [k],$ we have, $\val( f,\calW_\ell) \ge \left( \frac{1}{2}
  -\epsilon\right) \cdot c_\ell.$
\end{theorem}
\begin{proof}
  Consider the case when $\nu$ is taken to equal $\nu^\star$ in Step
  $5.$ By Lemma~\ref{lemma:low_variance:maxand}, with probability at
  least $1 -k\delta_0$ over the choice random choices of
  $g_{\nu^\star}$, we have that for {\em every} low variance instance
  $\ell \in [k]$, $\val(h^\star \cup g_{\nu^\star}, \calW_\ell) \geq
  (\frac{1}{2} - \frac{\epsilon}{2}) \cdot c_{\ell}$.  By
  Lemma~\ref{lemma:maxand:hoeffding} and a union bound, with
  probability at least $1 - \frac{t}{2}\cdot k\cdot 2e^{-tk/100} \geq
  1 - \delta_0$ over the choice of $g_{\nu^\star}$, for every high
  variance instance $\ell$ and for every $\typeC$ node $\nu_i^\ell$,
  $i\in [t/2]$, we have that $g_{\nu^\star}$ is $\Cgood$ for
  $\nu_i^\ell$. Thus with probability at least $1- (k+1)\delta_0$,
  both these events occur.  Henceforth we assume that both these
  events occur in Step $5(a)$ of the algorithm.

% We now show that there is a partial assignment $h: S^\star \to \{0,1\}$,
% such that when $h$ is considered in Step $5(b)$ of the algorithm, we
% have $\min_{ \ell \in [k] } \val( h \cup g, \calW_\ell) \geq
% \left(\frac{1}{2} - \epsilon\right) \cdot \opt$.
%
%\begin{lemma}
%\label{lemma:perturb_maxand}
%Let $h^{'} : S^\star \to \{0,1\}$ be any partial assignment.
%Suppose $g : V \setminus S^\star \to \{0,1\}$ is such that for every high variance instance $\ell \in [k]$ and for all $i\in [t/2]$ the event $A_{\nu_i^\ell}(g)$ occurred then there exists a fix perturbation $h$ of $h^{'}$ with the following property ( Let $f=h\cup g$ and $f'=h^{'}\cup g$):

Our next goal is to show that there exists a partial assignment $h:S^\star
\to \{0,1\}$ such that:
\begin{enumerate}
\item \label{item:max-and:proof:guarantees:1} For every instance $\ell \in [k]$, $\val(h \cup g_{\nu^\star} ,\calW_\ell) \geq \left( 1 - \frac{\epsilon}{2} \right) \cdot \val(h^\star \cup g_{\nu^\star},\calW_\ell)$.
\item \label{item:max-and:proof:guarantees:2} Moreover, for every high variance instance $\ell\in [k]$, $\val(h \cup
  g_{\nu^\star},\calW_\ell) \geq \left( 1 - \frac{\epsilon}{2} \right)  \cdot \finalW_\ell$.
\end{enumerate}
Before giving a proof of the existence of such an $h$, we show that
this completes the proof of the theorem. We claim that when
the partial assignment $h$ guaranteed above is considered in the
Step~$5(b)$ in the algorithm, we obtain an assignment with the
required approximation guarantees. 

 For every low variance
instance $\ell \in [k],$ since we started with $\val(h^\star \cup g_{\nu^\star},
\calW_\ell) \geq (\frac{1}{2} - \frac{\epsilon}{2}) \cdot c_\ell,$
property~\ref{item:max-and:proof:guarantees:1} above implies that
every low variance instance $\val(h \cup g_{\nu^\star}) \ge (\frac{1}{2} -
\epsilon) \cdot c_\ell.$ For every high variance instance $\ell \in
[k],$ since $h^\star = f^\star|_S,$
\[\val(h^\star \cup g_{\nu^\star}, \calW_\ell) \ge \val(f^\star, \calW_\ell) -
\activedeg_{\rho^\star}(\calW_\ell) \ge c_\ell - \finalW_\ell.\]
Combining this with properties~\ref{item:max-and:proof:guarantees:1}
and \ref{item:max-and:proof:guarantees:2} above, we get,
\[ \val(h \cup g_{\nu^\star}, \calW_\ell) \ge \left( 1 - \nfrac{\epsilon}{2}
\right) \cdot \max\{c_\ell - \finalW_\ell, \finalW_\ell\} \ge
\nfrac{1}{2} \cdot \left( 1 - \nfrac{\epsilon}{2} \right) \cdot c_\ell.\]
Thus, for all instances $\ell \in [k]$, we get $\val(h \cup g_{\nu^\star}) \ge
\left(\nfrac{1}{2}-\eps\right) \cdot c_{\ell}.$

Now, it remains to show the existence of such an $h$ by giving a
procedure for constructing $h$ by perturbing $h^\star$ (Note that this
procedure is only part of the analysis). For nodes $\nu, \nu'$ in the
tree, let us write $\nu \prec \nu'$ if $\nu$ is an ancestor of $\nu'$,
and we also say that $\nu'$ is ``deeper" than $\nu$.

\paragraph{Constructing $h$:}
\begin{enumerate}
\item Initialize $H \subseteq [k]$ to be the set of high variance instances.
\item Let $N_0 = \{ \nu^{\ell}_i \mid \ell \in H, i \in [t/2]\}$.
  Note that $N$ is a chain in the tree (since all the elements of $N$
  are ancestors of $\nu^\star$). Since every $\nu \in N$ is an
  ancestor of $\nu^\star$, we have $h_{\rho_\nu} = h^\star|_{S_{\rho_\nu}}.$
\item Initialize $D = \emptyset$, $N= N_0$, $h = h^\star$.
\item During the procedure, we will be changing the assignment $h$,
and removing elements from $N$. We will always maintain the
following two invariants:
\begin{itemize}
\item $|N| > \frac{t}{4}$.
\item For every $\nu \in N$, $h|_{S_{\rho_\nu}} =
  h^\star|_{S_{\rho_{\nu}}}$.
\end{itemize}
\item While $|D| \neq |H|$ do:
\begin{enumerate}
\item Let $$B = \left\{ v \in V \mid \exists \ell \in [k] \textrm{ with }
  \sum_{C \in \calC, C \owns v} \calW_{\ell}(C)\cdot C(h\cup g_{\nu^\star}) \geq
  \frac{\epsilon}{4k} \val(h\cup g_{\nu^\star}, \calW_\ell) \right\}.$$

Note that $|B| \leq \frac{8k^2}{\epsilon} < \frac{t}{8}$. 
\item Let $\nu \in N$ be the deepest element of $N$ for which:
$\{\calT_\nu^1, \calT_\nu^2 \} \cap B = \emptyset.$

Such a $\nu$ exists because:
\begin{itemize}
\item $|N| > \frac{t}{4} >  |B|$, and
\item there are at most $|B|$ nodes $\nu$ for which 
$\{\calT_\nu^1, \calT_\nu^2 \} \cap B \neq \emptyset$ (since
$\{\calT_\nu^1, \calT_\nu^2\}$ are all disjoint for distinct $\nu$).
\end{itemize}

\item Let $\ell \in H, i \in [t/2]$ be such that $\nu = \nu_i^\ell$. Let $x = \calT_\nu^1$ and $y =
  \calT_\nu^2$.  Let $\rho = \rho_\nu$. We will now see a way of
  modifying the values of $h(x)$ and $h(y)$ to guarantee that $\val(h
  \cup g_{\nu^\star}, \calW_\ell) \geq \finalW_\ell$. The procedure
  depends on whether $\nu$ is $\typeA$, $\typeB,$ or $\typeC$.
\begin{enumerate}
\item If $\nu$ is $\typeA$, then we know that $\backward_{\nu} \geq
  \frac{1}{2}\cdot\activedeg_\nu \geq 2 \cdot \finalW_\ell$.

The second invariant tells us that $\rho = h^\star|_{S_{\rho}} = h|_{S_{\rho}}$.
Thus we have:
\begin{align*}
\backward_{\nu} &= \sum_{C \in \calC^\backward_{\nu}} \calW_\ell(C)\\
&= \sum_{C \subseteq S_{\rho} \cup \{x\}, C \ni x, C \in \Act(\rho)} \calW_\ell(C)\\
&= \sum_{C \subseteq S_{\rho} \cup \{x\}, C \ni x, C \in \Act(h|_{S_{\rho}})} \calW_\ell(C).
\end{align*}
This implies that we can choose a setting of $h(x) \in \{0,1\}$ such that 
the total sum of weights of those constraints containing $x$ which
are satisfied by $h$ is:
\begin{align*}
 \sum_{C \subseteq S_{\rho} \cup \{x\}, C \ni x, C \in \Act(h|_{S_\rho})} \calW_\ell(C) C(h) &\geq \frac{1}{2} \sum_{C \subseteq S_{\rho} \cup \{x\}, C \ni x, C \in \Act(h|_{S_\rho})} \calW_\ell(C)\\
&= \frac{1}{2} \cdot \backward_{\nu}\\
&\geq \frac{1}{4}\cdot \activedeg_\nu\\
&\geq \finalW_\ell,   \mbox{\quad(by Lemma~\ref{lem:potential})}
\end{align*}
where the $\frac{1}{2}$ in the first inequality is because the
variable can appear as a {\em positive} literal or a {\em negative} literal in
those $\backward$ constraints. In particular, after making this
change, we have $\val(h \cup g_{\nu^\star}, \calW_\ell) \geq \finalW_\ell$.

\item If $\nu$ is $\typeB$, then we know that some constraint $C$
  containing $x$ and $y$ has $\calW_\ell(C) \geq \frac{1}{1600 tk}
  \cdot \activedeg_\nu \geq \finalW_\ell $. Thus we may choose
  settings for $h(x), h(y) \in \{0,1\}$ such that $C(h) = 1$.  Thus,
  after making this assignment to $h(x)$ and $h(y)$, we have $\val(h
  \cup g, \calW_\ell) \geq \finalW_\ell$.

\item If $\nu$ is $\typeC$, since $g_{\nu^\star}$ is $\Cgood$ for $\nu$, we can
  choose a setting of $h(x)$ so that the total weight of satisfied
  constraints in $\calW_\ell$ between $x$ and $V \setminus S^\star$ is
  at least $\frac{1}{64}\cdot \activedeg_\nu \geq \finalW_\ell.$ After
  this change, we have $\val(h \cup g_{\nu^\star}, \calW_\ell) \geq \finalW_\ell$.
\end{enumerate}

In all the above 3 cases, we only changed the value of $h$ at the variables
$x, y$. Since $\{x, y\} \cap B = \emptyset$,
we have that for every $j \in [k]$,
the new value $\val( h \cup g_{\nu^\star}, \calW_j)$ is at least $\left(1 - \frac{\epsilon}{2k} \right)$
times the old value $\val( h \cup g_{\nu^\star}, \calW_j)$.

\item Set $D = D \cup \{\ell\}$.
\item Set $N = \{ \nu^{\ell}_i \mid \ell \in H\setminus D, i \leq [t/2], \nu^\ell_i \prec \nu\}$.

  Observe that $|N|$ decreases in size by at most $ \frac{t}{2} +
  |B|$.  Thus, if $D \neq H$, we have
\begin{align*}
|N| &\geq |N_0| - |D| \cdot \frac{t}{2} - |D||B| \\
& = |H|\cdot\frac{t}{2} - |D| \cdot \frac{t}{2} - |D||B|\\
& \geq \frac{t}{2} - k|B| > \frac{t}{4}
\end{align*}

Also observe that we only changed the values of $h$ at
the variables ${\calT_\nu^1}$ and ${\calT_\nu^2}$. 
Thus for all $\nu' \preceq \nu$, we still have the property that 
$h|_{S_{\rho_{\nu'}}} = h^\star|_{S_{\rho_{\nu'}}}$.
\end{enumerate}
\end{enumerate}

For each high variance instance $\ell \in [k],$ in the iteration where
$\ell$ gets added to the set $D,$ the procedure ensures that at the
end of the iteration $\val(h \cup g_{\nu^\star}, \calW_\ell) \ge \finalW_\ell.$

Moreover, at each step we reduced the value of each $\val(h\cup g_{\nu^\star},
\calW_\ell)$ by at most $\frac{\epsilon}{2k}$ fraction of its previous
value. Thus, at the end of the procedure, for every $\ell \in [k],$
the value has decreased at most by a multiplicative factor of
$\left(1-\frac{\eps}{2k}\right)^k \ge \left(1-\frac{\eps}{2}\right).$
Thus, for every $\ell \in [k],$ we get $\val(h \cup g_{\nu^\star}, \calW_\ell) \ge
\left(1-\frac{\eps}{2}\right) \cdot \val(h^\star \cup g_{\nu^\star}, \calW_\ell),$
and for every high variance instance $\ell \in [k]$, we have
$\val(h\cup g_{\nu^\star}, \calW_\ell) \geq \left(1-\frac{\eps}{2}\right) \cdot
\finalW_\ell$. This proves the two properties of $h$ that we set out
to prove.

{\bf Running time : } Running time of the algorithm is
$2^{O(kt)}\cdot\poly(n)$ which is $2^{O(\nfrac{k^4}{\eps^2}\log
(\nfrac{k}{\eps^2}))}\cdot\poly(n).$
\end{proof}

%\clearpage

%%% This are comments that I (Sushant) need for working with emacs
%%% Local Variables: 
%%% mode: latex
%%% TeX-master: "new-simopt"
%%% End: 

\section{Simultaneous \maxcspwq}
\label{section:maxcsp}

In this section, we give our simultaneous approximation algorithm for $\maxcspwq$,
 and thus prove Theorem~\ref{theorem:results:maxcsp}.

\subsection{Reduction to simple constraints}
\label{sec:reducsimple}
For the problem \maxcspwq, $\calC$ is the set of all possible $q$-ary
constraints on $V$ with arity at most $w,$ \emph{i.e.}, each
constraint is of the form $C_f : [q]^{T} \rightarrow \{0,1\}$
depending only on the values of variables in an ordered tuple $T
\subseteq V$ with $|T| \leq w$. As a first step (mainly to simplify
notation), we give a simple approximation preserving reduction which
replaces $\calC$ with a smaller set of constraints. We will then
present our main algorithm

Define a {\em $w$-term} to be a contraint $C$ on exactly $w$ variables
which has exactly $1$ satisfying assignment in $[q]^w$, \emph{e.g.}
$(x_1 = 1) \wedge (x_2 = 7) \wedge \ldots \wedge (x_w = q-3).$ An
instance of the \maxconj problem is one where the
set of constraints $\calC$ is the set of all $w$-terms. We now use the
following lemma from~\cite{Trevisan98} that allows us to reduce a
\maxcspwq instance to a \maxconj instance.
\begin{lemma}[\cite{Trevisan98}]
  Given an instance $\calW_1$ of \maxcspwq, we can find a instance
  $\calW_2$ of {\sc Max-$w$-Conj-Sat$_q$} on the same set of
  variables, and a constant $\beta > 0$ such that for every
  assignment $f$, $\val(f, \calW_2) = \beta \cdot \val(f, \calW_1).$
  % the weight of satisfied constraints is the same in $\calW_1$ and
  % $\calW_2.$
\end{lemma}
\begin{proof}
  Given an instance $\calW_1$ of \maxcspwq, consider a constraint $C \in
  \calC$ with weight $\calW_1(C).$ We can assume without loss of
  generality that the arity of $C$ is exactly $w,$ and it depends on
  variables $x_1,\ldots ,x_w.$ For each assignment in $[q]^k$ that
  satisfies $C,$ we create a {\sc $w$-Conj-Sat$_q$} clause that is
  satisfied only for that assignment, and give it weight $\calW_1(C).$
  \emph{e.g.} If $C$ was satisfied by $x_1 = \ldots = x_w = 2,$ we
  create the clause $ (x_1=2) \wedge \ldots \wedge (x_w=2)$ with
  weight $\calW_1(C).$ It is easy to see that for every assignment to
  $x_1, \ldots, x_n,$ the weight of constraints satisfied in the new
  instance is the same as the weight of the constraints satisfied in
  the {\sc Max-$w$-Conj-Sat$_q$} instance created. Define $\beta$ to
  be the sum of weights of all the constraints in the new instance,
  then $\calW_2$ is obtained by multiplying the weight of all the
  constraints in the new instance by $\nfrac{1}{\beta}$ (to make sure
  they sum up to 1).
\end{proof}
Note that the scaling factor $\beta$ in the lemma above is immaterial
since we will give an algorithm with Pareto approximation guarantee.
%\Ssnote{@Swastik: Please read this part to see if it's okay.}

We say $(v,i)\in C$ if $v\in C$ and $v=i$ is in the satisfying
assignment of $C_f$. By abuse of notation, we say for a set of
variables $T,$ $T\subseteq C$ if for all $v\in T,$ there exists $i\in
[q],$ such that $(v,i)\in C.$

\subsection{Random Assignments}
\label{sec:maxcspwq:random}
% \Anote{Major change TODO 2:  Define $\varcalc_\rho(p, \calW) $ and $\meancalc_\rho(p, \calW)$ with partial assignment $\rho$ everywhere (where for a partial assignment $\rho$, $p : V\setminus S_{\rho} \rightarrow [0,1]$)}
In this section, we state and prove a lemma that gives a sufficient
condition for the value of a $\maxconj$ to be highly concentrated
under independent random assignments to the variables. Let $\Dist(q)$
denote the set of all probability distributions on the set $[q].$ For
a distribution $p \in \Dist(q)$ and $i \in q,$ we use $p_i$ to denote
the probability $i$ in the distribution $p.$
% \Anote{TODO 3 : Define $\Dist(q)$}
Let $\rho$ be a partial assignment.  Let $p : V \setminus S_\rho \to
\Dist(q)$ be such that $p(v)_i \geq \frac{1}{qw}$ for all $v \in V
\setminus S_{\rho}$ and all $i\in [q]$. Let $g: V\setminus S_\rho \to
[q]$ be a random assignment obtained by sampling $g(v)$ for each $v$
independently according to the distribution $p(v)$.

Define the random variable
\[Y \defeq \val(\rho \cup g , \calW) - \val(\rho, \calW) = \sum_{C \in
  \Act(\rho)} \calW(C) C(\rho \cup g).\]
 The random variable $Y$ measures the contribution of active
constraints to $\val(\rho \cup g, \calW).$ Note that the two
quantities $\E[Y]$ and $\Var[Y]$ can be computed efficiently given $p$.  We denote
these by $\meancalc_\rho(p, \calW)$ and $\varcalc_\rho(p, \calW)$. The
following lemma is a generalization of Lemma~\ref{lemma:max2and:truemeanvar}.
\begin{lemma}
\label{lemma:maxcsp:truemeanvar}
Let $p, g, Y$ be as above.
\begin{enumerate}
\item \label{lemma:truemeanvar:lowvar} If $\varcalc_\rho(p, \calW) < \delta_0 \epsilon_0^2 \cdot \meancalc_\rho(p, \calW)^2$ then $\Pr[Y <  (1-\epsilon_0) \E[Y] ] < \delta_0$.
\item \label{lemma:truemeanvar:highvar} If $\varcalc_\rho(p, \calW) \ge \delta_0\epsilon_0^2 \cdot \meancalc_\rho(p, \calW)^2 $, then there exists $v \in V\setminus S_{\rho}$ such that 
$$\activedeg_\rho(v, \calW) \geq \frac{\epsilon_0^2\delta_0}{{w^2(qw)}^{w}}\cdot \activedeg_{\rho}( \calW).$$
\end{enumerate}
%\Anote{is this tight?}
\end{lemma}
\begin{proof}
Item~\ref{lemma:truemeanvar:lowvar} of the lemma follows immediately from Chebyshev's inequality.  
%Case 2 is essentially identical to Lemma~\ref{lemma:uvar_lmean_relation}, and we omit the proof. (Note that $\meancalc_\rho(g, \calW)$ can be lower bounded by $\frac{1}{{(qw)}^{w}}\cdot\sum_{C \in \Act(\rho)} \calW(C).)$
We now prove Item $2$. First note that for every active constraint $C$ given $\rho,$ $\E[C(\rho \cup g)] \geq \frac{1}{(qw)^w}$ (this follows
from our hypothesis that $p(v)_i \geq \frac{1}{qw}$ for each $v \in V \setminus S_\rho$ and each $i \in [q]$).

We first bound $\meancalc_\rho(p, \calW)$ and $\varcalc_\rho(p,
\calW)$ in terms of the weights of active constraints:
\begin{align*}
\meancalc_\rho(p, \calW) & = \E[Y] = \E\left[\sum_{C \in \Act(\rho)} \calW(C) \cdot C(\rho \cup g)\right]\\
&= \sum_{C \in \Act(\rho)} \calW(C) \cdot\E[C(\rho \cup g)] \geq \sum_{C \in \Act(\rho)} \calW(C) \cdot \frac{1}{(qw)^w}\\
& = \frac{1}{(qw)^w} \sum_{C \in \Act(\rho)} \calW(C)
\end{align*}
%Similarly, we can upper bound $\varcalc_\rho(p, \calW)$ as follows:
\begin{align*}
\varcalc_\rho(p, \calW) &= \Var[Y] = \Var \left[ \sum_{C \in \Act(\rho)} \calW(C) \cdot C(\rho \cup g)\right]\\
& = \sum_{C_1,C_2 \in \Act(\rho)} \calW(C_1)\calW(C_2) \cdot(\E[C_1(\rho \cup g)C_2(\rho \cup g)] - \E[C_1(\rho \cup g)]\E[C_2(\rho \cup g)])\\
& \leq  \sum_{C_1\sim_\rho C_2} \calW(C_1)\calW(C_2)\cdot\E[C_1(\rho \cup g)] \\
& =  \sum_{C_1\in \Act(\rho)}\calW(C_1)\E[C_1(\rho \cup g)]\cdot\sum_{C_2\sim_\rho C_1} \calW(C_2) \\
& \leq  \sum_{C_1\in \Act(\rho)}\calW(C_1)\E[C_1(\rho \cup g)]\cdot\max_{C\in \Act(\rho)}\sum_{C_2\sim_\rho C} \calW(C_2)\\
& = \meancalc_\rho(p, \calW) \cdot \max_{C\in \Act(\rho)}\sum_{C_2\sim_\rho C} \calW(C_2).
\end{align*}
Hence, if the condition in case~\ref{lemma:truemeanvar:highvar} is true then it follows that,
\begin{align*}
  \max_{C\in \Act(\rho)}\sum_{C_2\sim_\rho C} \calW(C_2) &\geq
  \frac{\varcalc_\rho(p, \calW)}{\meancalc_\rho(p, \calW)} \geq
  \frac{\delta_0\epsilon_0^2}{(qw)^{w}}\cdot\sum_{C \in \Act(\rho)}
  \calW(C).
\end{align*}
We now relate these quantities to active degrees.
\begin{align*}
  \activedeg_\rho(\calW) &= \sum_{v\in V\setminus S_\rho}
\activedeg_\rho(v,\calW) = \sum_{v\in V\setminus S_\rho} \sum_{C\in
\Act(\rho), C\owns v} \calW_\ell(C)\\
  &= \sum_{C\in \Act(\rho)} \sum_{v\in C, v\in V\setminus S_\rho}
\calW_\ell(C) \leq \sum_{C\in \Act(\rho)} w\cdot \calW_\ell(C)\\
  &= w \sum_{C\in \Act(\rho)} \calW_\ell(C) 
\end{align*}
This means that there is an active constraint $C$, such that
$$ \sum_{C_2\sim_\rho C} \calW(C_2) \geq \frac{\delta_0\epsilon_0^2}{(qw)^{w}}\cdot\frac{1}{w}\activedeg_\rho(\calW)$$
Since $C$ is an active constraint and $|C\cap V\setminus S_\rho| \leq
w$, there is some variable $v \in C\cap V\setminus S_\rho,$ such that
$$\activedeg_\rho(v, \calW) = \sum_{C_2 \in \Act(\rho),\ C_2 \owns v}  \calW(C_2) \geq \frac{1}{w}
\sum_{C_2\sim_\rho C} \calW(C_2) \geq
\frac{\epsilon_0^2\delta_0}{{w^2(qw)}^{w}}\cdot \activedeg_{\rho}(
\calW).$$ as required.
\end{proof}

\subsection{LP relaxations}
%\Ssnote{@Swastik: read this part}
Our algorithm will use the Linear Programming relaxation for \maxconj
from the work of Trevisan~\cite{Trevisan98} (actually, a simple 
generalization to $q$-ary alphabets). The first LP, $\LL_1(\rho),$
described in Figure~\ref{fig:maxcsp-lp1}, describes the set of all
feasible solutions for the relaxation, consistent with the partial
assignment $\rho.$ Given a set of target values $(c_\ell)_{\ell \in [k]},$ the second
LP, $\LL_2(\rho)$ describes the set of feasible solutions to
$\LL_1(\rho)$ that achieve the required objective values.

%Given a partial assignment $\rho$, we can write the linear program
%for simultaneous $\maxcspwq$ as shown in
%figure~\ref{fig:maxcsp-lp}. Let $\lpopt_\rho$ denotes the optimal
%value of this linear program. Let $\lp_\rho : V\setminus S_\rho
%\times [q] \to [0,1]$ denotes the optimal LP assignment to variables
%in $V\setminus S_\rho \times [q]$ i.e. $\lp_\rho(v,i) = t_{v,i}$. We
%have a following claim from [Trev] \Anote{Slightly generalized to
%q-ary CSPs}

For $\vect$, $\vecz$ satisfying linear constraints $\LL_1(\rho)$, let
$\smooth(\vect)$ denote the map $p: V \setminus S_\rho \to \Dist(q)$
with $p(v)_i = \frac{w-1}{qw} + \frac{t_{v,i}}{w}$. The following
theorem from~\cite{Trevisan98} provides an algorithm to round this
feasible solution to obtain a good integral assignment.
\begin{lemma}
\label{lemma:maxcsp:lp}
Let $\rho$ be a partial assignment.
\begin{enumerate}
\item {\bf Relaxation: }For every $g_0 : V \setminus S_\rho \to [q]$, there exist
$\vect$, $\vecz$ satisfying $\LL_1(\rho)$ such that for every
$\maxconj$ instance $\calW$:
$$ \sum_{C \in \calC} \calW(C) z_C = \val(g_0 \cup \rho, W).$$
\item {\bf Rounding: }Suppose $\vect, \vecz$ satisfy $\LL_1(\rho)$. Then for every
  $\maxconj$ instance $\calW$:
  \[\val(\rho, \calW) + \meancalc_\rho(\smooth(\vect), \calW) \geq
  \frac{1}{q^{w-1}} \cdot \sum_{C \in \calC} z_C \calW(C).\]
\end{enumerate}
\end{lemma}
\begin{proof}
  We begin with the first part.  For $v \in S_\rho$, $i \in [q]$,
  define $t_{v,i} = 1$ if $\rho(v)= i$, and define $t_{v,i} = 0$
  otherwise.  For $v \in V \setminus S_{\rho}, i \in [q]$, define
  $t_{v,i} = 1$ if $g_0(v) = i$, and define $t_{v,i} = 0$ otherwise.
  For $C \in \calC$, define $z_C = 1$ if $C(g_0 \cup \rho) = 1$, and
  define $z_C = 0$ otherwise.  It is easy to see that these $\vect,
  \vecz$ satisfies $\LL_1(\rho)$, and that for every instance $\calW$:
$$ \sum_{C \in \calC} \calW(C) z_C = \val(g_0 \cup \rho, W).$$

Now we consider the second part. Let $\calW$ be any instance of
$\maxconj$.  Let $p = \smooth(t)$. Let $g: V \setminus S_\rho \to
[q]$ be sampled as follows: independently for each $v \in V \setminus
S_\rho$, $g(v)$ is sampled from the distribution $p(v)$.  We need to show that:
$$\sum_{C \not\in \Act(\rho)}  \calW(C) C(\rho) +
\E\left[\sum_{C \in \Act(\rho)} \calW(C)C(\rho \cup g) \right] \geq
\frac{1}{q^{w-1}} \cdot \sum_{C \in \calC} z_C \calW(C).$$ It is easy
to check that for $C \not\in \Act(\rho)$, $z_C > 0$ only if $C(\rho) =
1$, and thus $\sum_{C \not\in \Act(\rho)} C(\rho) \calW(C) \ge \sum_{C
  \not\in \Act(\rho)} z_C \calW(C)$.  For $C \in \Act(\rho)$, we have
the following claim:
\begin{claim}
\label{claim:trev_rounding}
For $C \in \Act(\rho)$, 
$\E[C(\rho \cup g)]  \ge \frac{z_C}{q^{w-1}}$.
\end{claim}
\begin{proof}
Suppose there are exactly $h$ variables in $C$ which are not in $S_\rho$. Let these variables
be $(v_i)_{i=1}^{h}$. Let $(v_i, a_i)_{i=1}^{h}$ be the assignment to these variables that makes $C$ satisfied.
\begin{align*}
\E[C(\rho \cup g)] = \Pr[C \text{ is satisfied by }\rho \cup g] & \geq \prod_{i=1}^{h}\left(\frac{w-1}{qw} + \frac{t_{v_i,a_i}}{w} \right)\\
&\geq \prod_{i=1}^{h} \left(\frac{w-1}{qw} + \frac{z_C}{w}\right) = \left( \frac{w-1}{qw} + \frac{z_C}{w}\right)^h\\
&= \left( \frac{w-1}{qw} + \frac{z_C}{w}\right)^w \geq \frac{z_C}{q^{w-1}}
\end{align*}
Here the last inequality follows form the observation that the minimum
of the function $\frac{\left( \frac{w-1}{qw} +
    \frac{z}{w}\right)^w}{z}$ as $z$ varies in $[0,1],$ is attained at
$z = \nfrac{1}{q}.$
\end{proof}

\begin{figure*}[h!]
\begin{tabularx}{\textwidth}{|X|}
\hline
\vspace{-26pt}
\begin{center}
\[\begin{array}{rrllr}
& z_C &\leq & t_{v,i} & \forall C\in \calC, \forall (v,i) \in C \\
& 1\geq t_{v,i} &\geq & 0 & \forall v \in V\setminus S_\rho, i\in [q]\\
& \sum_{i=1}^q t_{v,i} & =& 1  & \forall v \in V\\
& t_{v,i} & = & 1  & \forall v \in S_\rho \text{ and } i\in [q], \\
& & & & \text{ such that }  h_\rho(v) = i\\
\end{array}\]
\vspace{-24pt}
\end{center}
\\
\hline
\end{tabularx}
\caption{Linear inequalities $\LL_1(\rho)$}
  \label{fig:maxcsp-lp1}
\end{figure*}

\begin{figure*}[h!]
\begin{tabularx}{\textwidth}{|X|}
\hline
\vspace{-26pt}
\begin{center}
\[\begin{array}{rrllr}
\sum_{C\in \calC} \calW_\ell(C)\cdot z_C  \geq c_\ell & \forall \ell\in [k]\\
\vect, \vecz  \text{ satisfy }  \LL_1(\rho).&
\end{array}\]
\vspace{-24pt}
\end{center}
\\
\hline
\end{tabularx}
\caption{Linear inequalities $\LL_2(\rho)$}
  \label{fig:maxcsp-lp2}
\end{figure*}

This completes the proof of the Lemma.

\end{proof}

\pagebreak

\subsection{The Algorithm}

We now give our Pareto approximation algorithm for \maxcspwq in Figure~\ref{fig:weighted-maxcsp} (which uses the procedure from Figure~\ref{fig:tuple-maxcsp}).

\begin{figure*}[h!]
\begin{longtabu}{|X|}
\hline
\vspace{0mm}
{\bf Input}: A tree node $\nu$ and an instance $\calW_\ell$.\\
{\bf Output}: A tuple of variables of size at most $w$.

\begin{enumerate}[itemsep=0mm]
\item Let $v_1\in V\setminus S_{\rho_\nu}$ be a variable which
  maximizes the value of $\activedeg_{\rho_\nu}(v_1,\calW_\ell)$. Set
  $D \leftarrow\{v_1\}$.

\item While $|D| \leq w,$ do the following
\begin{enumerate}[itemsep=0mm]
\item  If there is a variable $v$ in $V\setminus S_{\rho_\nu}$ such that 
$$\activedeg_{\rho_\nu}(D \cup v,\calW_\ell) \geq \frac{\activedeg_{\rho_\nu}(D, \calW_\ell)}{(4qwtk)^{w}},$$
set $D \leftarrow D\cup v.$
\item Otherwise, go to Step~\ref{alg:tuple-selection:return}.
\end{enumerate}
\item 
\label{alg:tuple-selection:return}
Return $D$ as a tuple (in arbitrary order, with $v_1$ as the first
element).

\end{enumerate}
%\vspace{-10pt}
\\
\hline
\end{longtabu}
\caption{{\sc TupleSelection} for \maxconj}
  \label{fig:tuple-maxcsp}
\end{figure*}

\begin{figure*}
\begin{longtabu}{|X|}
\hline
\vspace{0mm}
{\bf Input}: $k$ instances of \maxconj $\calW_1,\ldots,\calW_k$ on the
variable set $V,$ $\eps > 0$ and  and target objective values $c_1, \ldots, c_k.$\\
{\bf Output}: An assignment to $V$\\
{\bf Parameters:}  $\delta_0 = \frac{1}{10(k+1)}$, $\epsilon_0 = \eps$, $\gamma=\frac{\eps_0^2\delta_0}{{w^2(qw)}^{w}}$, $t = \ceil{\frac{20w^2k^2}{\gamma}\cdot\log\left( \frac{10k}{\gamma}\right)}$
\begin{enumerate}[itemsep=0mm]
\item Initialize tree $T$ to be an empty $q^w$-ary tree (i.e., just 1
  root node and each node has at most $q^w$ children).
\item We will associate with each node $\nu$ of the tree:
\begin{enumerate}
\item A partial assignment $\rho_\nu$.
\item A special set of variables $\calT_\nu \subseteq V \setminus S_{\rho_{\nu}}$.
\item A special instance $\Inst_\nu \in [k]$.
\item A collection of integers $\cnt_{\nu,1}, \ldots, \cnt_{\nu, k}$.
\item A trit representing whether the node $\nu$ is living, exhausted
  or dead. %\Sknote{Notation change}
\end{enumerate}
\item Initialize the root node $\nu_0$ to (1) $\rho_{\nu_0} = (\emptyset, \emptyset)$, (2) have all $\cnt_{\nu_0,\ell}=0$, (3) living.

\item While there is a living leaf $\nu$ of $T,$ do the following: 
\begin{enumerate}
\item Check if the LP $\LL_2(\rho_\nu)$ has a feasible solution.
\begin{enumerate}
\item If $\vect, \vecz$ is a feasible solution, then define $p_\nu: V \setminus S_{\rho_{\nu}} \to \Dist(q)$ by
$p = \smooth(\vect)$.
\item If not, then declare $\nu$ to be dead and return to Step $4$. 
\end{enumerate}

\item For each $\ell \in [k]$,
compute $\varcalc_{\rho_\nu}(p, \calW_\ell), \meancalc_{\rho_\nu}(p, \calW_\ell)$.
\item If $\varcalc_{\rho_\nu}(p, \calW_\ell) \geq \delta_0 \epsilon_0^2 \meancalc_{\rho_\nu}(p, \calW_\ell)^2 $, then set $\flag_\ell = \true$, else set $\flag_\ell = \false$.
\item Choose the smallest $\ell \in [k]$, such that
$\cnt_\ell < t$ AND $\flag_\ell = \true$ (if any):
\begin{enumerate}

\item Set $\calT_\nu \leftarrow$  {\sc TupleSelection}($\nu, \calW_\ell$). Set $\Inst_\nu = \ell.$

\item Create $q^{w'}$ children of $\nu$, with labels $\nu b$ for each $b\in [q]^{w'}$ and define

\begin{itemize}[label = $-$] 
\item $\rho_{\nu b} = (S_{\rho_\nu} \cup \calT_\nu, h^{b}),$ where
  $h^{b}$ extends $h_{\rho_{\nu}}$ by $h^{b}(\calT_\nu^i) = b(i).$

\item For each $\ell' \in [k]$ with $\ell' \neq \ell$, initialize $\cnt_{\nu b, \ell'} = \cnt_{\nu, \ell'}.$ Initialize  $\cnt_{\nu b, \ell} = \cnt_{\nu, \ell}+1.$

\item Set $\nu b$ to be living.
\end{itemize}

\end{enumerate}
\item If no such $\ell$ exists, declare $\nu$ to be exhausted. 

\end{enumerate}
\item Now every leaf of $T$ is either exhausted or dead. For each
  exhausted leaf $\nu$ of $T$:
\label{alg:maxcsp:step:leaves}
\begin{enumerate} 
\item 
\label{alg:maxcsp:step:random-assignment}
Sample $g_\nu: V \setminus S_{\rho_\nu} \to [q]$ by independently
sampling $g_\nu(v)$ from the distribution $p_\nu(v)$.
\item 
%\begin{itemize}
%\item Set $f = h \cup g_{\nu}$.
  For every assignment $h: S_{\rho_\nu} \to [q],$ compute
  $\out_{h,g_\nu} \leftarrow \min_{\ell \in [k]} \frac{\val(h \cup
    g_{\nu}, \calW_\ell)}{c_\ell}.$ If $c_\ell = 0$ for some $\ell \in
  [k],$ we interpret $\frac{\val(h \cup g_{\nu},\calW_l)}{c_\ell}$ as
  $+\infty.$
%\end{itemize}
\end{enumerate}
\item Output the largest $\out_{ h, g_{\nu}}$ seen, and the assignment $h \cup g_{\nu}$ that produced it.
\end{enumerate}
\\ 
\hline
\end{longtabu}
\caption{Algorithm {\algowcsp} for approximating weighted simultaneous
  \maxconj}
  \label{fig:weighted-maxcsp}
\end{figure*}

\clearpage
\subsection{Analysis}
Notice that the depth of the tree $T$ is at most $kt$, and that for every
$\nu$, we have that $|S_{\rho_{\nu}}| \leq wkt$. This implies that the running time
is at most $q^{O(wkt)} \cdot \poly(n)$.

Let $f^\star: V \to [q]$ be an assignment such that
$\val(f^\star, \calW_\ell) \geq c_\ell$ for each $\ell \in [k]$. Let $\nu^\star$ be the the unique leaf of the tree $T$ for which
$\rho_{\nu^\star}$ is consistent with $f^\star$. (This $\nu^\star$ can be
found as follows: start with $\nu$ equal to the root. Set $\nu$ to equal
the unique child of $\nu$ for which $\rho_\nu$ is consistent with $f^\star$, and repeat
until $\nu$ becomes a leaf. This leaf is $\nu^\star$). Observe that since $f^\star$
is an assignment such that $\val(f^\star, \calW_\ell) \ge c_\ell$ for
every $\ell \in [k],$ by picking $g_0 = f^\star|_{V \setminus
  S^\star}$ in part 1 of Lemma~\ref{lemma:maxcsp:lp}, we know that
$\LL_2(\rho^\star)$ is feasible, and hence $\nu^\star$ must be an
exhausted leaf (and not dead).\\
\\
Define $\rho^\star = \rho_{\nu^\star}$, $S^\star = S_{\rho^\star}$,
$h^\star = h_{\rho^\star}$ and $p^\star = p_{\nu^\star}$ 

At the completion of Step~4, if $\ell
\in [k]$ satisfies $\cnt_{\nu^\star, \ell} = t$, we call instance $\ell$ a {\em
  high variance} instance.  Otherwise we call instance $\ell$ a
\emph{low variance} instance. 

\subsubsection{Low Variance Instances}
First we show that for the leaf $\nu^*$ in
Step~\ref{alg:maxcsp:step:leaves}, combining the partial assignment
$h^\star$ with a random assignment $g_{\nu^\star}$ in
Step~\ref{alg:maxcsp:step:random-assignment} is good for any low
variance instances with high probability.

\begin{lemma}
\label{lemma:low_variance:maxcsp}
Let $\ell \in [k]$ be any low variance instance. For the leaf node
$\nu^\star$, let $g_{\nu^\star}$ be the random assignment sampled in
Step~\ref{alg:maxcsp:step:random-assignment} of \algowcsp. Then with probability at least $1-\delta_0,$ the
assignment $f=h^\star \cup g_{\nu^\star}$ satisfies:
$$ \Pr_{g_{\nu^\star}} \left[ \val(f,\calW_\ell) \geq (\nicefrac{1}{q^{w-1}}-\nicefrac{\epsilon}{2})\cdot c_\ell \right] \geq 1-\delta_0.$$
\end{lemma}
\begin{proof}
  For every low variance instance $\ell$, we have that
  $\varcalc_{\rho_{\nu^\star}}(p^\star , \calW_\ell) < \delta_0\epsilon_0^2
  \cdot \meancalc_{\rho_{\nu^\star}}(p^\star , \calW_\ell)^2.$ Define $Y
  \defeq \val(\rho^\star \cup g_{\nu^\star}, \calW_\ell) -
  \val(\rho^\star,\calW_\ell)$.  By
  Lemma~\ref{lemma:maxcsp:truemeanvar}, we have $\Pr[ Y < (1-
  \epsilon_0) \E[Y] ] < \delta_0.$ Thus, with probability at least
  $1-\delta_0,$ we have,
\begin{align*}
  \val(f,\calW_\ell) &\geq \val(\rho^\star,\calW_\ell) +
  (1-\epsilon_0) \E[Y]\\
  &= \val(\rho^\star,\calW_\ell) + (1-\epsilon_0) \cdot \meancalc_{\rho_{\nu^\star}}(\smooth(\vect), \calW_\ell) \\
  &= (1-\epsilon_0) \cdot \left(\val(\rho^\star,\calW_\ell) +
    \meancalc_{\rho_{\nu^\star}}(\smooth(\vect), \calW_\ell) \right) \\
  &\geq \frac{1}{q^{w-1}}\cdot(1-\epsilon_0)\cdot \sum_{C \in \calC}
  \calW_\ell(C) \cdot z_C  \ge  \frac{1}{q^{w-1}}\cdot(1-\epsilon_0) \cdot c_\ell\geq
  \left( \frac{1}{q^{w-1}}- \frac{\epsilon}{2}\right) \cdot c_\ell,
\end{align*}
where we have used the second part of
  Lemma~\ref{lemma:maxcsp:lp}.
\end{proof}

Next, we will consider a small perturbation of $h^\star$ which will
ensure that the algorithm performs well on high variance instances too.
We will ensure that this perturbation does not affect the
success on the low variance instances.

\subsubsection{ High variance instances}

Fix a high variance instance $\ell$. Let $\nu$ be an ancestor of
$\nu^\star$ with $\Inst_{\nu} = \ell$. Let $\calT_\nu^1$ denote the
first element of the tuple $\calT_\nu.$ Define:
\begin{align*}
\activedeg_\nu &\defeq \activedeg_{\rho_{\nu}}(\calT_\nu^1, \calW_\ell).\\
\activedeg_{\calT_\nu} &\defeq \activedeg_{\rho_{\nu}}(\calT_\nu,\calW_\ell).
\end{align*}
\begin{observation}
\label{obs:activedegcalT}
For any node $\nu,$ in the tree,
$$ \activedeg_{\calT_\nu} \geq \frac{\activedeg_\nu}{{(4qwtk)}^{w\cdot(|\calT_\nu|-1)}}.$$
\end{observation}
\begin{proof}
For $\nu$ such that $|\calT_\nu| = 1,$ we have, by definition,  $\activedeg_{\calT_\nu} = \activedeg_\nu$ and the inequality follows.  
The lower bound is obvious from the {\sc Tuple Selection} procedure if
$|\calT_\nu| > 1.$ 
\end{proof}

Let $\calC_{\nu}$ be the set of all constraints $C$ containing all variables in $\calT_\nu$ which are active given $\rho_\nu$.

We call a constraint $C$ in $\calC_{\nu}$ a $\backward$ constraint if
$C$ only involves variables from $S_{\rho_{\nu}} \cup
\calT_\nu$. Otherwise we call $C$ in $\calC_{\nu}$ a $\forward$
constraint.  Let $\calC_{\nu}^{\backward}$ and
$\calC_{\nu}^{\forward}$ denote the sets of these constraints.
Finally, let $\calC_{\nu}^{\out}$ denote the set of all constraints
from $\calC_\nu$ that involve at least $one$ variable from $V \setminus S^\star$
and none from $S^\star\setminus S_{\rho_\nu}.$

Define $\backward$ degree and $\forward$ degree  as follows:
\begin{align*}
\backward_{\nu} &\defeq \sum_{C \in \calC_{\nu}^\backward}\calW_{\ell}(C),\\
\forward_{\nu} &\defeq \sum_{C \in \calC_{\nu}^\forward}\calW_{\ell}(C).
\end{align*}
Note that:
$$\activedeg_{\calT_\nu} =  \backward_{\nu}  + \forward_{\nu}.$$

Based on the above definitions, we classify $\nu$ into one of three categories:
\begin{enumerate}
\item If $\backward_{\nu} \geq \frac{1}{2}
  \cdot\activedeg_{\calT_\nu}$, then we call $\nu$ $\typeAB$.
%\item Otherwise, if there is a constraint containing only variables in
%  $\calT_\nu$ of weight at least $\frac{\activedeg_{\calT_\nu} }{{(4qwtk)}^{w}}$ then we call $\nu$ $\typeB.$

\item Otherwise, we call $\nu$ $\typeC$. 
\end{enumerate}

We have the following lemma about $\typeC$ nodes.
\begin{lemma}
\label{lemma:typeC}
For every $\typeC$ node $\nu,$ we have
\begin{enumerate}
\item \label{lemma:typeC:1} For every $v\in V\setminus (S_{\rho_\nu}\cup \calT_\nu),$
  $\activedeg_{\rho_\nu}(\calT_\nu \cup \{v\}, \calW_\ell) \leq
  \frac{\activedeg_{\calT_\nu} }{{(4qwtk)}^{w}}.$
\item\label{lemma:typeC:2} $\sum_{C \in \calC_{\nu}^{\out}} \calW_\ell(C) \geq \frac{1}{4}\cdot \activedeg_{\calT_\nu}.$
\end{enumerate}
\end{lemma}
\begin{proof}
  If $\nu$ is a $\typeC$ node, we must have that for every $v\in V\setminus(
  S_{\rho_\nu}\cup \calT_\nu),$
  \[ \activedeg_{\rho_\nu}(\calT_\nu \cup \{v\}, \calW_\ell) <
    \frac{\activedeg_{\calT_\nu} }{{(4qwtk)}^{w}}.\] This
  follows from the description of the {\sc TupleSelection} procedure,
  and the observation that $\activedeg_\nu(T, \calW_\ell)=0$ for any
  $T \subset V$ with $|T| > w.$
% since otherwise, for some $v\in V\setminus S_{\rho_\nu},$ if
% $\activedeg_{\rho_\nu}(\calT_\nu \cup v,\calW_\ell)$ is at least
% $\frac{1}{{(4qwtk)}^{w(|\calT_\nu|+1)}}$ which is a contradiction to
% the {\sc Tuple Selection} procedure.  

In particular, since $|S^\star|
\leq wtk$, the total weight of constraints containing $\calT_\nu$ and
some variable in $S^\star\setminus (S_{\rho_\nu}\cup \calT_\nu)$ is at most
\begin{align*}
  \sum_{v\in S^\star\setminus
    (S_{\rho_\nu}\cup \calT_\nu)}\activedeg_{\rho_\nu}(\calT_\nu \cup \{v\},
  \calW_\ell) &\leq
  \sum_{v\in S^\star\setminus(S_{\rho_\nu}\cup \calT_\nu)}   \frac{\activedeg_{\calT_\nu} }{{(4qwtk)}^{w}}\\
  &\leq |S^\star \setminus (S_{\rho_\nu}\cup \calT_\nu)|\cdot   \frac{\activedeg_{\calT_\nu} }{{(4qwtk)}^{w}}\\
  &\leq wtk\cdot   \frac{\activedeg_{\calT_\nu} }{{(4qwtk)}^{w}} \leq \frac{1}{4} \cdot \activedeg_{\calT_\nu}.
%\sum_{i=1}^{w-|\calT_\nu|}{|S^\star\setminus S_{\rho_\nu}| \choose i}\cdot q^i\cdot \frac{\activedeg_\nu }{{(4qwtk)}^{w\cdot(|\calT_\nu|+1)}} &\leq \sum_{i=1}^{w}{wtk \choose i} \cdot q^i\cdot \frac{\activedeg_\nu }{{(4qwtk)}^{w\cdot(|\calT_\nu|+1)}} \\
%&\leq (wtk)^w \cdot q^w\cdot\frac{\activedeg_\nu }{{(4qwtk)}^{w\cdot(|\calT_\nu|+1)}}\\
%&\leq \frac{\activedeg_\nu }{{4\cdot(4qwtk)}^{w\cdot(|\calT_\nu|)}}\\
%&\leq \frac{1}{4} \cdot \activedeg_{\calT_\nu}
\end{align*}
Thus, we get,
\begin{align*}
  \sum_{C \in \calC_{\nu}^{\out}} \calW_\ell(C)  &= \forward_{\nu} - \left\{\substack{\text{total weight of constraints containing}\\ \text{$\calT_\nu$ and some variable in $S^\star\setminus (S_{\rho_\nu}\cup \calT_\nu)$ } }\right\}\\
&\geq \frac{1}{2}\cdot\activedeg_{\calT_\nu} - \frac{1}{4}\cdot
\activedeg_{\calT_\nu} = \frac{1}{4}\cdot \activedeg_{\calT_\nu}.
\end{align*}
This completes the proof of second statement.
\end{proof}
%
%For $\nu$ which are $\typeC$, $\calT_\nu$
%has a large fraction of its active degree coming from constraints between
%$\Ver_\nu(1)$ and $V \setminus S^\star.$\Anote{active degree of $\calT_\nu$}

For a partial assignment $g : V\setminus S^\star \to [q]$, we say that $g$
is $\Cgood$ for $\nu$ if there exists a setting of variables in $\calT_\nu$ that satisfies at
least $\frac{1}{8\cdot{(qw)^w}}\cdot\activedeg_{\calT_\nu}$ weight amongst constraints in $\calC_{\nu}^{\out}.$

The next lemma allows us to prove that that for every node $\nu$ of
$\typeC$, with high probability, the random assignment $g_
{\nu^\star}: V \setminus S^\star \to [q],$ is $\Cgood$ for $\nu$.
\begin{lemma}
\label{lemma:hoeffding}
Let $\nu$ be $\typeC$.  Suppose $g : V \setminus S^\star \to [q]$ is a
random assignment obtained by independently sampling $g(v)$ for each
$v\in V \setminus S^\star$ from a distribution such that distribution
$\Pr[ g(v) = i] \geq \frac{1}{qw}$ for each $i\in [q].$ Then:
\[ \Pr_{g} [ g \mbox { is $\Cgood$ for $\nu$} ] \geq 1- 2\cdot e^{-tk/8qw}.\]
\end{lemma}
\begin{proof}
Let $\ell = \Inst_\nu$.

Consider a constraint $C \in \calC_{\nu}^{\out}.$ For partial
assignments $b : \calT_\nu \rightarrow [q]$ and $g: {V \setminus
  S^\star} \to [q],$ define $C(\rho_{\nu} \cup b\cup g) \in \{0,1\}$
to be 1 iff $C$ is satisfied by $\rho_\nu \cup b \cup g$. Since $C$
only contains variables from $S_{\rho_\nu} \cup \calT_\nu \cup (V
\setminus S^\star),$ we have that $C(\rho_{\nu} \cup b\cup g)$ is well
defined.
% as
% follows. $C(b\cup g) $ equals $1$ iff $C$ is satisfied by the
% assignment $b \cup g,$ and 0 otherwise. Since $C \in \calC_\nu^\out,$
% it only contains variables from $\calT_{\nu} \cup (V \setminus
% S^\star)$ and hence $C(b \cup g)$ is well defined.

Define $\score^{b}: [q]^{V \setminus S^\star} \to \mathbb R$ by
$$\score^{b}(g)  = \sum_{C \in \calC_{\nu}^\out} \calW_\ell(C) \cdot
C(\rho_\nu \cup b\cup g) .$$
In words, $\score^{(b)}(g)$ is the total weight of constraints in $\calC_\nu^\out$ satisfied by setting $S_{\rho_\nu}$ according to $\rho_\nu$, 
setting $\calT_{\nu}$ to $b$, and setting $V \setminus S^*$ according to $g$.

Note that for all $C \in \calC_{\nu}^{\out}$, $\E_g[ \sum_{b :
  \calT_\nu \rightarrow [q]}C(\rho_\nu \cup b\cup g)] \geq
\frac{1}{{(qw)}^{w-|\calT_\nu|}}$. This follows since $C$ is an active
constraint given $\rho_\nu,$ and involves all variables from
$\calT_\nu;$ hence there exists an assignment $b$ to $\calT_\nu$ and
an assignment for at most $w-|\calT_\nu|$ variables from constraint
$C$ in $V\setminus S^\star$ such that $C$ is satisfied. Since, $g$ is
a $\smooth$ distribution, this particular assignment to
$w-|\calT_\nu|$ in $V \setminus S^\star$ is sampled with probability
at least $\frac{1}{{(qw)}^{w-|\calT_\nu|}}.$ Hence, for this
particular choice of $b,$ $C$ is satisfied with probability at least
$\frac{1}{{(qw)}^{w-|\calT_\nu|}}.$ Thus:
\begin{align*}
\sum_{b : \calT_\nu \rightarrow [q]} \E_g\left[ \score^{b}(g)\right] & = \sum_{b : \calT_\nu \rightarrow [q]} \E_g\left[ \sum_{C \in \calC_{\nu}^\out} \calW_\ell(C) \cdot C(\rho_\nu \cup b\cup g)\right]\\
% &=  \sum_{C \in \calC_{\nu}^\out} \E_g\left[\sum_{b : \calT_\nu \rightarrow [q]}\calW_\ell(C) \cdot C(b\cup g)\right]\\
&=  \sum_{C \in \calC_{\nu}^\out} \calW_\ell(C) \cdot \E_g\left[\sum_{b : \calT_\nu \rightarrow [q]} C(\rho_\nu \cup b\cup g)\right]
\geq \frac{1}{{(qw)}^{w-|\calT_\nu|}} \sum_{C \in \calC_{\nu}^{\out}} \calW_\ell(C).
\end{align*}
Thus there exists $b : \calT_\nu \rightarrow [q]$ such that
\begin{align*}
\E_g[\score^{b}(g)] \geq\frac{1}{q^{|\calT_\nu|}}  \cdot\frac{1}{{(qw)}^{w-|\calT_\nu|}}\sum_{C \in \calC_{\nu}^{\out}} \calW_\ell(C)
 \geq \frac{1}{4} \cdot\frac{1}{{(qw)}^w} \cdot\activedeg_{\calT_\nu},
% &\text{by Lemma~\ref{lemma:typeC}}\\ 
% //& = \frac{1}{4\cdot(qw)^w} \cdot\activedeg_{\calT_\nu}&
\end{align*}
where the last inequality follows by Lemma~\ref{lemma:typeC}.

Fix this particular $b$ for which the above inequality holds. We are
going to use McDiarmid's inequality to show the concentration of
$\score^{b}(g)$ around its mean.  Since $\nu$ is $\typeC$, from
Lemma~\ref{lemma:typeC}, we know that for every vertex $v \in V
\setminus S^\star,$ changing $g$ on just $v$ can change the value of
$\score^{b}(g)$ by at most $c_v \defeq \activedeg_{\rho_\nu}(\calT_\nu
\cup \{v\}, \calW_\ell) \le   \frac{\activedeg_{\calT_\nu} }{{(4qwtk)}^{w}}.$ Thus by McDiarmid's inequality  (Lemma~\ref{lemma:mcdiarmid}),
\begin{align*}
  \Pr_{g} [ g \mbox { is not $\Cgood$ for $\nu$} ] &\leq \Pr_{g} \left[ \score^{b}(g) < \frac{1}{8\cdot(qw)^w} \cdot\activedeg_{\calT_\nu}\right]\\
  & \leq \Pr_{g}\left[ | \score^{b}(g) - \E_{g}[\score^{b}(g)] | > \frac{1}{8\cdot(qw)^w} \cdot\activedeg_{\calT_\nu}\right]\\
  & \leq 2\cdot\exp\left( \frac{-2\cdot\activedeg_{\calT_\nu}^2}{64
      (qw)^{2w} \cdot \sum_{v\in
        V\setminus S^\star} c_v^2 }\right). \\
  & \leq 2\cdot\exp\left( \frac{-2\cdot\activedeg_{\calT_\nu}^2}{64
      (qw)^{2w} \cdot(\max_v{c_v})\cdot\sum_{v\in  V\setminus S^\star} c_v}\right)\\
  & \leq 2\cdot\exp\left( \frac{-2\cdot\activedeg_{\calT_\nu}^2}{64  (qw)^{2w}\cdot (\max_v{c_v})\cdot \activedeg_{\calT_\nu}}\right)\\
  & \leq 2\cdot\exp\left( \frac{-2\cdot\activedeg_{\calT_\nu}}{64 (qw)^{2w}
      \cdot (\max_v{c_v}) }\right)\\
  & \leq 2\cdot\exp\left( \frac{-2\cdot\activedeg_{\calT_\nu}}{64 (qw)^{2w}
      \cdot (\frac{\activedeg_{\calT_\nu} }{{(4qwtk)}^{w}}) }\right)\\
%  & \leq 2\cdot\exp\left( \frac{-2\cdot\frac{\activedeg_{\nu}}{(4qwtk)^{w\cdot |\calT_\nu|}}}{64 (qw)^{2w}
%      \cdot \left( \frac{\activedeg_\nu
%      }{{(4qwtk)}^{w\cdot(|\calT_\nu|+1)}} \right) }\right) \mbox{\quad\quad By observation~\ref{obs:activedegcalT}}\\
& = 2\cdot\exp\left( \frac{-2\cdot (4qwtk)^{w}}{64\cdot(qw)^{2w}}\right) \leq  2\cdot\exp\left( \frac{-tk}{8qw}\right).
\end{align*}
%\Anote{Verified the last inequality}
%
%
%We want to upper bound \[ \sum_{v\in V\setminus S^\star}
%\activedeg_{\rho_\nu}(\calT_\nu \cup v, \calW_\ell)^2,\] as follows:
%\begin{align*}
%\sum_{v\in  V\setminus S^\star} \activedeg_{\rho_\nu}(\calT_\nu \cup v, \calW_\ell)^2 &\leq \max_{v\in V\setminus S^\star}{\activedeg_{\rho_\nu}(\calT_\nu \cup v, \calW_\ell)} \cdot\sum_{v\in  V\setminus S^\star} \activedeg_{\rho_\nu}(\calT_\nu \cup v, \calW_\ell)\\
%& \leq \max_{v\in V\setminus S^\star} {\activedeg_{\rho_\nu}(\calT_\nu \cup v, \calW_\ell)} \cdot \activedeg_{\calT_\nu}\\
%& \leq  \frac{\activedeg_{\calT_\nu}^2 }{{(4qwtk)}^{w\cdot(|\calT_\nu|+1)}}.
%\end{align*}
%where the last inequality follows form Lemma~\ref{lemma:typeC}. Hence,
%\begin{align*}
%\Pr_{g} [ g \mbox { is not $\Cgood$ for $\nu$} ] 
%& \leq 2\cdot\exp\left( \frac{-2\cdot\activedeg_{\calT_\nu}^2}{64\cdot(qw)^{2w} \cdot(\frac{\activedeg_{\calT_\nu}^2}{(4qwtk)^{w(|\calT_\nu| + 1)}})}\right)\\
%& = 2\cdot\exp\left( \frac{-2\cdot (4qwtk)^{w(|\calT_\nu|+1)}}{64\cdot(qw)^{2w}}\right) \leq  2\cdot\exp\left( \frac{-tk}{32}\right).
%\end{align*}
\end{proof}

For a high variance instance $\ell$, let $\nu^\ell_1, \ldots,
\nu^\ell_t$ be the $t$ nodes with $\Inst_{\nu} = \ell$ which lie on
the path from the root to $\nu^\star,$ numbered in order of their
appearance on the path from the root to $\nu^\star$.  Set
$\finalW_\ell = \activedeg_{\rho^\star}(\calW_{\ell}).$ This is the
active degree left over in instance $\ell$ after the restriction
$\rho^\star$.
\begin{lemma}
\label{lemma:maxcsp:earlyvariables}
For every high variance instance $\ell \in [k]$ and for each $i \leq [t/2]$,
$$\activedeg_{\nu^\ell_i} \geq \gamma\cdot(1-\gamma)^{-t/2} \cdot \finalW_\ell \geq 100\cdot(qw)^{w}\cdot(4qwtk)^{w^2}\cdot \finalW_\ell.$$
%\Anote{$100\cdot(qw)^{w}$ is needed in the analysis of typeC node}
\end{lemma}
We skip the proof of this lemma. The first inequality is identical to
the second part of Lemma~\ref{lemma:earlyvariables}, and the second inequality follows
from the choice of $t.$
% \begin{proof}
% Fix a high variance instance $\ell\in [k]$. Note that $b_i = \activedeg_{\rho_{\nu_i^\ell}}(\calW_\ell)$ decreases as
% $i$ increases. The main observation is that
% \begin{enumerate}
% \item $b_{i+1} \leq (1-\gamma) \cdot b_i$.
% \item $\activedeg_{\nu_i^\ell}  \geq \gamma b_i$
% \end{enumerate}
% Thus for all $\nu_i^\ell$ with $i\in [1,t/2]$, $\activedeg_{\nu^\ell_i} \geq \gamma\cdot(1-\gamma)^{-t/2} \cdot \finalW_\ell$ and also the choice of parameters implies for those $\nu_i^\ell$ $\activedeg_{\nu_i^\ell}$ is at least $100\cdot(qw)^{w}\cdot(4qwtk)^{w^2}\cdot\finalW_\ell$.
% \end{proof}

\subsubsection{Putting everything together}
We now show that when $\nu$ is taken to equal $\nu^\star$ in
Step~\ref{alg:maxcsp:step:leaves}, then with high probability over
the choice of $g_{\nu^\star}$ in Step $5(a)$ there is a setting of $h$
in Step $5(b)$ such that $\min_{\ell \in [k]} \val(h \cup
g_{\nu^\star}, \calW_\ell) \geq (\frac{1}{q^{w-1}} - \epsilon) \cdot
c_\ell.$
%
%\Ssnote{Do we want to include the dependence on $q,w$?}
\begin{theorem}
\label{thm:maxcspapprox}
Suppose the algorithm {\algowcsp} is given as inputs $\eps > 0,$ $k$
simultaneous weighted \maxconj instances $\calW_1,\ldots, \calW_k$ on
$n$ variables, and target objective value $c_1, \ldots,c_k$ with the
guarantee that there exists an assignment $f^\star$ such that for each
$\ell \in [k],$ we have $\val( f^\star, \calW_\ell) \ge c_\ell.$ Then,
the algorithm runs in $2^{O(\nfrac{k^4}{\eps^2}\log(\nfrac{k}{\eps}))}
\cdot \poly(n)$ time, and with probability at least 0.9, outputs an
assignment $f$ such that for each $\ell \in [k],$ we have, $\val(f,
\calW_\ell) \ge \left( \frac{1}{q^{w-1}} -\epsilon\right) \cdot
c_\ell.$
\end{theorem}
\begin{proof}
  Consider the case when $\nu$ is taken to equal $\nu^\star$ in
  Step~\ref{alg:maxcsp:step:leaves}.
% and $h$ is taken to be $h^\star$ in step $5(b)$.
  By Lemma~\ref{lemma:low_variance:maxcsp}, with probability at least
  $1 -k\delta_0$ over the random choices of $g_{\nu^\star},$ we have
  that for {\em every} low variance instance $\ell \in [k]$,
  $\val(h^\star \cup g_{\nu^\star}, \calW_\ell) \geq
  (\frac{1}{q^{w-1}} - \frac{\epsilon}{2}) \cdot c_\ell$.  By
  Lemma~\ref{lemma:hoeffding} and a union bound, with probability at
  least $1 - \frac{t}{2}\cdot k\cdot 2e^{-tk/8qw} \geq 1 - \delta_0$
  over the choice of $g_{\nu^\star}$, for every high variance instance
  $\ell$ and for every $\typeC$ node $\nu_i^\ell$, $i\in [t/2]$, we
  have that $g_{\nu^\star}$ is $\Cgood$ for $\nu_i^\ell$.  Thus with
  probability at least $1- (k+1)\delta_0$, both these events occur.
  Henceforth we assume that both these events occur in Step $5(a)$ of
  the algorithm.

Our next goal is to show that there exists a partial assignment $h:S^\star \to [q]$ such that
\begin{enumerate}
\item \label{item:max-csp:proof:guarantees:1} For every instance $\ell
  \in [k]$, $\val(h \cup g_{\nu^\star} ,\calW_\ell) \geq \left( 1 -
    \nfrac{\epsilon}{2} \right) \cdot \val(h^\star \cup
  g_{\nu^\star},\calW_\ell)$
\item \label{item:max-csp:proof:guarantees:2} For every high variance
  instance $\ell\in [k]$, $\val(h \cup g_{\nu^\star},\calW_\ell)
  \geq(1-\nfrac{\epsilon}{2})\cdot 10\cdot \finalW_\ell$.
\end{enumerate}

Before giving a proof of the existence of such an $h$, we show that
this completes the proof of the theorem. We claim that when
the partial assignment $h$ guaranteed above is considered in the
Step~$5(b)$ in the algorithm, we obtain an assignment with the
required approximation guarantees. 

 For every low variance
instance $\ell \in [k],$ since we started with $\val(h^\star \cup g_{\nu^\star},
\calW_\ell) \geq (\frac{1}{q^{w-1}} - \frac{\epsilon}{2}) \cdot c_\ell,$
property~\ref{item:max-csp:proof:guarantees:1} above implies that
every low variance instance $\val(h \cup g_{\nu^\star}) \ge (\frac{1}{q^{w-1}} -
\epsilon) \cdot c_\ell.$ For every high variance instance $\ell \in
[k],$ since $h^\star = f^\star|_S,$
\[\val(h^\star \cup g_{\nu^\star}, \calW_\ell) \ge \val(f^\star, \calW_\ell) -
\activedeg_{\rho^\star}(\calW_\ell) \ge c_\ell - \finalW_\ell.\]
Combining this with properties~\ref{item:max-csp:proof:guarantees:1}
and \ref{item:max-csp:proof:guarantees:2} above, we get,
\[ \val(h \cup g_{\nu^\star}, \calW_\ell) \ge \left( 1 - \frac{\epsilon}{2}
\right) \cdot \max\{c_\ell - \finalW_\ell, 10\cdot \finalW_\ell\} \ge
\frac{10}{11} \left( 1 - \frac{\epsilon}{2} \right) \cdot c_\ell.\]
Thus, for all instances $\ell \in [k]$, we get $\val(h \cup g_{\nu^\star}, \calW_\ell) \ge
\left(\frac{1}{q^{w-1}}-\frac{\eps}{2}\right) \cdot c_{\ell}.$

Now, it remains to show the existence of such an $h$ by giving a
procedure for constructing $h$ by perturbing $h^\star$ (Note that this
procedure is only part of the analysis). For nodes $\nu, \nu'$ in the
tree, let us write $\nu \prec \nu'$ if $\nu$ is an ancestor of $\nu'$,
and we also say that $\nu'$ is ``deeper" than $\nu$. \\

\noindent{\bf Constructing $h$:}
\begin{enumerate}
\item Initialize $H \subseteq [k]$ to be the set of high variance instances.
\item Let $N_0 = \{ \nu^{\ell}_i \mid \ell \in H, i \in [t/2]\}$.
Note that $N$ is a chain in the tree (since all the elements of $N$
are ancestors of $\nu^\star$). Since every $\nu \in N$ is an
  ancestor of $\nu^\star$, we have $h_{\rho_\nu} = h^\star|_{S_{\rho_\nu}}.$

\item Initialize $D = \emptyset$, $N= N_0$, $h = h^{\star}$.

\item During the procedure, we will be changing the assignment $h$,
and removing elements from $N$. We will always maintain the
following two invariants:
\begin{itemize}
\item $|N| > \frac{t}{4}$.
\item For every $\nu \in N$, $h|_{S_{\rho_\nu}} = h^{\star}|_{S_{\rho_{\nu}}}$.
\end{itemize}

\item While $|D| \neq |H|$ do:
\begin{enumerate}
\item Let $$B = \left\{ v \in V \mid \exists \ell \in [k] \mbox{ with }
  \sum_{C \in \calC, C \owns v} \calW_{\ell}(C)\cdot C(h\cup g_{\nu^\star}) \geq
  \frac{\epsilon}{2wk} \val(h\cup g_{\nu^\star}, \calW_\ell) \right\}.$$

Note that $|B| \leq \frac{2w^2 k^2}{\epsilon} < \frac{t}{4}.$

\item Let $\nu \in N$ be the deepest element of $N$ for which:
  $\calT_\nu\cap B = \emptyset.$

Such a $\nu$ exists because:
\begin{itemize}
\item $|N| > \frac{t}{4} >  |B|$, and
\item there are at most $|B|$ nodes $\nu$ for which 
$\calT_\nu \cap B \neq \emptyset$ (since
$\calT_\nu$ are all disjoint for distinct $\nu$).
\end{itemize}

\item Let $\ell \in H$ and $i \in [t/2]$ be such that $\nu = \nu_i^\ell$. Let $\rho = \rho_\nu$. We will now
  modify the assignment $h$ for variables in $\calT_\nu$ to guarantee
  that $\val(h \cup g_{\nu^\star}, \calW_\ell) \geq 10 \cdot
  \finalW_\ell$. The procedure depends on whether $\nu$ is $\typeAB$ or $\typeC$.
\begin{enumerate}
\item If $\nu$ is $\typeAB$, then we know that $\backward_{\nu} \geq \frac{1}{2}\cdot\activedeg_{\calT_\nu}$.

The second invariant tells us that $\rho = h^\star|_{S_{\rho}} = h|_{S_{\rho}}$.
Thus we have:
%\Anote{Abusing : $C \subseteq S_{\rho} \cup \calT_\nu$}
\begin{align*}
\backward_{\nu} &= \sum_{C \in \calC^\backward_{{\nu}}} \calW_\ell(C)\\
&= \sum_{C \subseteq S_{\rho} \cup \calT_\nu, C \supseteq \calT_\nu, C \in \Act(\rho)} \calW_\ell(C)\\
&= \sum_{C \subseteq S_{\rho} \cup \calT_\nu, C \supseteq \calT_\nu, C \in \Act(h|_{S_\rho})} \calW_\ell(C).
\end{align*}
This implies that we can modify the assignment $h$ on the variables
$\calT_\nu$ such that after the modification, the weights of satisfied
backward constraints is:
\begin{align*}
 \sum_{C \subseteq S_{\rho} \cup \calT_\nu, C \supseteq \calT_\nu, C \in \Act(h|_{S_\rho})} \calW_\ell(C) C(h) &\geq \frac{1}{q^{w}} \sum_{C \subseteq S_{\rho} \cup \calT_\nu, C \supseteq \calT_\nu, C \in \Act(h|_{S_\rho})} \calW_\ell(C)\\
&= \frac{1}{q^{w}} \cdot \backward_{\nu}\\
&\geq \frac{1}{2q^{w}}\cdot \activedeg_{\calT_\nu}\\
&\geq 10 \cdot \finalW_\ell. 
\end{align*}
where the $\frac{1}{q^{w}}$ factor in the first inequality appears
because there could be as many as $q^{w}$ possible assignments to
variables in $\calT_\nu,$ and the last inequality holds because of
Observation~\ref{obs:activedegcalT} and Lemma~\ref{lemma:maxcsp:earlyvariables}.
In particular, after making this change, we
have $\val(h \cup g_{\nu^\star}, \calW_\ell) \geq 10 \cdot
\finalW_\ell$.

%
%\item If $\nu$ is $\typeB$, then we know that some constraint 
%$C$ containing only variables in $\calT_\nu$ has weight at least $\calW_\ell(C) \geq \frac{\activedeg_{\calT_\nu} }{{(4qwtk)}^{w}} \geq 10\cdot \finalW_\ell$ (again, the last inequality by 
%Observation~\ref{obs:activedegcalT} and Lemma~\ref{lemma:maxcsp:earlyvariables}).
%
%Thus we may choose settings for $h(x) \in [q]$ for all $x\in \calT_\nu$ such that
%$C(h) = 1$.
%Thus, after making this assignment to $\calT_\nu$, we have  $\val(h \cup g_{\nu^\star}, \calW_\ell) \geq 10\cdot \finalW_\ell$.
%

\item If $\nu$ is $\typeC$, then we know that $g$ is $\Cgood$ for
  $\nu.$ Thus, by the definition of \Cgood, we can choose a setting of
  $\calT_\nu$ so that at least a total of
  $\frac{1}{8\cdot{(qw)^w}}\cdot \activedeg_{\calT_\nu} \geq 10 \cdot
  \finalW_\ell$ $\calW_\ell$-weight constraints between $\calT_\nu$
  and $V \setminus S^\star$ is satisfied. After this change, we have
  $\val(h \cup g_{\nu^\star}, \calW_\ell) \geq 10 \cdot \finalW_\ell$.
\end{enumerate}

In both the above cases, we only changed the value of $h$ at the variables
$\calT_\nu$. Since $\calT_\nu \cap B = \emptyset$,
we have that for every $j \in [k]$,
the new value $\val( h \cup g_{\nu^\star}, \calW_j)$ is at least $\left(1 - \frac{\epsilon}{2k} \right)$
times the old value $\val( h \cup g_{\nu^\star}, \calW_j)$.

\item Set $D = D \cup \{\ell\}$.
\item Set $N = \{ \nu^{\ell}_i \mid \ell \in H\setminus D, i \leq [t/2], \nu^\ell_i \prec \nu\}$.

Observe that $|N|$ decreases in size by at most $ \frac{t}{2} + |B|$.
Thus, if $D \neq H$, we have
\begin{align*}
|N| &\geq |N_0| - |D| \cdot \frac{t}{2} - |D||B| \\
& = |H|\cdot\frac{t}{2} - |D| \cdot \frac{t}{2} - |D||B|\\
& \geq \frac{t}{2} - k|B| > \frac{t}{4} 
\end{align*}

Also observe that we only changed the values of $h$ at
the variables $\calT_\nu$. 
Thus for all $\nu' \preceq \nu$ (i.e $\nu'\in N),$ we still have the property that 
$h|_{S_{\rho_{\nu'}}} = h^{\star}|_{S_{\rho_{\nu'}}}$.
\end{enumerate}
\end{enumerate}

For each high variance instance $\ell \in [k],$ in the iteration where
$\ell$ gets added to the set $D,$ the procedure ensures that at the
end of the iteration $\val(h \cup g_{\nu^\star}, \calW_\ell) \ge 10\cdot \finalW_\ell.$

Moreover, at each step we reduced the value of each $\val(h\cup g_{\nu^\star},
\calW_\ell)$ by at most $\frac{\epsilon}{2k}$ fraction of its previous
value. Thus, at the end of the procedure, for every $\ell \in [k],$
the value has decreased at most by a multiplicative factor of
$\left(1-\frac{\eps}{2k}\right)^k \ge \left(1-\frac{\eps}{2}\right).$
Thus, for every $\ell \in [k],$ we get $\val(h \cup g_{\nu^\star}, \calW_\ell) \ge
\left(1-\frac{\eps}{2}\right) \cdot \val(h^\star \cup g_{\nu^\star}, \calW_\ell),$
and for every high variance instance $\ell \in [k]$, we have
$\val(h\cup g_{\nu^\star}, \calW_\ell) \geq
\left(1-\frac{\eps}{2}\right) \cdot 10 \cdot
\finalW_\ell$. This proves the two properties of $h$ that we set out
to prove.

{\bf Running time : } Running time of the algorithm is
$2^{O(kt)}\cdot\poly(n)$ which is $2^{O(\nfrac{k^4}{\eps^2}\log
(\nfrac{k}{\eps^2}))}\cdot\poly(n).$
\end{proof}

%\clearpage

%%% This are comments that I (Sushant) need for working with emacs
%%% Local Variables: 
%%% mode: latex
%%% TeX-master: "new-simopt"
%%% End: 

\section{Simultaneous \maxwsat}
\label{section:maxwsat}
In this section, we give our algorithm for simultaneous \maxwsat. The
algorithm follows the basic paradigm from \maxand and \maxcsp, but
does not require a tree of evolutions (only a set of influential
variables), and uses an LP to boost the Pareto approximation factor to
$\left(\frac{3}{4} - \epsilon \right)$.
\subsection{Preliminaries}

%Section~\ref{sec:notation-1}, we define some additional notation for
%this section.

Let $V$ be a set of $n$ Boolean variables. Define $\calC$ to be the
set of all possible $w$-SAT constraints on the $n$ variable set $V$.
A \maxwsat instance is then described by a weight
function $\calW : \calC \to \rea_{\geq 0}$ (here $\calW(C)$ denotes
the weight of the constraint $C$). We will assume that $\sum_{C \in
  \calC} \calW(C) = 1$.

We say $v \in C$ if the variable $v$ appears in the constraint
$C$. For a constraint $C$, let $C^+$ (resp.  $C^-$) denote the set of
variables $v \in V$ that appear unnegated (resp. negated) in the
constraint $C$.

%\paragraph{Constraints and $\val$}
Let $f: V \to \{0,1\}$ be an assignment. For a constraint $C \in
\calC,$ define $C(f)$ to be 1 if the constraint $C$ is satisfied by
the assignment $f$, and define $C(f) = 0$ otherwise. Then, we have the
following expression for $\val(f,\calW)$:
$$\val(f,\calW) \defeq \sum_{C \in \calC} \calW(C) \cdot C(f).$$

\subsubsection{Active Constraints}
Our algorithm will maintain a small set $S \subseteq V$ of variables,
for which we will try all assignments by brute-force, and then use a
randomized rounding procedure for a linear program to obtain an
assignment for $V \setminus S$. We now introduce some notation for
dealing with this.

Let $S \subseteq V$. We say a constraint $C \in \calC$ is active given
$S$ if at least one of the variables of $C$ is in $V\setminus S$.  We
denote by $\Act(S)$ the set of constraints from $\calC$ which are
active given $S$.  For two constraints $C_1,C_2 \in \calC,$ we say
$C_1 \sim_{S} C_2$ if they share a variable that is contained in $V
\setminus S$. Note that if $C_1 \sim_{S} C_2,$ then $C_1, C_2$ are
both in $\Act(S)$. For two partial assignments $f_1 : S \to \B$ and
$f_2 : V\setminus S \to \B$, let $f = f_1\cup f_2$ is an assignment
$f : V \to \B$ such that $f(x) = f_1(x)$ if $x\in S$ otherwise
$f(x) = f_2(x)$.

Define the active degree of a variable $v \in V\setminus S$ given $S$ by:
$$\activedeg_S(v, \calW) \defeq \sum_{ C \in \Act(S), C \owns v} \calW(C).$$
We then define the active degree of the whole instance $\calW$ given $S$:
$$\activedeg_S(\calW) \defeq \sum_{v \in V\setminus S} \activedeg_{S}(v,
\calW).$$
For a partial assignment $h: S \to \{0,1\}$, we define
$$\val(h,\calW) \defeq  \sum_{\substack{C \in \calC \\  C \notin \Act(S)}} \calW(C) \cdot C(h).$$
Thus, for an assignment $g: V\setminus S \to \{0,1\}$, to the
remaining set of variables, we have the equality:
$$\val(h\cup g,\calW) - \val(h,\calW) = \sum_{C \in \Act(S)}\calW(C) \cdot C(h\cup g).$$

\subsubsection{LP Rounding}
Let $h: S \to \{0,1\}$ be a partial assignment.  We will use the
Linear Program $\LLwsat_1(h)$ to complete the assignment to $V
\setminus S$. For \maxsat, Goemans and Williamson~\cite{GW-maxsat}
showed, via a rounding procedure, that this LP can be used to give a
$\nfrac{3}{4}$ approximation. However, as in \maxand, we will be using
the rounding procedure due to Trevisan~\cite{Trevisan98} that also
gives a $\nfrac{3}{4}$ approximation for \maxwsat, because of its
smoothness properties.

Let $\vec{t}, \vec{z}$ be a feasible solution to the LP
$\LLwsat_1(h)$. Let $\smooth(\vec{t})$ denote the map $p: V \setminus
S \to [0,1]$ given by: $p(v) = \frac{1}{4} + \frac{t_v}{2}$.  Note
that $p(v) \in [\nfrac{1}{4}, \nfrac{3}{4}]$ for all $v$.
%For a fractional assignment $p: V\setminus S \to [0,1]$,
%define $\val(p,\calW)$ by:
%\[\val(p,\calW) \defeq \sum_{C \in \calC} \calW(C) \cdot z(C),\]
%%\[\mbox { where } \quad 
%z(C) \defeq \min\{ 1, \sum_{v \in C^+} p(v) + \sum_{v \in C^-}
%(1-p(v)) \}.\] Note that if $p$ is $\{0,1\}$ valued assignment, then
%the above definition is consistent with the original definition of
%$\val(p,\calW)$.
%
\begin{theorem}
\label{thm:sattrev}
Let $h: S \to \{0,1\}$ be a partial assignment.
\begin{enumerate}
\item For every $g_0 : V \setminus S \to \{0,1\}$, there exist
$\vect$, $\vecz$ satisfying $\LLwsat_1(h)$ such that for every
\maxwsat instance $\calW$:
$$ \sum_{C \in \calC} \calW(C) z_C = \val(g_0 \cup h, \calW).$$
\item Suppose $\vect, \vecz$ satisfy $\LLwsat_1(h)$. Let $p =
  \smooth(\vec{t})$.  Then for every \maxwsat instance $\calW$:
  \[ \av_g[\val(h \cup g, \calW) ] \geq \frac{3}{4} \cdot \sum_{C \in
    \calC} \calW(C)z_C,\] where $g: V\setminus S \to \{0,1\}$ is such
  that each $g(v)$ is sampled independently with $\E[g(v)] = p(v)$.
  % \Sknote{ Should understand the difference between $\E[\val(h \cup
  %   g)] - \val(h)$ and $\E[\val(h \cup g)]$. Which one do we prove,
  %   which one do we need?}
\end{enumerate}
\end{theorem}
\begin{proof}
The first part is identical to the first part of
Lemma~\ref{lemma:max2and:lp}.
%, and we skip the proof.
% We begin with the first part.  For $v \in S$, define $t_{v} =
%   h(v).$ For $v \in V \setminus S_{\rho},$ define $t_{v} =
%   g_0(v).$ For $C \in \calC$, define $z_C = 1$ if $C(g_0 \cup h) =
%   1$, and define $z_C = 0$ otherwise.  It is easy to see that these
%   $\vect, \vecz$ satisfies $\LLwsat_1(h)$, and that for every
%   instance $\calW$:
% $$ \sum_{C \in \calC} \calW(C) z_C = \val(g_0 \cup h, W).$$
%
For the second part. Let $\calW$ be any instance of
$\maxwsat$. Let $g: V \setminus S \to
\{0,1\}$ be sampled as follows: independently for each $v \in V
\setminus S$, $g(v)$ is sampled from $\{0,1\}$ such that
$\av[g(v)] = p(v)$.  We need to show that
\begin{align*}
  \av_g[\val(h \cup g, \calW) ] & =
  \sum_{C \in \calC \setminus\Act(S)} \calW(C)C(h) +
  \E\left[\sum_{C \in \Act(S)} \calW(C)C(\rho \cup g) \right] \\
& \ge \frac{3}{4} \cdot \sum_{C \in \calC \setminus \Act(S)} \calW(C) z_C + \frac{3}{4} \cdot \sum_{C \in \Act(S)} \calW(C) z_C
\end{align*}
% $$\meancalc_\rho(\smooth(\vect), \calW) \defeq \E\left[\sum_{C \in \Act(\rho)} \calW(C)C(\rho \cup g) \right] \geq
% \frac{1}{2} \cdot \sum_{C \in \Act(\rho)} \calW(C) z_C.$$ 
For $C \in \calC \setminus \Act(S),$ it is easy to verify that if
$z_C > 0,$ we must have $C(h) = 1.$ For $C \in \Act(S)$
the following claim gives us the required inequality:
\begin{claim}
\label{claim:maxsat:trev_rounding}
For $C \in \Act(S)$, 
$\E[C(h \cup g)]  \ge \frac{3}{4}\cdot z_C$.
\end{claim}
\begin{proof}
%\Anote{Note the definition of active in this case}
The claim is true if $C$ is satisfied by $h.$ Consider a clause $C$ which contains $l$ active variables but not satisfied by partial assignment $h.$ Under the smooth rounding, we have
\begin{align*}
\E[C(h \cup g)] =  \Pr[C \text{ is satisfied by }h \cup g] &= 1 - \left(\prod_{v\in C^+, v\in V\setminus S}\frac{3}{4} -
    \frac{t_v}{2}\right)\cdot\left(\prod_{v\in C^-, v\in V\setminus
      S}\frac{3}{4} - \frac{1-t_v}{2}\right)\\
%&=1-\prod_{i=1}^{l}\left(1-\left(\frac{1}{4} + \frac{t_{x_i}}{2}\right)\right) = 1-\prod_{i=1}^{l}\left(\frac{3}{4} - \frac{t_{x_i}}{2}\right) \\
 % &\geq 1-\left(\frac{\sum_{v\in C^+, v\in V\setminus S}\left(\frac{3}{4} -
 %    \frac{t_v}{2}\right) + \sum_{v\in C^-, v\in V\setminus S}\left(\frac{3}{4} -
 %    \frac{1-t_v}{2}\right)}{l} \right)^{l} \\
&\ge 1 - \left(\frac{3}{4} - \frac{\sum_{v\in C^{+}, v\in V\setminus S} t_v + \sum_{v\in C^{-}, v\in V\setminus S} (1-t_v)}{2l} \right)^l \\
&\geq  1 - \left(\frac{3}{4} - \frac{z_C}{2l} \right)^l \geq \frac{3}{4}\cdot z_C,
\end{align*}
where first inequality follows from AM-GM inequality. For any integer $l\geq 1$, the last inequality follows by noting that for a function $f(x) = 1 - \left(\frac{3}{4} - \frac{x}{2l} \right)^l - \frac{3}{4}\cdot x$, $f(0) \geq 0, f(1)\geq 0$ along with the fact the the function has no local minima in $(0,1).$
\end{proof}
%Suppose there are exactly $h$ variables in $C$ which are not in
%$S$. We have $h\leq 2.$
%\begin{align*}
%\E[C(h \cup g)] =  \Pr[C \text{ is satisfied by }h \cup g] &=
%  \left(\prod_{v\in C^+, v\in V\setminus S}\frac{1}{4} +
%    \frac{t_v}{2}\right)\cdot\left(\prod_{v\in C^-, v\in V\setminus
%      S}\frac{1}{4} + \frac{1-t_v}{2}\right) \\
%  & \geq \left(\frac{1}{4} + \frac{z_C}{2}\right)^h \geq
%  \left(\frac{1}{4} + \frac{z_C}{2}\right)^2 \geq \frac{z_C}{2}.
%\end{align*}
\end{proof}

% We omit the proof, because a very similar statement is proved in 
% the next section. 
%\Ssnote{Edit this line.}
%A very similar statement is proved in the next section, and so we omit
%the proof.

\subsection{Random Assignments}
\label{section:maxwsat:random}
We now give a sufficient condition
%(very similar to the \maxcut case)
for the value of a \maxwsat instance to be highly concentrated under a
{\em sufficiently smooth} independent random assignment to the
variables of $V \setminus S$ (This smooth distribution will come from
the rounding algorithm for the LP). When the condition does not hold,
we will get a variable of high active degree.

Let $S \subseteq V$, and let $h: S \to \{0,1\}$ be an arbitrary
partial assignment to $S$.  Let $p: V \setminus S \to [0,1]$ be such
that $p(v) \in [\nfrac{1}{4}, \nfrac{3}{4}]$ for each $v \in V
\setminus S$. Consider the random assignment $g: V \setminus S \to
\{0,1\}$, where for each $v \in V\setminus S$, $g(v) \in \{0, 1\}$ is
sampled independently with $\E[g(v)] = p(v)$. Define the random
variable
$$Y \defeq \val(h\cup g,\calW) - \val(h,\calW) = \sum_{C \in \Act(S)} \calW(C)\cdot C(h\cup g).$$
The random variable $Y$ measures the contribution of active
constraints to the instance $\calW$.

% As in the last section, we define the random variable
% $$Y \defeq \val(h\cup g,\calW) - \val(h,\calW) = \sum_{C \in \Act(S)} \calW(C)\cdot C(h\cup g),$$
% to measure the contribution of active constraints to the instance
% $\calW$. 

We now define two quantities depending only on $S$ (and importantly,
not on $h$), which will be useful in controlling the expectation and 
variance of $Y$.  The first quantity is an upper bound on $\Var[Y]$:
$$\varest \defeq \sum_{C_1 \sim_{S} C_2} \calW(C_1) \calW(C_2).$$
The second quantity is a lower bound on $\E[Y]$:
$$\meanest \defeq \frac{1}{4}\cdot \sum_{C \in \Act(S)} \calW(C).$$
%
% Again, an upper bound on the $\Var[Y]$ is given by
% We now define two quantities depending only on $S$ (and importantly,
% not on $h$), which will be useful in controlling the expectation and 
% variance of $Y$.  The first quantity is an upper bound on $\Var[Y]$:
% $\varest \defeq \sum_{C_1 \sim_{S} C_2} \calW(C_1) \calW(C_2).$$
% However, now the lower bound on $\E[Y]$ is given by:
% $$\meanest \defeq \frac{1}{4}\cdot \sum_{C \in \Act(S)} \calW(C).$$
%\Sknote{This 1/4 is correct? We should see how the next statement gets affected by the 1/4.}
%The only situation where we will use $\meanest$ is when
%$p: V \setminus S \to \{0,1\}$ is chosen to maximize
%$\E[Y]$; 
%
%
\begin{lemma}
\label{lemma:max-w-sat:uvar_lmean_relation}
Let $S \subseteq V$ be a subset of variables and $h: S \to \{0,1\}$ be an arbitrary partial assignment to $S.$ Let $p, Y, \varest, \meanest$ be as above.
\begin{enumerate}
\item If $\varest \le \delta_0 \epsilon_0^2
  \cdot \meanest^2,$ then $\Pr[Y < (1-\epsilon_0) \E[Y] ] < \delta_0$.
\item If $\varest \ge \delta_0\epsilon_0^2 \cdot \meanest^2 $, then there exists $v \in V\setminus S$ such that 
$$\activedeg_S(v, \calW) \geq \frac{1}{16w^2}\epsilon_0^2\delta_0 \cdot \activedeg_{S}( \calW).$$
\end{enumerate}
\end{lemma}

% We defer the formal proof to
% Section~\ref{section:deferred:maxwsat}. 
The crux of the proof is that independent
of the assignment $h: S \to \{0,1\}$, $\E[Y] \geq \meanest$ and
$\Var(Y) \leq \varest$ (this crucially requires that the rounding is
independent and \emph{smooth}, \emph{i.e.}, $p(v) \in [\nfrac{1}{4},
\nfrac{3}{4}]$ for all $v$; this is why we end up using Trevisan's
rounding procedure in Theorem~\ref{thm:sattrev}).  The first part is
then a simple application of the Chebyshev inequality. For the second
part, we use the assumption that $\varest$ is large, to deduce that
there exists a constraint $C$ such that the total weight of
constraints that share a variable from $V \setminus S$ with $C,$
\emph{i.e.}, $\sum_{C_2 \sim_S C} \calW(C_2),$ is large. It then
follows that at least one variable $v \in C$ must have large
activedegree given $S.$
\begin{proof}
We first prove that $\Var(Y) \leq \varest$.
Recall that the indicator variable $C(h\cup g)$ denotes whether a
constraint $C$ is satisfied by the assignment $h\cup g$,
and note that:
$$Y = \sum_{C \in \Act(S)} \calW(C) \cdot C(h\cup g).$$
Thus, the variance of $Y$ is given by
\begin{align*}
  \Var(Y) &= \sum_{C_{1}, C_{2}\in \Act(S)} \calW(C_1) \calW(C_2) \cdot ( \E[C_1(h\cup g)C_2(h\cup g)] - \E[C_1(h\cup g)]\E[C_2(h\cup g)] )\\
%&\leq \sum_{ C_{j_1}, C_{j_2}, \text{both active given $S$}} \calW_{j_1} \calW_{j_2} ( \E[Y_{j_1} Y_{j_2}] - \E[Y_{j_1}]\E[Y_{j_2}] )\\
&\leq \sum_{ C_{1} \sim_{S} C_{2} } \calW(C_1) \calW(C_2) = \varest,
\end{align*}
where the inequality holds because $\E[C_1(h\cup g)C_2(h\cup g)] -
\E[C_1(h\cup g)]\E[C_2(h\cup g)] \le 1$ for all $C_{1}, C_{2},$ and
$\E[C_1(h\cup g)C_2(h\cup g)] - \E[C_1(h\cup g)]\E[C_2(h\cup g)] = 0$ unless $C_{1} \sim_{S}
C_{2}$ because the rounding is performed independently for all the
variables. 

Moreover, since $p(v) \in [\nfrac{1}{4}, \nfrac{3}{4}]$ for all $v,$
we get that $\E[C(h \cup g)] \ge \nfrac{1}{4}$ for all $C\in \Act(S).$
Thus, we have $\E[Y] \ge \meanest$.  Given this, the first part of the
lemma easily follows from Chebyshev's inequality:
$$\Pr[Y <  (1-\epsilon_0) \E[Y] ] \leq \frac{\Var(Y)}{\epsilon_0^2 (\E[Y])^2} \leq \frac{\varest}{\epsilon_0^2 \meanest^2} \leq \delta_0.$$

%\begin{align*}
%\E[Y] &= \E[\sum_{C_j \text{ active given $S$}} \calW_j Y_j]\\
%&\geq \sum_{C_j \text{ active given $S$}} \calW_j   \E[ Y_j ]\\
%& \geq  \frac{1}{2} \left(\sum_{C_j \text{ active given $S$}} \calW_j \right).
%\end{align*}

For the second part of the lemma, we have:
\begin{align*}
  \delta_0 \epsilon_0^2 \meanest^2 &< \varest =\sum_{ C_{1} \sim_{S}
    C_{2} }
  \calW(C_1) \calW(C_2) \\
  & \leq \sum_{C_{1} \in \Act(S) } \calW(C_1) \left. \sum_{
      C_{2} \sim_{S} C_{1}} \calW(C_2)   \right. \\
  & \leq \left( \sum_{C_{1} \in \Act(S)} \calW(C_1) \right) \cdot
  \max_{C \in \Act(S)}\left. \sum_{ C_{2} \sim_{S} C } \calW(C_2)   \right. \\
  & = 4\cdot \meanest \cdot \max_{C \in \Act(S) }\left. \sum_{ C_{2}
      \sim_{S} C } \calW(C_2)\right. .
\end{align*}
Thus, there exists a constraint $C \in \Act(S)$ such that:
\begin{equation}
\label{eq:lemma:uvar_lmean_relation:maxcut:1}
\left. \sum_{ C_{2} \sim_{S} C} \calW(C_2) \right. \geq \frac{1}{4}\cdot\delta_0
\epsilon_0^2 \cdot \meanest \ge \frac{1}{16w}\delta_0 \epsilon_0^2
\cdot \activedeg_S(\calW), 
\end{equation}
where we used the fact that $\meanest = \frac{1}{4} \cdot ( \sum_{C
  \in \Act(S)} \calW(C) ) \geq \frac{1}{4w} \cdot
\activedeg_S(\calW)$, since we are counting the weight of a constraint
at most $w$ times in the expression $\activedeg_S(\calW).$ Finally,
the LHS of equation~\eqref{eq:lemma:uvar_lmean_relation:maxcut:1} is
at most $\sum_{u \in C \cap (V \setminus S)} \activedeg_{S}(u,
\calW)$.  Thus, there is some $u \in V \setminus S$ with:
$$\activedeg_{S}(u, \calW) \geq \frac{1}{16w^2}\delta_0 \epsilon_0^2 \cdot \activedeg_S(\calW).$$
\end{proof}

\subsection{Algorithm for Simultaneous \maxwsat}
In Figure~\ref{fig:weighted-maxwsat}, we give our algorithm for
simultaneous \maxwsat. The input to the algorithm consists of an
integer $k \ge 1,$ $\eps > 0$, and $k$ instances of \maxwsat, specified
by weight functions $\calW_1,\ldots,\calW_k,$ and target objective
values $c_1,\ldots,c_\ell.$
\begin{figure*}
\begin{tabularx}{\textwidth}{|X|}
\hline
\vspace{0mm}
{\bf Input}: $k$ instances of \maxwsat $\calW_1,\ldots,\calW_k$ on the
variable set $V,$ target objective values $c_1,\ldots,c_k,$ and $\eps > 0$. \\
{\bf Output}: An assignment to $V.$ \\
{\bf Parameters:} $\delta_0 = \frac{1}{10k}$, $\epsilon_0 =
\frac{\eps}{2}$, $\gamma=\frac{\eps_0^2\delta_0}{16w^2}$, $t =
\frac{2k}{\gamma}\cdot\log\left( \frac{11}{\gamma}\right).$
\begin{enumerate}[itemsep=0mm, label=\arabic*.]
\item Initialize $S \leftarrow \emptyset$.
\item For each instance $\ell \in [k]$, initialize $\cnt_\ell
  \leftarrow 0$ and $\flag_\ell \leftarrow \true.$
\item Repeat the following until for every $\ell \in [k]$, either $\flag_\ell = \false$ or
$\cnt_\ell  = t$:
\label{item:alg:max-w-sat:loop}
\begin{enumerate}
\item For each $\ell \in [k]$, compute
$\varest_\ell = \sum_{C_1 \sim_S C_2} \calW_\ell(C_1) \calW_\ell(C_2).$
\item For each $\ell \in [k],$ compute
$\meanest_\ell = \frac{1}{4} \sum_{C \in \Act(S)} \calW_\ell(C).$
\item For each $\ell \in [k],$ if $\varest_\ell \geq \delta_0
  \epsilon_0^2 \cdot \meanest_\ell^2 $, then set $\flag_\ell = \true$,
  else set $\flag_\ell = \false$.
\item Choose any $\ell \in [k]$, such that
$\cnt_\ell < t$ AND $\flag_\ell = \true$ (if any):
\label{item:alg:max-w-sat:high-var}
\begin{enumerate}
\item
\label{item:findv}
Find a variable $v \in V$ such that $\activedeg_S(v, \calW_\ell)
  \geq \gamma \cdot
  \activedeg_S(\calW_\ell).$ %If no such variable exists, output \fail.
\item Set $S \leftarrow S \cup \{v\}.$ We say that $v$ was brought
  into $S$ because of instance $\ell$.
\item Set $\cnt_\ell \leftarrow \cnt_\ell + 1$.
\end{enumerate}
\end{enumerate}
\item For each partial assignment $h_0: S \to \{0,1\}$:
\label{item:alg:max-w-sat:partial}
\begin{enumerate}
\item 
\label{item:alg:max-w-sat:lp}
If there is a feasible solution $\vec{t}, \vec{z}$ to the LP in
Figure~\ref{fig:maxwsat-lp2}, set $p = \smooth(\vec{t})$.
If not, return to Step~\ref{item:alg:max-w-sat:partial} and proceed to the next $h_0$.
\item 
  \label{item:alg:max-w-sat:random}
Define $g: V\setminus S \to \{0,1\}$ by independently sampling
$g(v) \in \{0,1\}$ with $\E[g(v)] = p(v),$ for each $v \in V\setminus S.$
\item 
\label{item:alg:max-w-sat:innerh}
For each $h: S \to \{0,1\},$ compute $\out_{h,g} = \min_{\ell \in [k]}
\frac{\val(h \cup g,\calW_\ell)}{c_\ell}$. If $c_\ell = 0$ for some
$\ell \in [k],$ we interpret $\frac{\val(h\cup g,\calW_\ell)}{c_\ell}$ as
$+\infty.$
\end{enumerate}
\item Output the largest $\out_{h,g}$ seen, and the assignment $h \cup g$.
\end{enumerate}
\\ 
\hline
\end{tabularx}
\caption{Algorithm {\algowwsat} for approximating weighted simultaneous
  \maxwsat}
  \label{fig:weighted-maxwsat}
\end{figure*}

% \begin{figure*}
% \begin{tabularx}{\textwidth}{|X|}
% \hline
% \vspace{-5mm}
% \begin{center}
% \[\begin{array}{rllr}
%  \sum_{C\in \calC} \calW_\ell(C)\cdot z_C & \geq & c_\ell & \forall \ell\in [k]\\  %(1-\nfrac{\eps}{2})\cdot
%  \sum_{v\in C^{+}} y_v + \sum_{v\in C^{-}} (1-y_v) &\geq & z_C & \forall C\in \calC\\
%  1\geq z_C &\geq & 0 & \forall C\in \calC\\
%  1\geq y_v &\geq & 0 & \forall v \in V\setminus S\\
%  y_v & = & h_0(v) & \forall v \in S
% \end{array}\]
% \end{center}
% \vspace{-30pt}
% \\
% \hline
% \end{tabularx}
% \caption{Linear Program $\LLsat(h_0)$, for a given partial assignment $h_0: S \to \{0,1\}$}
%   \label{fig:maxwsat-lp}
% \end{figure*}
% \Ssnote{remove $\lp_\rho. Change y_v$. Break up into LPs}

\begin{figure*}
\begin{tabularx}{\textwidth}{|X|}
\hline
\vspace{-26pt}
\begin{center}
\[\begin{array}{rrllr}
 \sum_{v\in C^{+}} t_v + \sum_{v\in C^{-}} (1-t_v) &\geq & z_C & \forall C\in \calC\\
 1\geq z_C &\geq & 0 & \forall C\in \calC\\
 1\geq t_v &\geq & 0 & \forall v \in V\setminus S\\
 t_v & = & h_0(v) & \forall v \in S \\
\end{array}\]
\vspace{-24pt}
\end{center}
\\
\hline
\end{tabularx}
\caption{Linear program $\LLwsat_1(h_0),$ for a given partial
  assignment $h_0: S \to \{0,1\}$}
  \label{fig:maxwsat-lp1}
\end{figure*}

\begin{figure*}
\begin{tabularx}{\textwidth}{|X|}
\hline
\vspace{-26pt}
\begin{center}
\[\begin{array}{rrllr}
\sum_{C\in \calC} \calW_\ell(C)\cdot z_C  \geq c_\ell & \forall \ell\in [k]\\
\vect, \vecz  \text{ satisfy }  \LLwsat_1(h_0).&
\end{array}\]
\vspace{-24pt}
\end{center}
\\
\hline
\end{tabularx}
\caption{Linear program $\LLwsat_2(h_0)$ for a given partial assignment
  $h_0: S \to \{0,1\}$}
  \label{fig:maxwsat-lp2}
\end{figure*}

\subsection{Analysis of Algorithm {\algowwsat}}
\label{sec:max-w-sat:analysis}
%\Sknote{Removed the lemma .. made it just text.}
%
It is easy to see that the algorithm always terminates in polynomial
time. Part 2 of Lemma~\ref{lemma:max-w-sat:uvar_lmean_relation} implies
that that Step~\ref{item:findv} always succeeds in finding a variable
$v$. Next, we note that Step~\ref{item:alg:max-w-sat:loop} always
terminates. Indeed, whenever we find an instance $\ell \in [k]$ in
Step~\ref{item:alg:max-w-sat:high-var} such that $\cnt_\ell < t$ and
$\flag_\ell = \true,$ we increment $\cnt_\ell.$ This can happen only
$tk$ times before the condition $\cnt_\ell < t$ fails for all $\ell
\in [k].$ Thus the loop must terminate within $tk$ iterations.

Let $S^\star$ denote the final set $S$ that we get at the end of
Step~\ref{item:alg:max-w-sat:loop} of {\algowwsat}. To analyze the
approximation guarantee of the algorithm, we classify instances
according to how many vertices were brought into $S^\star$ because of
them.
\begin{definition}[Low and high variance instances]
At the completion of Step~\ref{item:alg:max-w-sat:high-var} in Algorithm
{\algowwsat}, if $\ell \in [k]$ satisfies $\cnt_\ell = t$, we call
instance $\ell$ a {\em high variance} instance.  Otherwise we call
instance $\ell$ a \emph{low variance} instance.
\end{definition}
%

% \Sknote{Let us do the perturbation to $h^\star \cup g$, as we do later
%   in the more advanced algorithms.}

% \Sknote{Maybe we should define $S^\star$ to be the set $S$ at the end
%   of the algorithm ? This might make notation much much better,
%   because we keep talking about earlier $S$'s from the middle of the
%   algorithm}

At a high level, the analysis will go as follows: First we analyze
what happens when we give the optimal assignment to $S^\star$ in
Step~\ref{item:alg:max-w-sat:partial} For low variance instances, the
fraction of the constraints staisfied by the LP rounding will
concentrate around its expectation, and will give the desired
approximation. For every high variance instance, we will see that many
of its ``heavy-weight" vertices were brought into $S^\star$, and we
will use this to argue that we can satisfy a large fraction of the
constraints from these high variance instances by suitably perturbing
the optimal assignment to $S^\star$ to these ``heavy-weight''
vertices. It is crucial that this perturbation is carried out without
significantly affecting the value of the low variance instances.

Let $f^\star: V \to \{0,1\}$ be an assignment such that $\val(f^\star,
\calW_\ell) \geq c_\ell$ for each $\ell$. Let $h^\star = f^\star
|_{S^\star}$. Claim 1 from Theorem~\ref{thm:sattrev} implies that
$\LLwsat_2(h^\star)$ has a feasible solution.
% We first show that for low variance instances, when $g :
% V \setminus S^\star \to \{0,1\}$ is picked uniformly at random, then
% $\val(h^\star \cup g, \calW_\ell)$ is at least $\left(\nfrac{1}{2} -
%   \nfrac{\epsilon}{2}\right) \cdot c_\ell$ with high probability.
% \Sknote{Need to write this lemma .. basically take the low variance
% analysis from the proof of the theorem.}
For low variance instances, by combining Theorem~\ref{thm:sattrev} and
Lemma~\ref{lemma:max-w-sat:uvar_lmean_relation}, we show that
$\val(h^\star \cup g, \calW_\ell)$ is at least $\left(\nfrac{3}{4} -
  \nfrac{\epsilon}{2}\right) \cdot c_\ell$ with high probability.
\begin{lemma}
\label{lemma:maxwsat:lowvariance}
Let $\ell \in [k]$ be any low variance instance. 
Let $\vec{t}, \vec{z}$ be a feasible solution to $\LLwsat_2(h^\star)$.
Let $p = \smooth(\vec{t})$. 
Let $g : V \setminus S^\star \to \{0,1\}$ be such that each $g(v)$ is
sampled independently with $\E[g(v)] = p(v)$. Then the assignment
$h^\star \cup g$ satisfies:
$$ \Pr_{g} \left[ \val(h^\star
\cup g,\calW_\ell) \geq
  (\nicefrac{3}{4}-\nicefrac{\epsilon}{2})\cdot c_\ell \right] \geq 1-\delta_0.$$
\end{lemma}
\begin{proof}
  Since $\ell$ is a {\em low variance} instance, $\flag_{\ell} =
  \false$ when the algorithm terminates. Thus $\varest_\ell <
  \delta_0\epsilon_0^2 \cdot \meanest_\ell^2$.  Let $g: V \to \{0,1\}$
  be the random assignment picked in
  Step~\ref{item:alg:max-w-sat:random}. Define the random
  variable \[Y_\ell \defeq \val(h^\star \cup g,\calW_\ell) -
  \val(h^\star,\calW_\ell).\]

By Lemma~\ref{lemma:max-w-sat:uvar_lmean_relation}, we know that with probability
  at least $1-\delta_0,$ we have $ Y_\ell \ge (1- \epsilon_0) \E[Y_\ell].$
  Thus, with probability at least $1-\delta_0,$ we have,
\begin{align*}
  \val(h^\star\cup g,\calW_\ell) & = \val(h^\star,\calW_\ell) + Y_\ell \ge
  \val(h^\star,\calW_\ell)  + (1- \epsilon_0) \E[Y_\ell]\\
  & \ge (1-\epsilon_0) \cdot \av[\val(h^\star,\calW_\ell) + Y_\ell] =
  (1-\epsilon_0)\cdot \av[\val(h^\star\cup g,W_\ell)] \\
& \ge \nfrac{3}{4}\cdot (1-\epsilon_0) \cdot \sum_{C \in \calC} \calW_\ell(C) z_C \ge
  \left(\nfrac{3}{4} - \nfrac{\eps}{2}\right) \cdot c_\ell ,
\end{align*}
where the last two inequalities follow from Claim 2 in
Theorem~\ref{thm:sattrev} and the constraints in $\LLwsat_2$
respectively.
%\Anote{Used $\eps_0 = \eps$}
\end{proof}

Now we analyze the high variance instances. We prove the following
lemma that proves that at the end of the algorithm, the activedegree
of high variance instances is small, and is dominated by the
activedegree of any variable that was included in $S$ ``early on".
\begin{lemma}
\label{lemma:earlyvariables}
%Let $\gamma= \frac{\epsilon_0^2\delta_0}{16}.$
For all high variance instances $\ell\in [k],$ we have 
\begin{enumerate}
\item $\activedeg_{S^\star}(\calW_\ell) \leq w(1-\gamma)^t.$ 
\item For each of the first $\nicefrac{t}{2}$ variables that were brought
  inside $S^\star$ because of instance $\ell,$ the total weight of constraints
  incident on each of that variable and totally contained inside $S^\star$ is at
  least $10\cdot\activedeg_{S^\star}(\calW_\ell).$
\end{enumerate}
\end{lemma}
%The proof is quite simple.
% and is deferred to
% Section~\ref{section:deferred:maxwsat}. 
The crucial observation is
that when a variable $u$ is brought into $S$ because of an instance
$\ell,$ the activedegree of $u$ is at least a $\gamma$ fraction of the
total activedegree of instance $\ell.$ Thus, the activedegree of
instance $\ell$ goes down by a multiplicative factor of $(1-\gamma).$
This immediately implies the first part of the lemma. For the second
part, we use the fact that $t$ is large, and hence the activedegree of
early vertices must be much larger than the final activedegree of
instance $\ell$.
\begin{proof}
% \Sknote{This whole proof should be really short, because the idea is simple. Somehow there is too much notation right now. I tried some simplifications, please see if it is ok.} 
  Consider any {\em high variance} instance $\ell \in[k]$. Initially,
  when $S=\emptyset,$ we have $\activedeg_\emptyset(\calW_\ell) \leq
  w$ since the weight of every constraint is counted at most $w$ times, once for
  each of the 2 active variables of the constraint,
% we count every constraint weight at most twice in the
%   active degree 
  and $\sum_{C\in \calC } \calW_\ell(C)= 1$. 
%Whenever we bring
%  variables because of other instances into $S$, the active degree of any
%  variable $v \in V$ in the instance $\ell$ w.r.t. the current $S$
%  never goes up i.e.
For every $v$, note that $\activedeg_{S_2}(v,\calW_\ell) \leq
  \activedeg_{S_1}(v,\calW_\ell)$ whenever $S_1 \subseteq S_2$. 
% Thus, as $S$ keeps growing through the execution of the algorithm, $\activedeg_S(\calW_\ell)$ keeps 
% reducing.

  % We first prove the first part. 
  Let $u$ be one of the variables that ends up in $S^\star$ because of
  instance $\ell.$ Let $S_u$ denote the set $S \subseteq S^\star$ just
  before $u$ was brought into $S^\star$.  When $u$ is added to $S_u$,
  we know that $\activedeg_{S_u}(u, \calW_\ell) \geq \gamma \cdot
  \activedeg_{S_u}(\calW_\ell).$ Hence, $\activedeg_{S_u \cup
    \{u\}}(\calW_\ell) \le \activedeg_{S_u}(\calW_\ell) -\activedeg_{S_u}(u,
  \calW_\ell) \le (1-\gamma) \cdot \activedeg_{S_u}(\calW_\ell).$ Since
  $t$ variables were brought into $S^\star$ because of instance
  $\ell,$ and initially $\activedeg_{\emptyset}(\calW_\ell) \le w,$ we
  get $\activedeg_{S^\star}(\calW_\ell) \le w(1-\gamma)^t.$
% $\activedeg_S(\calW_\ell)$ goes down by a
%   multiplicative factor of $(1-\gamma).$ But this follows immediately
%   from the fact that $u$ is brought into $S$ because of instance
%   $\ell$ precisely when $\activedeg_S(u, \calW_\ell) \geq \gamma \cdot
%   \activedeg_S(\calW_\ell)$.  \Anote{Does it need one more line of
%     explanation?}

%We now prove the second part. 
  Now, let $u$ be one of the first $\nfrac{t}{2}$ variables that ends
  up in $S^\star$ because of instance $\ell.$ Since at least
  $\nfrac{t}{2}$ variables are brought into $S^\star$ because of
  instance $\ell,$ after $u,$ as above, we get
  $\activedeg_{S^\star}(\calW_\ell) \le (1-\gamma)^{\nfrac{t}{2}}
  \cdot \activedeg_{S_u}(\calW_\ell).$ Combining with
  $\activedeg_{S_u}(u, \calW_\ell) \geq \gamma \cdot
  \activedeg_{S_u}(\calW_\ell),$ we get $\activedeg_{S_u}(u,
  \calW_\ell) \geq \gamma (1-\gamma)^{-\nfrac{t}{2}}
  \activedeg_{S^\star}(\calW_\ell),$
% For a variable $u$ that ends up in $S^\star$ at the end of the algorithm, let $S_u$ denote the  set $S \subseteq S^\star$ just before $u$ was brought into $S^\star$.
%   Let $u_1,u_2,\cdots,u_t$ (in that order) be the variables brought into $S^\star$
%   because of instance $\ell$. We know that
%   $\activedeg_{S^\star}(\calW_\ell) \le \activedeg_{S_{u_t}}(\calW_\ell) \le
%   (1-\gamma)^{t-i}\activedeg_{S_{u_i}}(\calW_\ell).$ Moreover, we know
%   that for any $i,$ we have $\activedeg_{S_{u_i}}(u_i,\calW_\ell) \geq
%   \gamma \cdot\activedeg_{S_{u_i}}(\calW_\ell).$ Hence, for any $i \in
%   \floor{t/2},$ we have \[\activedeg_{S_{u_i}}(u_i,\calW_\ell) \ge
%   \gamma (1-\gamma)^{i-t} \activedeg_{S^\star}(\calW) \ge
%   \tfrac{\gamma}{(1-\gamma)^{t/2}} \cdot \activedeg_{S^\star}(\calW_\ell),\]
  which is at least $11\cdot\activedeg_{S^\star}(\calW_\ell),$ by the
  choice of parameters. Since any constraint incident on a vertex in
  $V \setminus S^\star$ contributes its weight to
  $\activedeg_{S^\star}(\calW_\ell),$ the total weight of constraints
  incident on $u$ and totally contained inside $S^\star$ is at least
  % $11\cdot\activedeg_{S^\star}(\calW_\ell) -
  % \activedeg_{S^\star}(\calW_\ell),$ which is at least
  $10\cdot\activedeg_{S^\star}(\calW_\ell)$ as required.  
% \Sknote{Old
%     proof below}\Anote{commented in the file}
%
%Let $S_u$ be the
%  set of variables in $S$ just before $u$ was picked. We know that
%  $\activedeg_{S_{u}}(u,\calW_\ell) \geq \gamma
%  \cdot\activedeg_{S_{u}}(\calW_\ell)$. Thus,
%\begin{align*}
%  \activedeg_{S_u \cup \{u\}}(\calW) & = \sum_{v \in V\setminus \{S
%    \cup
%    \{u\}\}} \activedeg_{S_u \cup \{u\}}(v, \calW) \\
%  & \le \activedeg_{S_u}(\calW) - \activedeg_{S_{u}}(u,\calW_\ell) \le
%  (1-\gamma) \activedeg_{S_u}(\calW).
%\end{align*}
%By induction, we get $\activedeg_S(\calW_\ell) \le (1-\gamma)^t
%  \cdot \activedeg_{\emptyset}(\calW_\ell) \le 2 (1-\gamma)^t$.
%
%  Let $u_1,u_2,\cdots,u_t$ be the set of variables brought into $S$
%  because of instance $\ell$ and let $S_{u_1}\subseteq
%  S_{u_2}\subseteq \cdots \subseteq S_{u_t}$ be the set of variables
%  in set $S$ just before we bring variables $u_1,u_2,\cdots,u_t$
%  respectively. From the above induction argument, 
\end{proof}%\Ssnote{@Swastik: Please read this part.}

% By lemma~\ref{lemma:uvar_lmean_relation}, we know that
%   $\activedeg_{S_{u_1}}(u_1,\calW_\ell) \geq \gamma
%   \cdot\activedeg_{S_{u_1}}(\calW_\ell)$. Hence, when we bring that
%   variable into set $S$ the active degree of instance $\ell$ reduces
%   to $(1-\gamma)$ fraction of its current active degree which is at
%   most $2(1-\gamma)$. Continuing this way, we have
%   $\activedeg_{S_{u_i}}(\calW_\ell) \leq 2(1-\gamma)^{i-1}$ for all
%   $i\in [t]$. And hence $\activedeg_S(\calW_\ell) \leq (1-\gamma)^t
%   \cdot \activedeg_{\emptyset}(\calW_\ell) \le 2 (1-\gamma)^t$.

% \Ssnote{Rewrite the part below this.}
% Since $u_t$ is the last vertex that we brought inside $S$, we have $\wt(u_t) \geq \gamma\cdot\activedeg_S(\calW_\ell)$. In general, we have $\wt_{\calW_\ell}(u_i) \geq \gamma(1+\gamma)^{t-i}\cdot\activedeg_S(\calW_\ell)$. And hence, for all $i\in [t/2]$, we have $ \wt_{\calW_\ell}(u_i) \geq \gamma(1+\gamma)^{t/2}\cdot\activedeg_S(\calW_\ell)  = \frac{\gamma}{\sqrt{\delta}}\cdot\activedeg_S(\calW_\ell)$ which is at least $7\cdot\activedeg_S(\calW_\ell)$. Therefore, total weight of $u_i$ inside $S$ is at least $7\cdot\activedeg_S(\calW_\ell) - \activedeg_S(\calW_\ell)$ which is at least $6\cdot\activedeg_S(\calW_\ell)$ as required.

% \Anote{We need to define weight associated with a variable in an instance, using $\wt_{\calW_\ell}$ here}

% \Anote{Will set the expression $2(1-\gamma)^t$ later, update : I don't think we need this}

%\Sknote{ Making changes here for the new perturb algorithm. Please check}
We now describe a procedure {\sc Perturb} (see
Figure~\ref{fig:maxcut-perturb}) which takes $h^\star: S^\star \to \B$ and
$g : V \setminus S^\star \to \B$, and produces 
a new $h : S^\star \to \B$ such that
for all (low variance as well as  high variance) instances $\ell \in [k]$,
$\val(h \cup g, \calW_\ell)$ is not much smaller than $\val(h^\star \cup g, \calW_\ell)$, and 
%$$\val(h \cup g, \calW_\ell) \ge \left(1 - \frac{\epsilon}{2}\right) \cdot \val(h^\star \cup g, \calW_\ell),$$
furthermore, for all high variance instances $\ell \in [k]$, $\val(h
\cup g, \calW_\ell)$ is large. The procedure works by picking a
special variable in $S^\star$ for every { high variance} instance and
perturbing the assignment of $h^\star$ to these special variables. The
crucial feature used in the perturbation procedure, which holds for
\maxwsat (but not for \maxand), is that it is possible to satisfy a
constraint by just changing one of the variables it depends on.  The
partial assignment $h$ is what we will be using to argue that
Step~\ref{item:alg:max-w-sat:partial} of the algorithm produces a good
Pareto approximation.  More formally, we have the following Lemma.
% \Anote{Are we repeating above things, look at
%   Lemma~\ref{lemma:maxcut_perturbationeffect} } \Sknote{ removed the
%   repetition.}
 %\Sknote { Perturb algorithm + proof of lemma should be changed. should match notation from above.}

\begin{figure*}
\begin{tabularx}{\textwidth}{|X|}
\hline
\vspace{0mm}
{\bf Input}: $h^\star: S^\star \to \B$ and
$g : V \setminus S^\star \to \B$\\
{\bf Output}: A perturbed assignment $h :  S^\star \to \{0,1\}.$ 
\begin{enumerate}[itemsep=0mm, label=\arabic*.]
\item Initialize $h \leftarrow h^\star.$
\item For $\ell = 1, \ldots, k$, if instance $\ell$ is a high variance
  instance case (i.e., $\cnt_\ell = t$), we pick a special variable
  $v_\ell \in S^\star$ associated to this instance as follows:
\begin{enumerate}
\item Let $B = \{ v \in V \mid \exists \ell \in [k] \mbox{ with }
  \sum_{C \in \calC, C \owns v} \calW_{\ell}(C)\cdot C(h \cup g) \geq
  \frac{\epsilon}{2k} \cdot\val(h\cup g, \calW_\ell) \}$.  Since the weight of
  each constraint is counted at most $w$ times, we know that $|B| \leq \frac{2wk^2}{\epsilon}$.
\item Let $U$ be the set consisting of the first $t/2$ variables
brought into $S^\star$ because  of instance $\ell$.
\item Since $\nfrac{t}{2} > |B|+ k$, there exists some $u \in U$
  such that $u \not\in B \cup \{v_1, \ldots, v_{\ell-1}\}$.  We define
  $v_\ell$ to be $u$.
\item By Lemma~\ref{lemma:earlyvariables}, the total $\calW_\ell$ weight of constraints that are incident on $v_\ell$
  and only containing variables from $S^\star$ is at least
 % $\wt_{\calW_\ell}(v_\ell) \geq
  $10\cdot\activedeg_{S^\star}(\calW_\ell)$. We update $h$ by 
setting $h(v_\ell)$ to be that value from $\{0,1\}$
such that at least half of the $\calW_\ell$ weight of these constraints is
  satisfied.
\end{enumerate}
\item Return the assignment $h.$
%%\vspace{-3mm}
\end{enumerate}
% \Anote{ Shall we call the final $f|_S$ with some different name, as we
%   are not using the notation $f|_S$ in the later sections. Also, $f$
%   seems more general name for this special assignment}
\\
\hline
\end{tabularx}
\caption{Procedure {\sc Perturb} for perturbing the optimal
  assignment}
  \label{fig:maxcut-perturb}
\end{figure*}

\begin{lemma}
\label{lemma:maxcut_perturbationeffect} 
For the assignment $h$ obtained from Procedure {\sc Perturb} (see
Figure~\ref{fig:maxcut-perturb}), for each $\ell\in [k]$, $\val(h\cup
g,\calW_\ell) \geq (1-\nicefrac{\epsilon}{2}) \cdot \val(h^\star \cup
g, \calW_\ell)$. Furthermore, for each {\em high variance} instance
$\calW_\ell$, $\val(h\cup g,\calW_\ell) \geq
4 \cdot \activedeg_{S^\star}(\calW_\ell).$
\end{lemma}
\begin{proof}
  Consider the special variable $v_\ell$ that we choose for {\em high
    variance} instance $\ell \in [k]$. Since $v_\ell \notin B,$ the
  constraints incident on $v_\ell$ only contribute at most a
  $\nfrac{\eps}{2k}$ fraction of the objective value in each
  instance. 
  % all variables $v$ such that if we change the value of $v$ given by
  % $h\cup g$ then it reduces the objective value of some instance by at
  % least $\frac{\epsilon}{2k}$ fraction of it's current value. Since
  % $v_\ell$ is not present in $B$, 
  Thus, changing the assignment $v_\ell$ can reduce the value of any
  instance by at most a $\frac{\epsilon}{2k}$ fraction of their
  current objective value. Also, we pick different special variables
  for each {\em high variance} instance. Hence, the total effect of
  these perturbations on any instance is that it reduces the objective
  value (given by $h^\star \cup g$) by at most $1 - (1-
  \frac{\epsilon}{2k})^k \le \frac{\epsilon}{2}$ fraction.  Hence for
  all instances $\ell \in [k]$, $\val(h\cup g,\calW_\ell) \geq
  (1-\nicefrac{\epsilon}{2})\cdot\val(h^\star\cup g,\calW_\ell)$.

  For a {\em high variance instance} $\ell \in [k]$, since $v_\ell \in
  U,$ the variable $v_\ell$ must be one of the first $\nfrac{t}{2}$
  variables brought into $S^\star$ because of $\ell.$ Hence, by
  Lemma~\ref{lemma:earlyvariables} the total weight of constraints
  that are incident on $v_\ell$ and entirely contained inside
  $S^\star$ is at least
  $10\cdot\activedeg_{S^\star}(\calW_\ell)$. Hence, there is an
  assignment to $v_\ell$ that satisfies at least at least half the
  weight of these \maxwsat constraints\footnote{This is not true if
    they are \maxand constraints.} in $\ell.$ At the end of the
  iteration, when we pick an assignment to $v_\ell,$ we have
  $\val(h\cup g, \calW_\ell) \ge
  5\cdot\activedeg_{S^\star}(\calW_\ell).$ Since the later
  perturbations do not affect value of this instance by more than
  $\nfrac{\epsilon}{2}$ fraction, we get that for the final assignment
  $h$, $\val(h\cup g,\calW_\ell) \ge (1-\nfrac{\epsilon}{2})\cdot 5
  \cdot \activedeg_{S^\star}(\calW_\ell) \geq 4\cdot
  \activedeg_{S^\star}(\calW_\ell).$
%  Since the total weight of constraints incident on $V \setminus S$ is
%  bounded by $\activedeg_{S^\star}(\calW_\ell),$ we get $\val(h, \calW_\ell)
%  \ge 3\cdot \activedeg_{S^\star}(\calW_\ell).$
\end{proof}

Given all this, we now show that with high probability the algorithm
finds an assignment that satisfies, for each $\ell \in [k]$, at least
$(\nfrac{3}{4}-\epsilon)\cdot c_\ell$ weight from instance
$\calW_\ell$. The following theorem immediately implies
Theorem~\ref{theorem:results:maxwsat}.
\begin{theorem}
Let $w$ be a constant. Suppose we're given $\eps \in (0,\nfrac{2}{5}],$ $k$ simultaneous
\maxwsat instances $\calW_1, \ldots, \calW_\ell$ on $n$ variables, and
target objective value $c_1, \ldots,c_k$ with the guarantee that there
exists an assignment $f^\star$ such that for each $\ell \in [k],$ we
have $\val( f^\star, \calW_\ell) \ge c_\ell.$ Then, the algorithm
{\algowwsat} runs in time $2^{O(\nfrac{k^3}{\eps^2}\log
(\nfrac{k}{\eps^2}))}\cdot\poly(n),$ and with probability at least
$0.9,$ outputs an assignment $f$ such that for each $\ell \in [k],$ we
have, $\val(f, \calW_\ell) \ge \left( \nfrac{3}{4} -\epsilon\right)
\cdot c_\ell.$ 
\end{theorem}
\begin{proof}
  Consider the iteration of Step~\ref{item:alg:max-w-sat:partial} of the
  algorithm when $h_0$ is taken to equal $h^\star$. Then, by Part 1 of
  Theorem~\ref{thm:sattrev}, the LP in Step~\ref{item:alg:max-w-sat:lp}
  will be feasible (this uses the fact that $\val(f^\star, \calW_\ell)
  \geq c_\ell$ for each $\ell$).

  By Lemma~\ref{lemma:maxwsat:lowvariance} and a union bound, with
  probability at least $1 -k\delta_0 > 0.9$, over the choice of $g$,
  we have that for {\em every} low variance instance $\ell \in [k]$,
  $\val(h^\star \cup g, \calW_\ell) \geq (\nfrac{3}{4} -
  \nfrac{\epsilon}{2}) \cdot c_{\ell}$. Henceforth we assume that the
  assignment $g$ sampled in Step~\ref{item:alg:max-w-sat:random} of the
  algorithm is such that this event occurs. Let $h$ be the output of
  the procedure {\sc Perturb} given in Figure~\ref{fig:maxcut-perturb}
  for the input $h^\star$ and $g.$
% \Anote{Should we name for this special g and h?}.
  By Lemma~\ref{lemma:maxcut_perturbationeffect}, $h$ satisfies
\begin{enumerate}
\item\label{item:max-w-sat:proof:guarantees:1} For every instance $\ell \in [k]$, $\val(h\cup g,\calW_\ell) \geq (1-\nfrac{\eps}{2})\cdot \val(h^\star\cup g,\calW_\ell).$
\item \label{item:max-w-sat:proof:guarantees:2}For every high variance
  instance $\ell \in [k]$, $\val(h\cup g,\calW_\ell) \geq
 4 \cdot  \activedeg_{S^\star}(\calW_\ell).$
\end{enumerate}
%Since we are trying all assignment to $S^\star$, we will be 
We now show that the desired Pareto approximation behavior is achieved
when $h$ is considered as the partial assignment in
Step~\ref{item:alg:max-w-sat:innerh} of the algorithm. We analyze the
guarantee for low and high variance instances separately.

% By Lemma~\ref{item:alg:max-w-sat:random} and a union bound, with probability at least
% $1 - k \delta_0 > 0.9$, the $g$ chosen in Step~\ref{ } will be such that for every low variance 
% instance $\ell \in [k]$, we have $\val(h^\star \cup g, \calW_\ell) \geq 
% (\nfrac{3}{4} - \nfrac{\epsilon}{2}) \cdot c_{\ell}$. 

% Henceforth we assume
% this happens.

% Now let $h$ be the output of the procedure  {\sc Perturb}~\ref{fig:maxcut-perturb} when we feed it $h^\star$ and $g$\Anote{Should we name for this special g and h?}. By Lemma~\ref{lemma:maxcut_perturbationeffect}, it satisfies
% \begin{enumerate}
% \item\label{item:maxwsat:proof:guarantees:1} For every instance $\ell \in [k]$, $\val(h\cup g,\calW_\ell) \geq (1-\nfrac{\eps}{2})\cdot \val(h^\star\cup g,\calW_\ell)$
% \item \label{item:maxwsat:proof:guarantees:2}For every high variance instance $\ell \in [k]$, $\val(h\cup g,\calW_\ell) \geq 4 \cdot \activedeg_{S^\star}(\calW_\ell)$
% \end{enumerate}
%Since we are trying all assignment to $S^\star$, we will be 
% We now show that the desired Pareto approximation behavior is achieved When this $h$ is considered in Step~\ref{item:innerh} of the algorithm. We analyze the guarantee for low and high variance instances separately.
% when we consider this perturbed assignment in Step~\ref{item:alg:max-cut:partial}. 

For any {\em low variance} instance $\ell \in [k],$ from
property~\ref{item:max-w-sat:proof:guarantees:1} above, we have
$\val(h\cup g,\calW_\ell) \geq (1-\nfrac{\eps}{2})\cdot
\val(h^\star\cup g,\calW_\ell)$. Since we know that $\val(h^\star \cup
g, \calW_\ell) \geq (\nfrac{3}{4}-\nfrac{\eps}{2})\cdot c_\ell$, we
have $\val(h\cup g,\calW_\ell) \geq (\nfrac{3}{4}-\eps)\cdot c_\ell$.

% Consider a {low variance} instance $\ell$, if any. From property~\ref{item:maxwsat:proof:guarantees:1} above, we have $\val(h\cup g,\calW_\ell) \geq (1-\nfrac{\eps}{2})\cdot \val(h^\star\cup g,\calW_\ell)$. Since we know that $\val(h^\star \cup g, \calW_\ell) \geq (\nfrac{3}{4}-\nfrac{\eps}{2})\cdot c_\ell$, we have  $\val(h\cup g,\calW_\ell) \geq (\nfrac{3}{4}-\eps)\cdot c_\ell$.

For every high variance instance $\ell \in [k],$ since $h^\star =
f^\star|_{S^\star},$ for any $g$ we must have,
\[\val(h^\star \cup g, \calW_\ell) \ge \val(f^\star, \calW_\ell) -
\activedeg_{S^\star}(\calW_\ell) \ge c_\ell - \activedeg_{S^\star}(\calW_\ell) .\]
Combining this with properties~\ref{item:max-w-sat:proof:guarantees:1}
and \ref{item:max-w-sat:proof:guarantees:2} above, we get,
\begin{align*}
\val(h \cup g, \calW_\ell) &\ge \left( 1 - \nfrac{\epsilon}{2}
\right) \cdot \max\{c_\ell - \activedeg_{S^\star}(\calW_\ell), 4 \cdot
\activedeg_{S^\star}(\calW_\ell)\} \\
&\ge \left( \nfrac{3}{4} - \epsilon \right) \cdot c_\ell.
\end{align*}

Thus, for all instances $\ell \in [k]$, we get $\val(h \cup g) \ge
\left(\nfrac{3}{4}-\eps\right) \cdot c_{\ell}.$ Since we are taking
the best assignment $h\cup g$ at the end of the algorithm {\algowwsat}, the theorem follows.
% For every high variance instance $\ell \in
% [k],$ since $h^\star = f^\star|_{S^\star},$
% \[\val(h^\star \cup g, \calW_\ell) \ge \val(f^\star, \calW_\ell) -
% \activedeg_{S^\star}(\calW_\ell) \ge c_\ell - \activedeg_{S^\star}(\calW_\ell) .\]
% Combining this with properties~\ref{item:maxwsat:proof:guarantees:1}
% and \ref{item:maxwsat:proof:guarantees:2} above, we get,
% \begin{align*}
% \val(h \cup g, \calW_\ell) &\ge \left( 1 - \nfrac{\epsilon}{2}
% \right) \cdot \max\{c_\ell - \activedeg_{S^\star}(\calW_\ell), 4\cdot \activedeg_{S^\star}(\calW_\ell)\} \\
% &\ge \left( \nfrac{3}{4} - \epsilon \right) \cdot c_\ell.
% \end{align*}
%
% Thus, for all instances $\ell \in [k]$, we get $\val(h \cup g) \ge
% \left(\nfrac{3}{4}-\eps\right) \cdot c_{\ell}.$
%
% Since we are taking the best assignment $h\cup g$ at the end of the algorithm {\sc WeightedMC}, the theorem follows.
% \Sknote{ Should go through the above and change the name of MAXCUT to MAXSAT.
% Also remove "weighted" from everywhere. Note that to save
% repetition, we are using the perturbation procedure +  the
% perturbation lemma from MAXCUT.}

{\bf Running time : } Running time of the algorithm is
$2^{O(kt)}\cdot\poly(n)$ which is $2^{O(\nfrac{k^3}{\eps^2}\log
(\nfrac{k}{\eps^2}))}\cdot\poly(n).$
\end{proof}

\nocite{*}

\bibliographystyle{alpha}
\bibliography{refs}

\newcommand{\etalchar}[1]{$^{#1}$}
\begin{thebibliography}{AGK{\etalchar{+}}11}

\bibitem[ABG06]{Angel06}
Eric Angel, Evripidis Bampis, and Laurent Gourvès.
\newblock Approximation algorithms for the bi-criteria weighted {MAX-CUT}
  problem.
\newblock {\em Discrete Applied Mathematics}, 154(12):1685 -- 1692, 2006.

\bibitem[AGK{\etalchar{+}}11]{sat-fpt3}
Noga Alon, Gregory Gutin, Eun~Jung Kim, Stefan Szeider, and Anders Yeo.
\newblock Solving {MAX}-{\it r}-{SAT} above a tight lower bound.
\newblock {\em Algorithmica}, 61(3):638--655, 2011.

\bibitem[ALM{\etalchar{+}}98]{AroraLMSS98}
Sanjeev Arora, Carsten Lund, Rajeev Motwani, Madhu Sudan, and Mario Szegedy.
\newblock Proof verification and the hardness of approximation problems.
\newblock {\em Journal of the ACM}, 45(3):501--555, 1998.

\bibitem[AS98]{AroraS98}
Sanjeev Arora and Shmuel Safra.
\newblock Probabilistic checking of proofs: A new characterization of {NP}.
\newblock {\em Journal of the ACM}, 45(1):70--122, 1998.

\bibitem[BRS11]{Barak-Rag-Ste}
Boaz Barak, Prasad Raghavendra, and David Steurer.
\newblock Rounding semidefinite programming hierarchies via global correlation.
\newblock In {\em FOCS}, pages 472--481, 2011.

\bibitem[BS04]{BollobasS04}
B.~Bollob\'{a}s and A.~D. Scott.
\newblock Judicious partitions of bounded-degree graphs.
\newblock {\em Journal of Graph Theory}, 46(2):131--143, 2004.

\bibitem[Cha13]{Chan13}
Siu~On Chan.
\newblock Approximation resistance from pairwise independent subgroups.
\newblock In {\em Proceedings of the Forty-fifth Annual ACM Symposium on Theory
  of Computing}, STOC '13, pages 447--456. ACM, 2013.

\bibitem[CMM06]{CharikarMM06}
Moses Charikar, Konstantin Makarychev, and Yury Makarychev.
\newblock Note on {MAX-2SAT}.
\newblock {\em Electronic Colloquium on Computational Complexity (ECCC)},
  13(064), 2006.

\bibitem[Dia11]{Diakonikolas11}
Ilias Diakonikolas.
\newblock {\em Approximation of Multiobjective Optimization Problems}.
\newblock PhD thesis, Columbia University, 2011.

\bibitem[DRS02]{DinurRS02}
Irit Dinur, Oded Regev, and Clifford~D. Smyth.
\newblock The hardness of 3 - uniform hypergraph coloring.
\newblock In {\em Proceedings of the 43rd Symposium on Foundations of Computer
  Science}, FOCS '02, pages 33--, Washington, DC, USA, 2002. IEEE Computer
  Society.

\bibitem[GRW11]{Glasser11}
Christian Gla{\ss}er, Christian Reitwie{\ss}ner, and Maximilian Witek.
\newblock Applications of discrepancy theory in multiobjective approximation.
\newblock In {\em FSTTCS'11}, pages 55--65, 2011.

\bibitem[GS11]{Venkat-Sinop}
Venkatesan Guruswami and Ali~Kemal Sinop.
\newblock Lasserre hierarchy, higher eigenvalues, and approximation schemes for
  graph partitioning and quadratic integer programming with psd objectives.
\newblock In {\em FOCS}, pages 482--491, 2011.

\bibitem[GW93]{GW-maxsat}
Michel~X. Goemans and David~P. Williamson.
\newblock A new $\frac{3}{4}$-approximation algorithm for {MAX SAT}.
\newblock In {\em IPCO}, pages 313--321, 1993.

\bibitem[GW95]{GW-sdp95}
Michel~X. Goemans and David~P. Williamson.
\newblock Improved approximation algorithms for maximum cut and satisfiability
  problems using semidefinite programming.
\newblock {\em J. ACM}, 42(6):1115--1145, November 1995.

\bibitem[H{\aa}s01]{Hastad01}
Johan H{\aa}stad.
\newblock Some optimal inapproximability results.
\newblock {\em J. ACM}, 48(4):798--859, 2001.

\bibitem[IP01]{ImpagliazzoP01}
Russell Impagliazzo and Ramamohan Paturi.
\newblock On the complexity of k-{SAT}.
\newblock {\em Journal of Computer and System Sciences}, 62(2):367 -- 375,
  2001.

\bibitem[IPZ01]{ImpagliazzoPZ01}
Russell Impagliazzo, Ramamohan Paturi, and Francis Zane.
\newblock Which problems have strongly exponential complexity?
\newblock {\em J. Comput. Syst. Sci.}, 63(4):512--530, December 2001.

\bibitem[Kho02]{Khot02}
S.~Khot.
\newblock On the power of unique 2-prover 1-round games.
\newblock pages 767--775, 2002.

\bibitem[KO07]{KuhnO07}
Daniela K\"{u}hn and Deryk Osthus.
\newblock Maximizing several cuts simultaneously.
\newblock {\em Comb. Probab. Comput.}, 16(2):277--283, March 2007.

\bibitem[KSTW01]{KhannaSTW01}
Sanjeev Khanna, Madhu Sudan, Luca Trevisan, and David~P. Williamson.
\newblock The approximability of constraint satisfaction problems.
\newblock {\em SIAM J. Comput.}, 30(6):1863--1920, December 2001.

\bibitem[Mar13]{Marx13}
D\'{a}niel Marx.
\newblock Slides : {CSP}s and fixed{-}parameter tractability.
\newblock \url{http://www.cs.bme.hu/~dmarx/papers/marx-bergen-2013-csp.pdf},
  2013.

\bibitem[MM12]{MakarychevM12}
Konstantin Makarychev and Yury Makarychev.
\newblock Approximation algorithm for non-boolean max k-csp.
\newblock In Anupam Gupta, Klaus Jansen, Jos√© Rolim, and Rocco Servedio,
  editors, {\em Approximation, Randomization, and Combinatorial Optimization.
  Algorithms and Techniques}, volume 7408 of {\em Lecture Notes in Computer
  Science}, pages 254--265. Springer Berlin Heidelberg, 2012.

\bibitem[MR99]{sat-fpt1}
Meena Mahajan and Venkatesh Raman.
\newblock Parameterizing above guaranteed values: Maxsat and maxcut.
\newblock {\em J. Algorithms}, 31(2):335--354, 1999.

\bibitem[MRS09]{sat-fpt2}
Meena Mahajan, Venkatesh Raman, and Somnath Sikdar.
\newblock Parameterizing above or below guaranteed values.
\newblock {\em J. Comput. Syst. Sci.}, 75(2):137--153, 2009.

\bibitem[Pat08]{Patel08}
Viresh Patel.
\newblock Cutting two graphs simultaneously.
\newblock {\em J. Graph Theory}, 57(1):19--32, January 2008.

\bibitem[PY00]{PapadimitriouY00}
Christos~H. Papadimitriou and Mihalis Yannakakis.
\newblock On the approximability of trade-offs and optimal access of web
  sources.
\newblock In {\em Foundations of Computer Science, 2000. Proceedings. 41st
  Annual Symposium on}, pages 86--92, 2000.

\bibitem[Rag08]{Raghavendra08}
Prasad Raghavendra.
\newblock Optimal algorithms and inapproximability results for every {CSP}?
\newblock In {\em Proceedings of the 40th annual ACM symposium on Theory of
  computing}, STOC '08, pages 245--254, New York, NY, USA, 2008. ACM.

\bibitem[RS04]{RautenbackS04}
Dieter Rautenbach and Zolt{\'a}n Szigeti.
\newblock {\em Simultaneous large cuts}.
\newblock Forschungsinstitut f{\"u}r Diskrete Mathematik, Rheinische
  Friedrich-Wilhelms-Universit{\"a}t, 2004.

\bibitem[RS09]{RaghavendraS09}
Prasad Raghavendra and David Steurer.
\newblock How to round any {CSP}.
\newblock In {\em In Proc. 50th IEEE Symp. on Foundations of Comp. Sci}, 2009.

\bibitem[RT12]{Raghavendra-Tan}
Prasad Raghavendra and Ning Tan.
\newblock Approximating csps with global cardinality constraints using sdp
  hierarchies.
\newblock In {\em SODA}, pages 373--387, 2012.

\bibitem[Sch78]{Schaefer78}
Thomas~J. Schaefer.
\newblock The complexity of satisfiability problems.
\newblock In {\em Proceedings of the Tenth Annual ACM Symposium on Theory of
  Computing}, STOC '78, pages 216--226, New York, NY, USA, 1978. ACM.

\bibitem[Tre98]{Trevisan98}
L.~Trevisan.
\newblock Parallel approximation algorithms by positive linear programming.
\newblock {\em Algorithmica}, 21(1):72--88, 1998.

\end{thebibliography}

\newpage

\appendix

\newpage

\renewcommand{\NAE}{\mathsf{NAE}}
\newcommand{\Id}{\mathsf{Id}}
\newcommand{\Neg}{\mathsf{Neg}}
\newcommand{\Eq}{\mathsf{Equality}}

%\Ssnote{To do:
%\begin{itemize}
%\item Define the boolean functions $\NAE, \Id, \Neg, \Eq.$ 
%\item Definition of ETH? What about using a different font for ETH?
%\end{itemize}
%}

\section{Hardness results for large $k$}
\label{section:hardness}

In this section, we prove our hardness results for simultaneous CSPs. Recall the theorem that we are trying to show.

\begin{theorem}[restated]
\label{thm:dichotomy}
Assume the Exponential Time Hypothesis. Let $\calF$ be a fixed finite
set of Boolean predicates. If $\calF$ is not $0$-valid or $1$-valid,
then for $k = \omega(\log n )$, then detecting positivity of $k$-fold
simultaneous MAX-$\calF$-CSPs on $n$ variables requires time
superpolynomial in $n$.
\end{theorem}

The main notion that we will use for our hardness reductions
is the notion of a ``simultaneous-implementation''.

\begin{definition}[Simultaneous-Implementation]
Let $\{x_1, \ldots, x_w\}$ be a collection of variables (called {\bf
primary} variables).  Let $P: \{0,1\}^w \to \{\true, \false\}$ be a
predicate.  Let $\{y_1, \ldots, y_t\}$ be another collection of
variables (called {\bf auxiliary} variables).  

Let $\mathcal C_1, \ldots, \mathcal C_k$ be sets of constraints on
$\{x_1, \ldots, x_w, y_1, \ldots, y_t\}$, where for each $i \in [k]$,
$\calC_i$ consists of various applications of predicates to tuples of
distinct variables from $\{x_1, \ldots, x_w, y_1, \ldots, y_t\}$.  We
say that $\calC_1, \ldots, \calC_k$ {\bf simultaneously-implements}
$P$ if for every assignment to $x_1, \ldots, x_w,$ we have,
\begin{itemize}
\item If $P(x_1, \ldots, x_w) = \true$, then there exists
a setting of the variables $y_1, \ldots, y_t$ such that
each collection $\calC_1, \ldots, \calC_k$ has at least one
satisfied constraint.
\item If $P(x_1, \ldots, x_w) = \false$, then for every
setting of the variables $y_1, \ldots, y_t$, at least
one of the collections $\calC_1, \ldots, \calC_k$ has
no satisfied constraints.
\end{itemize}
\end{definition}

We say that a collection of predicates $\calF$ simultaneously-implements
$P$ if there is a simultaneous-implementation of $P$
where for each collection $\calC_i$ ($i \in [k]$), 
every constraint in $\calC_i$ is an application of some predicate
from $\calF$.

%\Anote{is this a correct place for definitions?}

The utility of simultaneous-implementation lies in the following lemma.
\begin{lemma}
\label{lem:implements-reduction}
  Let $P$ be a predicate.  Suppose checking satisfiability of CSPs on
  $n$ variables with $m$ constraints, where each constraint is an
  application of the predicate $P$, requires time $T(n,m),$ with
  $T(n,m) = \omega(m+n).$ Suppose $\calF$ simultaneously-implements
  $P$.  Then detecting positivity of $O(m)$-fold simultaneous
  MAX-$\calF$-CSP on $O(m+n)$ variables requires time $\Omega(T(n,m))$.
\end{lemma}
\begin{proof}
  Suppose we have a $P$-CSP instance $\Phi$ with $m$ constraints on
  $n$ variables. For each of the constraints $C \in \Phi$, we
  simultaneously-implement $C$ using the original set of
  variables as primary variables, and new auxiliary variables for each
  constraint. Thus, for every $C \in \Phi,$ we obtain $k$
  MAX-$\calF$-CSP instances $\calC^C_1, \ldots, \calC^C_k,$ for some
  constant $k.$ The collection of instances $\{\calC^C_i\}_{C \in \Phi, i
    \in [k]}$ constitute the $O(m)$-simultaneous MAX-$\calF$-CSP instance
 on $O(m+n)$ variables.

 If $\Phi$ is satisfiable, we know that there exists an assignment to
 the original variables such that each $C \in \Phi$ is
 satisfied. Hence, by the simultaneously-implements property, there
 exists as assignment to all the auxiliary variables such that each
 $\calC_i^C$ has at least one satisfied constraint. If $\Phi$ is
 unsatisfiable, for any assignment to the primary variables, at least
 one constraint $C$ must be unsatisfied. Hence, by the
 simultaneously-implements property, for any assignment to the
 auxiliary variables, there is an $i \in [k]$ such that $\calC^C_i$
 has no satisfied constraints. Thus, our simultaneous MAX-$\calF$-CSP
 instance has a non-zero objective value iff $\Phi$ is
 satisfiable. Since this reduction requires only $O(m+n)$ time,
 suppose we require $T^\prime$ time for detecting positivity of a
 $O(m)$-simultaneous MAX-$\calF$-CSP instance on $O(m+n)$ variables,
 we must have $T^\prime + O(m+n) \ge T(m,n),$ giving $T^\prime =
 \Omega(T(m,n))$ since $T(m,n) = \omega(m+n).$
\end{proof}

The simultaneous-implementations we construct will be based on a related
notion of implementation arising in approximation preserving reductions.
We recall this definition below.

\begin{definition}[Implementation]
Let $x_1, \ldots, x_w$ be a collection of variables
(called {\bf primary} variables).
Let $P: \{0,1\}^w \to \{\true, \false\}$ be a predicate.

Let $y_1, \ldots, y_t$ be another collection of
variables (called {\bf auxiliary} variables).
Let $C_1, \ldots, C_d$ be constraints on
$\{x_1, \ldots, x_w, y_1, \ldots, y_t\}$,
where for each $i \in [d]$, the variables feeding
into $C_i$ are all distinct.

We say that $C_1, \ldots, C_d$ {\bf $e$-implements} $P$
if for every assignment to $x_1, \ldots, x_w$ we have,
\begin{itemize}
\item If $P(x_1, \ldots, x_w) = \true$, then there exists
a setting of the variables $y_1, \ldots, y_t$ such that
at least $e$ of the constraints $C_1, \ldots, C_d$ evaluate to $\true$.
\item If $P(x_1, \ldots, x_w) = \false$, then for every
setting of the variables $y_1, \ldots, y_t$, at most 
$e-1$ of the constraints $C_1, \ldots, C_d$ evaluate to $\true$.
\end{itemize}
\end{definition}

We say that a collection of predicates $\calF$ implements
$P$ if there is some $e$ and an $e$-implementation of $C$
where all the constraints $C_1, \ldots, C_d$ come from $\calF$.

We will be using following predicates in our proofs.
\begin{itemize}
\item $\Id, \Neg$ : These are the unary predicates defined as $\Id(x) = x$ and $\Neg(x) = \bar{x}$.
\item $\NAE$: $w$-ary $\NAE$ predicate on variables $x_1, \ldots, x_w$ is defined as $\NAE(x_1,\ldots,x_w) = \false$ iff all the $x_i$'s are equal.
\item $\Eq$ : $\Eq$ is a binary predicate given as $\Eq(x,y) = \true$ iff $x$ equals $y$.
\end{itemize}

We will use the following Lemmas from \cite{KhannaSTW01}.

\begin{lemma}[\cite{KhannaSTW01}]
\label{lemma:xor_implementation}
Let $f$ be a predicate which is not $0$-valid,
and which is closed under complementation. 
Then $\{f \}$ implements $\XOR(x, y)$.
\end{lemma}

\begin{lemma}[\cite{KhannaSTW01}]
\label{lemma:TF_implementation}
Let $f$ be a predicate not closed under complementation, and let $g$
be a predicate that is not $0$-valid. 
Then $\{f,g\}$ implements $\Id$, and $\{f,g\}$ implements $\Neg$.
\end{lemma}

We will now prove lemmas that will capture the property of {\em simultaneous implementation} which will be used in proving Theorem~\ref{thm:dichotomy}.
\begin{lemma}
\label{lemma:XOR_NAE}
If $\{f\}$ simultaneously-implements predicate $\XOR$ on $two$ variables, then $\{f\}$ also simultaneously-implements the predicate $\NAE$ on $three$ variables.
\end{lemma}
\begin{proof}
  Consider an $\NAE$ constraint $\NAE(x,y,z).$ Let $\calA_1, \ldots, \calA_d$ be the simultaneous implementation of
  constraint $\XOR(x, y),$ using predicate $f$ and a set of auxiliary
  variables $y_1, \ldots, y_t$ for some $t$. Similarly, let
  $\calB_1, \ldots, \calB_d$ and $\calC_1, \ldots,
  \calC_d$ be the simultaneous implementation of constraint $\XOR(y,
  z)$ and $\XOR(x, z)$ respectively using $f$ and on a same set of
  auxiliary variables $y_1, \ldots, y_t,$ constructed by replacing
  the variables $(x,y)$ in $\{\calA_1, \ldots, \calA_d\}$ with $(y,z)$
  and $(x,z)$ respectively. 
% Observe 
%   these implementations are  $symmetric$ in the sense that for each $i\in
%   [d],$ an instance $\calB_i$ is same as instance $\calA_i$ with $y$
%   replaced with $z$ and $x$ replaced with $y$, similarly for other
%   instances\Anote{is this clear?}. 
  We construct sets of constraints $\calD_1, \ldots, \calD_d$ as
  follows: for each $i\in [d],$ $\calD_i$ consists
  of all constraints from $\calA_i, \calB_i,$ and $\calC_i.$ We now
  show that $\{\calD_1, \ldots, \calD_d\}$ simultaneously-implement
  $\NAE(x, y, z).$

  First, notice that $\NAE(x, y, z)$ is $\false$ iff all constraints
  $\XOR(x, y),$ $\XOR(y, z)$ and $\XOR(x, z)$ are $\false.$ Consider
  the case when $\NAE(x, y, z)$ is $\false$. Since we are using same
  set of auxiliary variables and the implementation is symmetric, for
  every setting of variables $y_1, \ldots, y_t,$ there exists a fixed
  $i\in [d]$ such that each of $\calA_i, \calB_i$ and $\calC_i$ has no
  satisfied constraints. And hence, instance $\calD_i$ has no
  satisfied constraints. If $\NAE(x,y,z)$ is $\true$ then at least one
  of $\XOR(x, y),$ $\XOR(y, z)$ or $\XOR(x, z)$ must be $\true.$
  Without loss of generality, we assume that $\XOR(x,y)$ is $\true.$
  Thus, there exists a setting of variables $y_1, \ldots, y_t$
  such that each of $\calA_1, \ldots, \calA_d,$ has at least one
  satisfied constraint, and hence each of $\calD_1, \ldots,
  \calD_d$ too has at least one such constraint.
\end{proof}

% \Ssnote{@Amey: Read the above lemma to check for argument / notation /
%   language usage.}

% \Ssnote{@Amey: Style notes: Use ldots instead of cdots when describing
% a sequence such as $\calA_1, \ldots, \calA_d.$ Usually the first and
% the last term are sufficient here. Also have a comma after ldots. Don't leave an empty line
% before end\{proof\}. Use emph if you want to italicize text. }

\begin{lemma}
\label{lemma:equality_implementation}
Let $f$ be a predicate not closed under complementation, not $0$-valid and not $1$-valid. $f$ can simultaneously-implement $\Eq.$
\end{lemma}
\begin{proof}
Consider an equality constraint $\Eq(x,y)$, our aim is to
simultaneously-implement this constraint using predicate $f.$ 

Since $f$ satisfies the properties of
Lemma~\ref{lemma:TF_implementation}, we can implement $\Id(x)$ and
$\Id(y)$ using $f.$ Let $X^T_1, \ldots, X^T_{d_1}$ be an
$e_1$-implementation of $\Id(x)$ using $f$ and some set of auxiliary
variables $A_1$ for some $e_1 < d_1.$ Similarly, let $Y^T_1, \ldots,
Y^T_{d_1}$ be an $e_1$-implementation of $\Id(y)$ using $f$ and a set of
auxiliary variables $A_2$.

We can also implement $\Neg(x)$ and $\Neg(y)$ using $f.$ Let $X^F_1, \ldots,
X^F_{d_2}$ be an $e_2$-implementation of $\Neg(x)$ using $f$ and a set
of auxiliary variables $B_1$ for some $e_2 < d_2.$ Similarly, let
$Y^F_1, \ldots, Y^F_{d_1}$ be an $e_2$-implementation of $\Neg(y)$ using
$f$ and a set of auxiliary variables $B_2.$

We now describe the construction of the simultaneous-implementation.
The implementation uses all auxiliary variables in $A_1, A_2, B_1,$
and $B_2.$ Each instance in the simultaneous-implementation is labeled
by a tuple $(M, N, a, b)$ where $M\subseteq [d_1]$ with $|M| =
d_1-e_1+1,$ $N\subseteq [d_2]$ with $|N|= d_2-e_2+1,$ and $(a,b) \in
\{ (T,F), (F,T) \}.$ An instance corresponding to a tuple $(M, N,
a, b)$ has following set of constraints in $f$:
$$\{ X^a_m, Y^b_m | m\in M, n\in N \} $$
We will now prove the simultaneous-implementation property of the
above created instance. Consider the case when $x=y=\true$ (other case
being similar). We know that in this case, there exists a setting of
auxiliary variables $A_1$ used in the implementation of $\Id(x)$ which
satisfies at least $e_1$ constraints out of $X^T_1, \ldots,
X^T_{d_1}.$ Similarly, there exists a setting of auxiliary variables
$A_2$ used in the implementation of $\Id(y)$ which satisfies at least
$e_1$ constraints out of $Y^T_1, \ldots, Y^T_{d_1}.$ Fix this setting
of auxiliary variables in $A_1, A_2,$ and any arbitrary setting for
auxiliary variables in $B_1$ and $B_2$. Thus, the instance labeled by
tuple the $(M, N, a, b)$ either contains $d_1-e_1+1$ constraints from
$X^T_1, \ldots, X^T_{d_1}$ if $a=T,$ or else, it contains $d_1-e_1+1$
constraints from $Y^T_1, \ldots, Y^T_{d_1}$. In any case, the property
of $e_1$-implementation implies that at least one constraint is
satisfied for this instance.

Now we need to show that if $x\neq y,$ then for any setting of
auxiliary variables, there exists an instance which has no satisfied
constraints. Consider the case when $x=\true$ and $y=\false$ (other
case being similar). Consider any fixed assignment to the auxiliary
variables in $A_1, A_2, B_1,$ and $B_2.$ We know that for this fixed
assignment to the auxiliary variables in $B_1,$ there exists a subset
$N\subseteq [d_2]$ of size at least $d_2-e_2+1,$ such that all
constraints in $\{ X^F_j | j\in N\}$ are unsatisfied. Similarly, for
this fixed assignment to variables in $A_2,$ there exists a subset
$M\subseteq [d_1]$ of size at least $d_1-e_1+1$ such that all
constraints in $\{ Y^T_i | i\in M\}$ are unsatisfied. Thus, the
instance corresponding to tuple $(M, N, F, T)$ has no satisfied
constraints.
\end{proof}
% \Ssnote{@Amey: Note the usage of 'assignment to' instead of 'setting
%   of', and 'such that' instead of just 'such'. A sentence should not
%   begin with 'And'}

We now prove Theorem~\ref{thm:dichotomy}.
\begin{proof}
We take cases on whether $\calF$ contains some $f$ which is closed under complementation.
%\begin{itemize}

\medskip
\noindent
{\bf Case 1:} Suppose there exists some $f \in \calF$ which is
  closed under complementation. In this case, it is enough to show
  that $f$ simultaneously-implements $\XOR$. To see this, assume that we
  can simultaneously-implement $\XOR$ using $f$. Hence, by
  Lemma~\ref{lemma:XOR_NAE}, we can simultaneously-implement the
  predicate $\NAE$ on $three$ variables using $f$. We start with an
  NAE-3-SAT instance $\phi,$ on $n$ variables with $m$ constraints.
  For each constraint $C\in \phi,$ we create a set of $O(1)$ many
  instances which simultaneously-implement $C.$ The final simultaneous
  instance is the collection of all instances that we get with each
  simultaneous-implementation of constraints in $\phi.$

  In the completeness case, when $\phi$ is satisfiable, then by the
  property of simultaneous-implementation, we have that there exists a
  setting of auxiliary variables, from each implementation of $\NAE$
  constraints, such that each instance has at least one constraint
  satisfied. And hence, the value of the final simultaneous instance
  is non $zero$.

  In the soundness case, for any assignment to the variables $x_1,
  \ldots, x_n$ there exists a constraint (say $C$) which is not
  satisfied. Hence one of the instance from the simultaneous
  implementation of this constraint has value $zero$ no matter how we
  set the auxiliary variables. And hence, the whole simultaneous
  instance has value zero in this case.

  To prove the theorem in this case, it remains to show that we can
  simultaneously-implement $\XOR(x,y)$ using $f.$ Since $f$ is closed
  under complementation, we can $e$-implement $\XOR$ using $f$ (for
  some $e$) by Lemma~\ref{lemma:xor_implementation}. Let $C_1, \ldots,
  C_d$ be the set of $f$-constraints that we get from this
  $e$-implementation, $e<d.$ The collection of instances contains one
  instance for every subset $J\subseteq [d]$ of size $d-e+1.$ The
  instance labeled by $J \subseteq [d]$ contains all constraints from
  the set $\{C_j | j\in J\}.$ Hence, there ${d \choose e-1}$ instances
  in the collection. Note that we used the same set of auxiliary
  variables in this simultaneous-implementation. We now show that this
  collection of instances simultaneously-implements $\XOR(x,y)$.  To see
  this, consider the case when $\XOR(x,y)$ is $\true$. Thus.  there is
  an assignment to the auxiliary variables that satisfies at least $e$
  constraints out of $C_1, \ldots, C_d.$ Hence, for this particular
  assignment, the instance labeled by $J,$ where $J \subset [d]$ is
  any subset of size $d-e+1$, has at least one satisfied
  constraint. When $\XOR(x,y)$ is $\false,$ then for any assignment to
  the auxiliary variables, there is some $J \subseteq [d]$ of size
  $d-e+1$ such that no constraints in the set $\{C_j | j\in J\}$ are
  satisfied. Hence, for this assignment, the instance labeled with $J$
  has no satisfied constraints. This shows that $f$
  simultaneously-implements predicate $\XOR$ on two variables.

  Combining the two arguments above, we get that $\{f\}$
  simultaneously-implements 3-$\NAE.$ Since 3-$\NAE$ has
a linear time gadget reduction from 3-{\sc SAT}~\cite{Schaefer78},
  and the ETH implies that 3-{\sc SAT} on $s$ variables 
and $O(s)$ clauses requires time $2^{\Omega(s)}$~\cite{ImpagliazzoP01, ImpagliazzoPZ01}, we get that
checking satisfiability of a 3-$\NAE$
  instance with $\omega(\log n)$ constraints on $\omega(\log
  n)$ variables requires time super-polynomial in
  $n$. Thus, using
  Lemma~\ref{lem:implements-reduction} implies that detecting
  positivity of an $\omega(\log n)$-simultaneous MAX-$f$-CSP
  requires time superpolynomial in $n.$

% \Ssnote{@Swastik: Please read the last para. I am not so sure this is
%   the most technically precise way of writing this.}

\medskip
\noindent
 {\bf Case 2:} Suppose that for all $f \in \calF$, $f$ is not
  closed under complementation. Let $f \in \calF$ be any predicate of
  arity $r.$ Since, $f$ is not closed under complementation, there
  exist $\alpha, \beta \in \{0,1\}^r$ that satisfy $\alpha_i \oplus
  \beta_i = 1$ for all $i \in [r],$ and $f(\alpha) = 0,$ $f(\beta) =
  1.$ We can reduce a 3-\sat instance with $n$ variables and $m =
  \poly(n)$ clauses to $m$ simultaneous instances over $n$ variables
  involving the predicate $f.$ For every clause $C$ of the form $x
  \vee y \vee z,$ we create an instance with $3$ equal weight
  constraints $\{f(\alpha \oplus (x,\dots,x)), f(\alpha \oplus
  (y,\dots,y)),f(\alpha \oplus (z,\dots,z))\},$ where $\oplus$ denotes
  bitwise-xor, or equivalently, we negate the variable in the $i$-th
  position iff $\alpha_i = 1.$

  It is straightforward to see that the original 3-\sat formula is
  satisfiable if and only if there is an assignment to the variables
  that simultaneously satisfies a non zero fraction of the constraints
  in each of the instances.

  In the above reduction, we must be able to apply the predicate to
  several copies of the same variable. In order to remove this
  restriction, we replace each instance with a collection $\calC$ of
  instances as follows: Consider an instance $\{f(\alpha \oplus
  (x,\dots,x)), f(\alpha \oplus (y,\dots,y)),f(\alpha \oplus
  (z,\dots,z))\}.$ We add to our collection $\calC,$ an instance
  $\{f(\alpha \oplus (a_1,\dots,a_r)), f(\alpha \oplus
  (b_1,\dots,b_r)),f(\alpha \oplus (c_1,\dots,c_r))\},$ where $a_i,
  b_i$ and $c_i$ for all $i\in [r],$ are the fresh set of
  variables. Using Lemma~\ref{lemma:equality_implementation}, we can
  simultaneously-implement each constraint of the form $x=a_i$,
  $y=b_i$ and $z=c_i$ using $f$. We add all the instances obtained
  from the simultaneous-implementations to the collection $\calC.$
  Notice that, we have replaced each original instance with only
  $O(1)$ many instances. Hence, we have $O(m)$ many instances in our
  final construction. Thus, as in the first case, assuming ETH we
  deduce that detecting positivity of an $\omega(\log
  n)$-simultaneous MAX-$f$-CSP requires time super-polynomial in $n$.
% \end{itemize}
\end{proof}

\subsection{Hardness for Simultaneous \maxwsat}
\begin{proposition}[Proposition~\ref{proposition:results:hardness:maxwsat} restated]
  For all integers $w \ge 4$ and $\epsilon > 0$, given $k \ge 2^{w-3}$
  simultaneous instances of {\sc Max-E$w$-SAT} that are simultaneously
  satisfiable, it is \np-hard to find a $(\nfrac{7}{8}+\eps)$-minimium
  approximation.
\end{proposition}
% \begin{proposition}
% For all $w \in \mathbb{N}$ and $\epsilon > 0$, the minimization version of $2^w$ simultaneous instances of MAX-E(w+3)-SAT is NP-hard to approximate better than $\frac{7}{8}+\epsilon.$
% \end{proposition}
\begin{proof}
  We know that given a satisfiable {\sc Max-E3-SAT} instance, it is
  \np-hard to find an assignment that satisfies a
  $(\nfrac{7}{8}+\eps)$ fraction of the
  constraints~\cite{Hastad01}. We reduce a single {\sc Max-E3-SAT}
  instance to the given problem as follows : Let $\Phi$ be an instance
  of {\sc Max-E3-SAT} with clauses $\{C_i\}_{i=1}^m$ on variable set
  $\{x_1,\ldots, x_n\}$. Given $w \ge 4,$ let $\{z_1,\ldots,
  z_{w-3}\}$ be a fresh set of variables. For every, $a\in
  \{0,1\}^{w-3}$, we construct a {\sc Max-E$w$-SAT} instance with
  clauses $\{C_i \vee \vee_{j=1}^w (z_j\oplus a_j)\}_{i=1}^m,$ where
  $z_j\oplus 0 = z_j$ and $z_j\oplus 1 = \bar{z_j}.$ It is
  straightforward to see that for any assignment, its value on $\Phi$
  is the same as the minimum of its value on the {\sc Max-E$w$-SAT}
  instances, immediately implying the result. 
% \$\frac{7}{8}+\epsilon$ approximation to
  % this simultaneous problem gives $\frac{7}{8}+\epsilon$ approximation
  % to $\phi$. The proposition follows since MAX-E3-SAT is NP-hard to
  % approximate better than $\frac{7}{8}+\epsilon$ for all $\epsilon>0.$
  % ~\ref{}
\end{proof}

\section{Algorithm for Unweighted \maxcut}
\label{section:maxcut-unweighted}
For simultaneous unweighted \maxcut instances, we can use the
Goemans-Williamson SDP to obtain a slightly better approximation. The
algorithm, {\sc UnweightedMC}, is described in
Figure~\ref{fig:maxcut-unweighted}.

Let $V$ be the set of vertices. Our input consists of an integer $k \ge 1,$ and $k$ unweighted
instances of \maxcut, specified by indicator functions
$\calW_1,\ldots,\calW_k$ of edge set. Let $m_\ell$ denotes the number
of edges in graph $\ell\in[k]$. We consider these graphs as weighted
graphs with all non-zero edge weights as $\frac{1}{m_\ell}$ so that
the total weight of edges of in a graph is $1$. For a given subset $S$ of vertices, we say an edge is {\em active} if at least one of its endpoints is in $V\setminus S.$
\begin{figure*}[h]
\begin{tabularx}{\textwidth}{|X|}
\hline
\vspace{0mm}
{\bf Input}: $k$ unweighted instances of \maxcut $\calW_1,\ldots,\calW_k$ on the vertex set $V$. \\
{\bf Output}: A cut of $V.$
\begin{enumerate}[label=\arabic*.]

\item Set $\epsilon \defeq \frac{1}{1600\cdot c_0^2k^2}, t=
  \frac{100k}{\eps^2},  S=\emptyset, D=\emptyset$ ($c_0$ is the
  constant from Lemma~\ref{lem:cmm_maxcsp}).
\item If every graph has more than $t$ edges, then go to Step~\ref{item:alg:umc:partial}.
\item \label{item:alg:umc:loop}Repeat until there is no $\ell \in
  [k]\setminus D$ such that the instance $\calW_\ell$ has less than
  $t^{3^{|D|}}$ active edges given $S$.
\begin{enumerate}
\item Let $\ell \in [k]$ be an instance with the least number of edges active edges given $S$.
\item Add all the endpoints of the edge set of instance $\calW_\ell$ into set $S$.
\item $D \leftarrow D\cup \ell$
\end{enumerate}
%\item For each $\ell \in [k],$ assign each edge in  $\calW_\ell$ a weight
% of $\frac{1}{m_\ell},$ where $m_\ell$ is the number of edges in $\calW_\ell.$
\item \label{item:alg:umc:partial} For each partial assignment $h: S \to \{0,1\}$ (If $S=\emptyset$ then do the following steps without considering partial assignment $h$)
\begin{enumerate}
\item \label{item:alg:umc:partial:sdp}Run the SDP algorithm for instances in $[k]$ given by Lemma ~\ref{lem:cmm_maxcsp} with $h$ as a partial assignment. Let $h_1$ be the assignment returned by the algorithm. (Note $h_1|_S = h$)
\item \label{item:alg:umc:partial:urandom}Define $g: V \setminus S \to \{0,1\}$ by independently sampling
$g(v) \in \{0,1\}$ with $\E[g(v)] = \nicefrac{1}{2},$ for each $v \in V\setminus S.$ In this case the cut is given by an assignment $h\cup g$.
\item Let $\out_{h}$ be the better of the two solutions ($h_1$ and $h\cup g$).
\end{enumerate}
\item Output the largest $\out_{h}$ seen.
\vspace{-3mm}
\end{enumerate}
\\ 
\hline
\end{tabularx}
\caption{Algorithm {\algouwmaxcut} for approximating unweighted
  simultaneous \maxcut}
  \label{fig:maxcut-unweighted}
\end{figure*}

\subsection{Analysis of \algouwmaxcut}
For analysing the algorithm {\algouwmaxcut}, we need the following
lemma that is proven by combining SDP rounding for 2-SAT
from~\cite{CharikarMM06} with a Markov argument. A proof is
included in Section~\ref{section:sdp} for completeness.
\begin{lemma}
\label{lem:cmm_maxcsp}
For $k$ simultaneous instances of any {\sc MAX-2-CSP} such that there exists
an assignment which satisfies a $1-\eps$ weight of the constraints in
each of the instances, there is an efficient algorithm that, for $n$
large enough, given an optimal partial assignment $h$ to a subset of
variables, returns a full assignment which is consistent with $h$ and
simultaneously satisfies at least $1-c_0k\sqrt{\eps}$ (for an absoute
constant $c_0$) fraction of the constraints in each instance with
probability 0.9.
\end{lemma}

Let $S^\star, D^\star$ denote the set $S$ and $D$ that we get at the end of step 3 of the algorithm  {\algouwmaxcut}. Let $f^\star : V \rightarrow \B$ denote the optimal assignment and let $h^\star = f^\star|_S{^\star}$.
\begin{theorem}
  For large enough $n$, given $k$ simultaneous unweighted \maxcut
instances on $n$ vertices, the algorithm {\algouwmaxcut} returns computes a $\left(\frac{1}{2} +
  \Omega\left(\frac{1}{k^2}\right)\right)$-minimum approximate
  solution with probability at least $0.9.$ The running time is
$2^{2^{2^{O(k)}}}
\cdot \poly(n).$ 
\end{theorem}
\begin{proof}
  We will analyze the approximation guarantee of the algorithm when
  the optimal partial assignment $h^\star$ to the variables $S^\star$
  is picked for $h$ in Step~\ref{item:alg:umc:partial} of the
  algorithm. Note that Step~\ref{item:alg:umc:partial:sdp} and
  \ref{item:alg:umc:partial:urandom} maintain the assignment to the
  set $S^\star$ given in Step~\ref{item:alg:umc:partial} Hence, for
  all instances $\ell \in D^\star$, we essentially get the optimal cut
  value $\val(f^\star, \calW_\ell)$. We will analyze the effect of
  rounding done in Step~\ref{item:alg:umc:partial:sdp} and
  ~\ref{item:alg:umc:partial:urandom} on instances in $[k]\setminus
  D^\star$ for a partial assignment $h^\star$ to $S^\star$. Since we
  are taking the best of the two roundings, it is enough to show the
  claimed guarantee for at least one of these two steps.

Let $\opt$ be the value of optimal solution for a given set of instances $[k]$. We consider two cases depending on the value of this optimal solution. 
%\begin{enumerate}

\medskip
{\bf 1.} $\opt \ge (1-\epsilon)$: In this case, we show that the cut
  returned in Step~\ref{item:alg:umc:partial:sdp} is good with high
  probability. 

  Since the $\opt$ is at least $(1-\epsilon)$, and $h^\star$ is an
  optimal partial assignment, we can apply Lemma~\ref{lem:cmm_maxcsp}
  such that with probability at least $0.9$ we get a cut of value at
  least $(1-10c_0 k\cdot\sqrt{\epsilon})$ for all graphs $\ell \in
  [k]\setminus D^\star$, for some constant $c_0$. In this case, the
  approximation guarantee is at least :
$$(1-10c_0k\cdot \sqrt{\epsilon}) \geq \frac{3}{4}.$$

\medskip
{\bf 2.} $\opt< (1-\epsilon)$: In this case, we show that the cut
  returned in Step~\ref{item:alg:umc:partial:urandom} gives the
  claimed approximation guarantee with high probability.

Fix a graph $\ell \in [k]\setminus D^\star$, if any. Let $m_\ell$ be the number of edges in this graph. We know that $m_\ell \ge t^{3^{|D^\star|}}$ and also $|S^\star|\leq 4t^{3^{|D^\star| - 1}}$.
Let $Y_\ell$ be a random variable defined as
\[Y_\ell \defeq \val(h^\star \cup g,\calW_\ell),\] that specifies the
fraction of {\em total} edges that are cut by assignment $h^\star \cup
g$ where $g$ is a random partition $g$ of a vertex set $V\setminus
S^\star$. The number of edges of graph $\ell$ that are not active
given $S^\star$ is at most $\nfrac{1}{2}\cdot |S^\star|^2$. If
$|D^\star| = 0,$ we know that all the edges in graph $\ell$ are
active. Otherwise, using the bounds on $m_\ell$ and $|S^\star|,$ we
get that at least a $(1-\nfrac{1}{t})$ fraction of the total edges are
active given $S^\star$. This implies that for uniformly random
partition $g$,
$$\E_g[Y_\ell] \ge \frac{1}{2}  \sum_{\substack{C \in \calC \\  C \in \Act(S^\star)}} \calW_\ell(C) \ge \nfrac{1}{2}(1-\nfrac{1}{t}).$$
We now analyze the variance of a random variable $Y_\ell$ under uniformly random assignment $g : V\setminus S^\star \rightarrow \B$.
\begin{align*}
\Var_g[Y_\ell]&= \sum_{C_{1}, C_{2}\in \Act(S^\star)} \calW(C_1) \calW(C_2) \cdot ( \E[C_1(h^\star\cup g)C_2(h^\star\cup g)] - \E[C_1(h^\star\cup g)]\E[C_2(h^\star\cup g)] ).
\end{align*}
The term in the above summation is $zero$ unless we have either $C_1 =
C_2$ (in which case we know $\E[C_1(h^\star\cup g)C_2(h^\star\cup g)]
- \E[C_1(h^\star\cup g)]\E[C_2(h^\star\cup g)] = \nfrac{1}{4}$) or
when the edges $C_1$ and $C_2$ have a common endpoint in $V\setminus
S^\star$ and the other endpoint in $S^\star$ (in this case
$\E[C_1(h^\star\cup g)C_2(h^\star\cup g)] - \E[C_1(h^\star\cup
g)]\E[C_2(h^\star\cup g)] \leq \nfrac{1}{4}$). 
For $v\in V\setminus S^\star$, let $\kappa_v$ be the set of edges whose one endpoint is $v$ and other endpoint in $S^\star$. Thus,

%Define a {\em cluster}
%as a maximal set of edges which share a common endpoint in $V\setminus
%S^\star$ and whose other endpoint is in $S^\star$. Let $\xi$ be a set
%of all {\em clusters} in graph $\ell$ given $S^\star$.  Thus, we have
%\Anote{This is a very bad notation, or is it?}
\begin{align*}
\Var_g[Y_\ell]&\leq \frac{1}{4}\sum_{C \in \Act(S^\star)} \calW(C)^2  + \frac{1}{4}\sum_{v\in V\setminus S^\star}\sum_{C_{1}, C_{2} \in \kappa_v} \calW(C_1) \calW(C_2)\\
& = \frac{1}{4m_\ell} + \frac{1}{4}\frac{1}{m_\ell^2}\sum_{v\in V\setminus S^\star} |\kappa_v|^2 \\
&\leq \frac{1}{4m_\ell} + \max_{v\in V\setminus S^\star}{|\kappa_v|} \cdot \frac{1}{4}\frac{1}{m_\ell^2}\sum_{v\in V\setminus S^\star} |\kappa_v|\\
&\leq \frac{1}{4m_\ell} + |S^\star| \cdot \frac{1}{4}\frac{1}{m_\ell^2}\cdot m_\ell\\
&\leq \frac{1}{4m_\ell} + \frac{1}{4} \frac{|S^\star|}{m_\ell}\\
&\leq \frac{1}{4t^{3^{|D^\star|}}} + \frac{1}{4}\frac{|S^\star|}{t^{3^{|D^\star|}}} \leq \frac{1}{2t}.
\end{align*}
Hence, by Chebyshev's Inequality, we have
\[ \Pr\left[Y_\ell < \frac{1}{2}\cdot(1-\epsilon_0 -
  \nfrac{1}{t})\right] \leq
\frac{4\var_g[Y_\ell]}{\epsilon_0^2}
\leq\frac{4 \cdot \nfrac{1}{2t}}{\epsilon_0^2} \leq \frac{2}{\eps_0^2t}.\]

By a union bound, with probability at least
$1-\frac{2k}{\epsilon_0^2\cdot t},$ we get a simultaneous cut of value
at least $\frac{1}{2}\cdot(1-\epsilon_0 - \nfrac{1}{t})$ for all $\ell
\in [k]\setminus D^\star.$ If we take $\epsilon_0 =
\frac{\sqrt{20k}}{\sqrt{t}}$, then with probability at least $0.9$ we
get a cut of value at least $\frac{1}{2}\cdot(1-\epsilon_0 -
\nfrac{1}{t})$ for all $\ell\in [k]\setminus D^\star$. In this case,
the approximation guarantee is at least 
$$  \frac{\frac{1}{2}\cdot(1-\epsilon_0 - \nfrac{1}{t})}{\left( 1 - \frac{1}{(40c_0k)^2}\right)} = \left( \frac{1}{2} +\Omega\left(\frac{1}{k^2}\right)\right).$$
\end{proof}

%%% Local Variables: 
%%% mode: latex
%%% TeX-master: "new-simopt"
%%% End: 

\section{Semidefinite Programs for Simultaneous Instances}
\label{section:sdp}
In this section, we study Semidefinite Programming (SDP) relaxations for simultaneous MAX-2-CSP instances.

\subsection{Integrality gaps for Simultaneous \maxcut SDP}
In this section, we show the integrality gaps associated with the natural SDP of {\em minimum approximation} problem for $k$-fold simultaneous \maxcut. 

Suppose we have $k$ simultaneous \maxcut instances on the set of vertices $V=\{x_1,\ldots,x_n\},$ specified by the associated weight functions $\calW_1,\ldots ,\calW_k.$ As before, let $\calC$ denotes the set of all possible edges on $V.$ We assume that for each $\ell \in [k],$ $\sum_{C \in \calC} \calW_{\ell}(C) = 1.$ Following Goemans and Williamson~\cite{GW-sdp95}, the semi-definite programming relaxation for such an instance is described in Figure~\ref{fig:maxcut-sdp}.
\begin{figure*}[h]
\begin{tabularx}{\textwidth}{|X|}
\hline
\vspace{-10mm}
\begin{center}
\[\begin{array}{rrllr}
& \textrm{maximize } \quad  t & & &\\
\textrm{ s.t. } &   \displaystyle \sum_{\substack{C \in \calC \\ C = (x_i, x_j)  } } \frac{1}{2}\cdot \calW_{\ell}(C) \cdot \left( 1 -  \pair{v_i,v_j}\right) & \ge & t & \forall \ell \in [k] \\
&  \norm{v_i}^2 & =& 1& \textrm{ for } i=1,\ldots,n  \\
\end{array}\]
\vspace{-24pt}
\end{center}
\\
\hline
\end{tabularx}
\caption{Semidefinite Program (SDP) for minimum approximation Simultaneous \maxcut}
  \label{fig:maxcut-sdp}
\end{figure*}
We now prove the following claims about integrality gap for the above SDP.
\begin{claim}
\label{claim:maxcut-gap1}
For
weighted instances, the SDP for minimum approximation of simultaneous \maxcut does not have {\em any} constant
integrality gap.
\end{claim}
\begin{proof}
Consider $3$ simultaneous instances such that all but a
tiny fraction of the weight of instance $i$ is on edge $i$ of a
3-cycle. Clearly, no cut can simultaneously cut all the three edges in
the three cycle, and hence the optimum is tiny. However, for the
simultaneous SDP, a vector solution that assigns to the three vertices
of the cycle three vectors such that $\pair{v_i,v_j} = -\nfrac{1}{2}$
for $i \neq j$ gives a constant objective value for all three
instances.
\end{proof}

\begin{claim}
\label{claim:maxcut-gap2}
For every fixed $k$, there exists $k$-instances of \maxcut where the SDP relaxation has value
$1- \Omega\left(\frac{1}{k^2}\right)$, while the maximum simultaneous
cut has value only $\frac{1}{2}.$ Moreover, the random hyperplane rounding for a good vector solution for this instance, returns a simultaneous cut of value 0.
\end{claim}
\begin{proof}
Let $k$ be odd. We define $k$ graphs
on $kn$ vertices. Partition the vertex set into $S_0,S_1,...,S_{k-1},$
each of size $n$. Graph $G_i$ has only edges $(x,y)$ such that $x\in
S_i$ ans $y\in S_{(i+1)\mod k},$ each of weight $\nfrac{1}{n^2}$. The
optimal cut must contain exactly half the number of vertices from each
partition, giving a simultaneous cut value of $\nfrac{1}{2}$. Whereas,
the following SDP vectors achieve a simultaneous objective of
$\left(1-O(\frac{1} {k^2} )\right):$ For all vertices in $S_i,$ we
assign the vector $\left(\cos \frac{i}{k}\pi, \sin \frac{i}{k}\pi
\right).$ It is straightforward to see that applying the hyperplane rounding algorithm to this
vector solution gives (with probability 1) a simultaneous cut value of
0.
\end{proof}

\subsection{SDP for Simultaneous \maxcsp}
For \maxcsp, we will be interested in the regime where the optimum assignment satisfies at least a $(1-\eps)$ fraction of the constraints in each of the instances.

Given a MAX-2-CSP instance, we use the standard reduction to transform it into a MAX-2-SAT instance: We reduce each constraint of the 2-CSP instance with a set of at most 4 2-SAT constraints such that for any fixed assignment, the 2-CSP constraint is satisfied iff all the 2-SAT constraints are satisfied, and if the 2-CSP constraint is not satisfied, then at least one of the 2-SAT constraint is not satisfied. \emph{e.g.} We replace $x_1 \wedge x_2$ with $x1 \vee x_2, \widebar{x_1} \vee x_2,$ and $x_1 \vee \widebar{x_2}.$ Similarly, we replace $x_1 \neq x_2$ with $x_1 \vee x_2$ and $\widebar{x_1} \vee \widebar{x_2}.$ We distribute the weight of the 2-CSP constraint equally amongst the 2-SAT constraints.

Given $k$ simultaneous MAX-2-CSP instances, we apply the above reduction to each of the instances to obtain $k$ simultaneous \maxsat instances. The above transformation guarantees the following:
\begin{itemize}
\item {\bf Completeness} If there was an assignment of variables that simultaneously satisfied all the constraints in each of the MAX-2-CSP instances, then the same assignment satisfies all the constraints in each of the \maxsat instances.
\item {\bf Soundness} If no assignment of variables simultaneously satisfied more than $(1-\eps)$ weighted fraction of the constraints in each of the MAX-2-CSP instances, then no assignment simultaneously satisfies more than $(1-\nfrac{\eps}{4})$ weighted fraction of the constraints in each of the MAX-2-SAT instances.
\end{itemize} 

From now on, we will assume that we have $k$ simultaneous \maxsat
instances on the set of variables $\{x_1,\ldots,x_n\},$ specified by
the associated weight functions $\calW_1,\ldots ,\calW_k.$ As before
$\calC$ denotes the set of all possible 2-SAT constraints on $V.$ We
assume that for each $\ell \in [k],$ $\sum_{C \in \calC}
\calW_{\ell}(C) = 1.$ Following Charikar \etal~\cite{CharikarMM06},
the semi-definite programming relaxation for such an instance is
described in Figure~\ref{fig:maxsat-sdp}.

For convenience, we replace each negation $\widebar{x_i}$ with a new variable $x_{-i},$ that is equal to $\widebar{x_1}$ by definition. For each variable $x_i \in V,$ the SDP relaxation will have a vector $v_i.$ We define $v_{-i} = -v_i.$ We will also have a unit vector $v_0$ that is intended to represent the value 1. For a subset $S$ of variables and a partial assignment $h : S\rightarrow \B$, we write the following SDP for the simultaneous \maxsat optimization problem:

\begin{figure*}[h]
\begin{tabularx}{\textwidth}{|X|}
\hline
\vspace{-10mm}
\begin{center}
\[\begin{array}{rrllr}
& \textrm{maximize } \quad  t & & &\\
\textrm{ s.t. } &   \displaystyle \sum_{\substack{C \in \calC \\ C = x_i \vee x_j} } \calW_{\ell}(C) \cdot \left( \norm{v_0}^2 - \frac{1}{4} \pair{v_i - v_0,v_j-v_0}\right) & \ge & t & \forall \ell \in [k] \\
&  \pair{v_i - v_0,v_j-v_0} & \ge & 0 & \forall \textrm{  constraints } x_i \vee x_j \\
&  \norm{v_i}^2 & =& 1& \textrm{ for } i=-n,\ldots,n  \\
&  v_i & = & -v_{-i} & \textrm{ for } i=1,\ldots,n\\
&  v_i & = & v_0 & \forall i\in S \textrm{ s.t. } h(i) = 1 \\
&  v_j & = & -v_0 & \forall j\in S \textrm{ s.t. } h(j) = 0\\
\end{array}\]
\vspace{-5mm}
\end{center}
\\
\hline
\end{tabularx}
\caption{Semidefinite Program (SDP) with a partial assignment $h:S\rightarrow \B$ for Simultaneous \maxsat}
  \label{fig:maxsat-sdp}
\end{figure*}
\begin{align*}
\end{align*}
\vspace{-20pt}

We first observe that for an optimal partial assignment $h$, the optimum of the above SDP is at least the optimum of the simultaneous maximization problem, by picking the solution $v_i = v_0$ if $x_i = \true,$ and $v_i = -v_0$ otherwise. For this vector solution, we have $\nfrac{1}{4} \cdot \left( \norm{v_0}^2 - \pair{v_i - v_0,v_j-v_0}\right) = 1$ if the constraint $x_1 \vee x_2$ is satisfied by the assignment, and 0 otherwise. Since $\sum_{C \in \calC} \calW_{\ell}(C) = 1$ for all $\ell,$ the optimum of the SDP lies between 0 and 1.

Note that the rounding algorithm defined in ~\cite{CharikarMM06} does not depend on the structure of the vectors in the SDP solution. Thus, the following theorem that was proved without a partial assignment in~\cite{CharikarMM06} also applies to above SDP.
\begin{theorem}
\label{thm:cmm}
Given a single \maxsat instance ($k=1$), there is an efficient randomized rounding algorithm such that, if the optimum of the above SDP is $1-\eps,$ for $n$ large enough, it returns an assignment such that the weight of the constraints satisfied is at least  $1-O(\sqrt{\eps})$ in expectation.
\end{theorem}

Now, using Markov's inequality, we can prove the following corollary.
\begin{corollary}
\label{cor:cmm_maxsat}
For $k$ simultaneous instances of \maxsat, there is an efficient randomized rounding algorithm such that if the optimum of the above SDP is $1-\eps,$ for $n$ large enough, it returns an assignment that simultaneously satisfies at least $1-O(k\sqrt{\eps})$ fraction of the constraints in each instance with probability 0.9.
\end{corollary}
\begin{proof}
We use the rounding algorithm given by Theorem~\ref{thm:cmm} to round a solution to the SDP for the $k$ simultaneous instances that achieves an objective value of $1-\eps.$
Observe that this solution is also a solution for the SDP for each of the instances by itself with the same objective value. Thus, by Theorem~\ref{thm:cmm}, for each of the instances, we are guaranteed to find an assignment such that the weight of the constraints satisfied is at least  $1-c_0\eps$ in expectation, for some constant $c > 0.$ Since, for any instance, the maximum weight an assignment can satisfy is at most 1, with probability at least $1-\nfrac{1}{10\cdot k}$ for each instance, we get an assignment such that the weight of the constraints satisfied is at least $1-10ck\cdot \sqrt{\eps}.$ Thus, applying a union bound, with probability at least $1-\nfrac{1}{10},$ we obtain an assignment such that the weight of the satisfied constraints in \emph{all} the $k$ instances is at least $1-10ck\cdot \sqrt{\eps}.$
\end{proof}

Combining the above corollary with the reduction from any MAX-2-CSP to MAX-2-SAT, and the completeness of the SDP, we get a proof of Lemma~\ref{lem:cmm_maxcsp}.

%%% Local Variables: 
%%% mode: latex
%%% TeX-master: "new-simopt"
%%% End: 

\section{Concentration inequalities}
\label{section:concentration}
\begin{lemma}[McDiarmid's Inequality]
\label{lemma:mcdiarmid}
Let $X_1, X_2, \cdots, X_m$ be independent random variables, with $X_i$ taking values in a set $A_i$ for each $i$. Let $\score : \prod A_i \rightarrow \mathbb{R}$ be a function which satisfies:
$$ | \score(x) - \score(x') | \leq \alpha_i $$
whenever the vector $x$ and $x'$ differ only in the $i$-th co-ordinate. Then for any $t>0$
$$\Pr[ |\score(X_1,X_2,\cdots,X_m) - \E[ \score(X_1,X_2,\cdots,X_m) ]| \geq t] \leq 2\exp{\left(\frac{-2t^2}{\sum_i \alpha_i^2} \right)} $$
\end{lemma}

\section{The need for perturbing $\opt$}
\label{section:example}
We construct 2 simultaneous instances of {\sc Max-1-SAT}. Suppose the
algorithm will picks at most $r$ influential variables.  Construct the
two instances on $r+1$ variables, with the weights of the variables
decreasing geometrically, say, with ratio $\nfrac{1}{3}$. The first
instance requires all of them to be \true, where as the second
instance requires all of them to be \false. Under a reasonable
definition of ``influential variables'', the only variable left behind
should the vertex with the least weight.  We consider the Pareto
optimal solution that assigns {\true} to all but the last variable. If
we pick the optimal assignment for the influential variables, and then
randomly assign the rest of the variables, with probability $\nfrac{1}{2}$, we
get zero on the second instance.

%%% Local Variables: 
%%% mode: latex
%%% TeX-master: "new-simopt"
%%% End: 

\end{document}